\documentclass[a4paper,10pt]{article}

\usepackage{geometry}
\usepackage{cite}
\usepackage{graphicx}
\usepackage{amsfonts}
\usepackage{amssymb}
\usepackage{amsmath}
\usepackage{verbatim}
\usepackage{float}
\usepackage{scrextend}
\usepackage{xspace}
\usepackage{algpseudocode}
\usepackage{tipa}
\usepackage{amsthm}

\usepackage[dvipsnames]{xcolor}
\usepackage[utf8]{inputenc}
\usepackage{amsthm}
\usepackage{indentfirst}
\usepackage{graphicx}
\usepackage{algorithm}
\usepackage{algpseudocode}
\usepackage{algorithmicx}
\usepackage{amssymb}
\usepackage{tipa}
\usepackage{stmaryrd}
\usepackage{color}
\usepackage{cellspace}
\usepackage{colortbl}
\usepackage{caption}
\usepackage{enumitem}
\usepackage{amsmath}
\usepackage{balance}
\usepackage{footmisc}
\usepackage{multirow}
\usepackage{tikz}

\newtheorem{definition}{Definition}
\newtheorem{theorem}{Theorem}
\newtheorem{claim}{Claim}
\newtheorem{observation}{Observation}
\newtheorem{corollary}{Corollary}

\newtheorem{lemma}{Lemma}

\newcommand{\Gathering}{$\mathcal{GDG}$}

\newcommand{\TerminationOne}{\textbf{$\textrm{Term}_\textrm{1}$}}
\newcommand{\TerminationTwo}{\textbf{$\textrm{Term}_\textrm{2}$}}

\newcommand{\rFourOne}{\textbf{$\textrm{T}_\textrm{1}$}}
\newcommand{\rFourTwo}{\textbf{$\textrm{T}_\textrm{2}$}}
\newcommand{\rFourThree}{\textbf{$\textrm{T}_\textrm{3}$}}

\newcommand{\rThreeOne}{\textbf{$\textrm{W}_\textrm{1}$}}

\newcommand{\rTwoOne}{\textbf{$\textrm{K}_\textrm{1}$}}
\newcommand{\rTwoTwo}{\textbf{$\textrm{K}_\textrm{2}$}}
\newcommand{\rTwoThree}{\textbf{$\textrm{K}_\textrm{3}$}}
\newcommand{\rTwoFour}{\textbf{$\textrm{K}_\textrm{4}$}}

\newcommand{\rOneOne}{\textbf{$\textrm{M}_\textrm{1}$}}
\newcommand{\rOneTwo}{\textbf{$\textrm{M}_\textrm{2}$}}
\newcommand{\rOneThree}{\textbf{$\textrm{M}_\textrm{3}$}}
\newcommand{\rOneFour}{\textbf{$\textrm{M}_\textrm{4}$}}
\newcommand{\rOneFive}{\textbf{$\textrm{M}_\textrm{5}$}}
\newcommand{\rOneSix}{\textbf{$\textrm{M}_\textrm{6}$}}
\newcommand{\rOneSeven}{\textbf{$\textrm{M}_\textrm{7}$}}
\newcommand{\rOneEight}{\textbf{$\textrm{M}_\textrm{8}$}}
\newcommand{\rOneNine}{\textbf{$\textrm{M}_\textrm{9}$}}
\newcommand{\rOneTen}{\textbf{$\textrm{M}_\textrm{10}$}}
\newcommand{\rOneEleven}{\textbf{$\textrm{M}_\textrm{11}$}}

\newcommand{\righter}{$righter$}
\newcommand{\dumbSearcher}{$dumb\-Sear\-cher$}
\newcommand{\awareSearcher}{$aware\-Sear\-cher$}
\newcommand{\potentialMin}{$po\-ten\-tial\-Min$}
\newcommand{\waitingWalker}{$waiting\-Walker$}
\newcommand{\minWaitingWalker}{$min\-Waiting\-Walker$}
\newcommand{\headWalker}{$head\-Walker$}
\newcommand{\tailWalker}{$tail\-Walker$}
\newcommand{\minTailWalker}{$min\-Tail\-Walker$}
\newcommand{\leftWalker}{$left\-Walker$}
\newcommand{\towerElection}{$tower\-Min$}
\newcommand{\towerTermination}{$tower\-Ter\-mi\-na\-tion$}

\newcommand{\rdv}{{\it RDV}}
\newcommand{\AmITheMin}{{\tt M}}
\newcommand{\MinWaitToBeKnown}{{\tt K}}
\newcommand{\Walk}{{\tt W}}
\newcommand{\WaitTermination}{{\tt T}}

\newcommand{\myparagraph}[1]{{~}\\\noindent {\textbf{{#1}.}}}
\newcommand{\fixme}[1]{\fbox{\textsl{{\bf #1}}}}
\newcommand{\FIXME}[1]{\fixme{#1} \marginpar[\null\hspace{2cm} FIXME]{FIXME}}

\newcommand{\ie}{{\em i.e.}\xspace}

\newcommand{\eg}{{\em e.g.}\xspace}
\newcommand{\etc}{{\em etc.}}

\newcommand{\etal}{{\em et al.}\xspace}

\newcommand{\COT}{\ensuremath{\mathcal{COT}}}
\newcommand{\RE}{\ensuremath{\mathcal{RE}}}
\newcommand{\BRE}{\ensuremath{\mathcal{BRE}}}
\newcommand{\AC}{\ensuremath{\mathcal{AC}}}
\newcommand{\ST}{\ensuremath{\mathcal{ST}}}

\newcommand{\G}{\ensuremath{\mathbb{G}}}
\newcommand{\GE}{\ensuremath{\mathbb{G}_{E}}}
\newcommand{\GW}{\ensuremath{\mathbb{G}_{W}}}
\newcommand{\GEW}{\ensuremath{\mathbb{G}_{EW}}}

\newcommand{\GEBold}{\ensuremath{\mathbf{\mathbb{G}_{E}}}}
\newcommand{\GEWBold}{\ensuremath{\mathbf{\mathbb{G}_{EW}}}}

\title{Gracefully Degrading Gathering in Dynamic Rings}
\author{Marjorie Bournat, Swan Dubois and Franck Petit \\ \footnotesize Sorbonne Université, CNRS, Inria, LIP6, F-75005 Paris, France}
\date{}

\newif\ifJournal
\Journaltrue

\begin{document}
\maketitle

\begin{abstract}
Gracefully degrading algorithms [Biely \etal, TCS 2018] are designed to 
circumvent impossibility results in dynamic systems by adapting themselves to 
the dynamics. Indeed, such an algorithm solves a given problem under some
dynamics and, moreover, guarantees that a weaker (but related) problem is
solved under a higher dynamics under which the original problem is impossible to 
solve. The underlying intuition is to solve the problem whenever possible but to
provide some kind of quality of service if the dynamics become (unpredictably) higher.

In this paper, we apply for the first time this approach to robot networks. We 
focus on the fundamental problem of gathering a squad of autonomous robots on an 
unknown location of a dynamic ring. In this goal, we introduce a set of 
weaker variants of this problem. Motivated by a set of impossibility results
related to the dynamics of the ring, we propose a gracefully degrading gathering
algorithm.
\end{abstract}

%

\section{Introduction}

The classical approach in distributed computing consists in, first, fixing a set of assumptions that captures
the properties of the studied system (atomicity, synchrony, faults, communication modalities, \etc) and, then, focusing 
on the impact of these assumptions in terms of calculability and/or of complexity on a given problem. When coming to dynamic
systems, it is natural to adopt the same approach. Many recent works focus on defining pertinent assumptions for 
capturing the dynamics of those systems \cite{CFQS12,KLO10,XFJ03}. When these assumptions become very weak, that is, when the
system becomes highly dynamic, a somewhat frustrating but not very surprising conclusion emerge: many fundamental distributed
problems are impossible at least, in their classical form \cite{BRSSW18,BDKP16,CFMS15}.

To circumvent such impossibility results, Biely \etal recently introduced the {\em gracefully degrading} approach \cite{BRSSW18}. This approach 
relies on the definition of weaker but related variants of the considered problem. A gracefully degrading algorithm guarantees that it will
solve simultaneously the original problem under some assumption of dynamics and each of its variants under some other (hopefully weaker) assumptions.
As an example, Biely \etal provide a consensus algorithm that gracefully degrades to $k$-set agreement when the dynamics of the system increase. The
underlying idea is to solve the problem in its strongest variant when connectivity conditions are sufficient but also to provide (at the opposite of a 
classical algorithm) some minimal quality of service described by the weaker variants of the problem when those conditions degrade.

Note that, although being applied to dynamic systems by Biely \etal for the first time, this natural idea is not a new
one. Indeed, {\emph{indulgent} algorithms \cite{AGGT12,DG05,G00} provide similar graceful degradation of the problem to satisfy with respect to synchrony (not with
respect to dynamics). {\em Speculation} \cite{DG13,
KADCW09} is a related, but somewhat orthogonal, concept. A speculative algorithm solves the 
problem under some assumptions and moreover provides stronger properties (typically better complexities) whenever conditions are better.

The goal of this paper is to apply graceful degradation to robot networks where a cohort of autonomous robots have to coordinate their actions in order to
solve a global task. We focus on \emph{gathering} in a {\em dynamic ring}. In this problem, starting from any initial position, robots must meet on an 
arbitrary location in a bounded time. Note that we can classically split this specification into a liveness property (all robots terminate 
in bounded time) and a safety property (all robots that terminate do so on the same node). 

\myparagraph{Related works} Several models of dynamic graphs have been defined recently \cite{CFQS12,LVM17,XFJ03}. In this paper, we
adopt the \emph{evolving graph} model \cite{XFJ03} in which a dynamic graph is simply a sequence of static graphs on a
fixed set of nodes: each graph of this sequence contains the edges of the dynamic graph present at a given time. We
also consider the hierarchy of dynamics assumptions introduced by Casteigts \etal \cite{CFQS12}. 
The idea behind this
hierarchy is to gather all dynamic graphs that share some temporal connectivity properties within classes. 
This allows us to compare the strength of these temporal connectivity properties based on the inclusion of classes
between them. 
We are interested in the following classes: 
\COT~(\emph{connected-over-time} graphs) where edges may appear and disappear without any recurrence nor periodicity assumption but
guaranteeing that each node is infinitely often reachable from any other node; 
\RE~(\emph{recurrent-edge} graphs) where any edge that appears at least once does so recurrently; 
\BRE~(\emph{bounded-recurrent-edge} graphs) where any edge that appears at least once does so recurrently in a bounded time;
\AC~(\emph{always-connected} graphs) where the graph is connected at each instant; and
\ST~(\emph{static} graphs) where any edge that appears at least once is always present.
Note that 
$\ST\subset\BRE\subset\RE\subset\COT$ and $\ST\subset\AC\subset\COT$ by definition.

In robot networks, the gathering problem was extensively studied in the context of static graphs, \eg, \cite{FKKSS04,
KMP06,SN17}. The main 
motivation of this vein of research is to characterize the initial positions of the robots allowing gathering in each studied topology in function
of the assumptions on the robots as identifiers, communication, vision range, memory, \etc~
On the other hand, few algorithms have been designed for robots evolving in dynamic graphs. The majority of them deals with the problem of 
exploration \cite{BDD16,BDP17,FMS13,IKW14,LDFS16} (robots must visit each node of the graph at least once or infinitely often depending on the variant of the problem).
In the most related work to ours \cite{LFPPSV17}, Di Luna \etal study the gathering problem in dynamic rings. They first note the impossibility of the problem 
in the \AC~class and consequently propose a weaker variant of the problem, the near-gathering: all robots must gather in finite time on two adjacent nodes. They characterize the impact of 
chirality (ability to agree on a common orientation) and cross-detection (ability to detect whenever a robot cross the same edge in the
opposite direction) on the solvability of the problem. All their algorithms are designed for the \AC~class and are not gracefully degrading.

\begin{table}[t]
\begin{center}
\begin{tabular}{|c|cccc|}
\hline
 & \G & \GE & \GW & \GEW \\
\hline
\COT & \textcolor{red}{Impossible (Cor. \ref{GImpossible} \& \ref{GGWImpossible})} & \textcolor{red}{Impossible (Cor. \ref{GEImpossibleCOT})} & \textcolor{red}{Impossible (Cor. \ref{GGWImpossible})} & \textcolor{OliveGreen}{Possible (Th. \ref{CorollaryC5})} \\
\AC & \textcolor{red}{Impossible (Cor. \ref{GImpossible})} & \textcolor{red}{Impossible (Th. \ref{GEImpossible})} & \textcolor{OliveGreen}{Possible (Th. \ref{theoremC9})} & --- \\
\RE & \textcolor{red}{Impossible (Cor. \ref{GGWImpossible})} & \textcolor{OliveGreen}{Possible (Th. \ref{theoremC6})} & \textcolor{red}{Impossible (Cor. \ref{GGWImpossible})} & --- \\
\BRE & \textcolor{OliveGreen}{Possible (Th. \ref{theoremC7})} & --- & --- & --- \\
\ST & \textcolor{OliveGreen}{Possible (Cor. \ref{theoremStatic})} & --- & --- & --- \\
\hline
\end{tabular}
\end{center}
\caption{\label{tab:summary}Summary of our results. The symbol --- means that a stronger variant of the problem is already proved solvable under the dynamics assumption. Our 
algorithm is gracefully degrading since it solves each variant of the gathering problem as soon as dynamics assumptions allow it.}\vspace{-1cm}
\end{table}

\myparagraph{Contributions} 
By contrast with the work of Di Luna \etal \cite{LFPPSV17}, in this paper we choose to keep 
unchanged the safety of the classical gathering problem (all robots that terminate do so on the same node) and, to circumvent impossibility results, we weaken only
its liveness: at most one robot may not terminate or (not exclusively) all robots that terminate do so eventually (and not in a \emph{bounded} time as in the classical specification). 
This choice is motivated by the approach adopted with indulgent algorithms \cite{AGGT12,G00,DG05}: the safety captures the ``essence'' of the problem and should be 
preserved even in degraded variants of the problem.
Namely, we obtain the four following variants of the gathering problem: 
\G~(\emph{gathering}) all robots terminate on the same node in \emph{bounded} time;
\GE~(\emph{eventual gathering}) all robots terminate on the same node in \emph{finite} time;   
\GW~(\emph{weak gathering}) all robots but (at most) one terminate on the same node in \emph{bounded} time; and
\GEW~(\emph{eventual weak gathering}) all robots but (at most) one terminate on the same node in \emph{finite} time.

We present then a set of impossibility results, summarized in Table \ref{tab:summary}, for these specifications for different classes of dynamic rings. 
Motivated by these impossibility results, our main contribution is a gracefully degrading gathering algorithm. For
each class of dynamic rings we consider, our algorithm solves the strongest possible of our variants of the gathering problem
(refer to Table \ref{tab:summary}). Note that this challenging property is obtained without any knowledge or detection of the dynamics by the robots that always execute the same algorithm.
Our algorithm needs that robots have distincts identifiers, chirality, strong multiplicity detection (\ie ability to count the number of colocated robots), memory (of size sublinear in the size of the ring and identifiers), and communication capacities but deals with (fully) anonymous ring.
Note that these assumptions (whose necessity is left as an open question here) are incomparable with those of Di Luna \etal \cite{LFPPSV17} that assume anonymous but home-based robots (hence, non fully anonymous rings).
This algorithm brings two novelties with respect to the state-of-the-art: $(i)$ it is the first gracefully degrading algorithm dedicated to robot networks; and $(ii)$ it is the 
first algorithm solving (a weak variant of) the gathering problem in the class \COT (the largest class of dynamic graphs that guarantees an exploitable recurrent property).

\ifJournal
\myparagraph{Roadmap} The organization of the paper follows. Section \ref{sec:model} presents formally the model we consider. Section \ref{sec:impossibility} sums up impossibility 
results while Section \ref{sec:algo} presents our gracefully degrading algorithm. Section \ref{sec:proof} proves the correctness of our gracefully degrading 
algorithm. Finally, Section \ref{sec:conclu} concludes the paper with some comments. 
\else
\myparagraph{Roadmap} The organization of the paper follows. Section \ref{sec:model} presents formally the model we consider. Section \ref{sec:impossibility} sums up impossibility 
results while Section \ref{sec:algo} presents our gracefully degrading algorithm. Section \ref{sec:conclu} concludes the paper. 
\fi

\section{Model}\label{sec:model}


\myparagraph{Dynamic graphs} We consider the model of 
\emph{evolving graphs} \cite{XFJ03}. Time is
discretized and mapped to $\mathbb{N}$. An evolving graph $\mathcal{G}$ is an
ordered sequence $\{G_{0}, G_{1}, \ldots\}$ of subgraphs of a given static graph 
$G = (V,E)$ such that, for any $i \geq 0$, we call $G_{i} = (V, E_{i})$ the
snapshot of $\mathcal{G}$ at time $i$. Note that $V$ is static and $|V|$ is denoted by~$n$. 
We say that the edges of $E_{i}$ are \emph{present} in $\mathcal{G}$ at time $i$. $G$ 
is the \emph{footprint} of $\mathcal{G}$. The 
\emph{underlying graph} of $\mathcal{G}$, denoted by $U_\mathcal{G}$, is the static
graph gathering all edges that are present at least once in $\mathcal{G}$ 
(\ie $U_\mathcal{G}=(V,E_\mathcal{G})$ with
$E_\mathcal{G}=\bigcup_{i=0}^{\infty}E_i$).
An \emph{eventual missing edge} is an edge of $E$ such that there 
exists a time after which this edge is never present in $\mathcal{G}$. A 
\emph{recurrent edge} is an edge of $E$ that is not eventually 
missing. The \emph{eventual underlying graph} of $\mathcal{G}$, denoted
$U_\mathcal{G}^\omega$, is the static graph gathering all recurrent edges of
$\mathcal{G}$ (\ie $U_\mathcal{G}^\omega=(V,E_\mathcal{G}^\omega)$ where 
$E_\mathcal{G}^\omega$ is the set of recurrent edges of $\mathcal{G}$). 
We only consider graphs whose footprints are anonymous 
and unoriented rings of size $n \geq 4$.
The class \COT~(connected-over-time graphs) contains all evolving graphs such that their
eventual underlying graph is connected. 
The class \RE~(recurrent-edges graphs) gathers all
evolving graphs whose footprint contains only recurrent edges. 
The class \BRE~(bounded-recurrent-edges graphs) includes all evolving graphs
in which there exists a $\delta\in\mathbb{N}$ such that each edge of the footprint 
appears at least once every $\delta$ units of time.
The class \AC~(always-connected graphs) collects all evolving graphs where the
graph $G_{i}$ is connected for any $i \in \mathbb{N}$.
The class \ST~(static graphs) encompasses all evolving graphs where the
graph $G_{i}$ is the footprint for any $i \in \mathbb{N}$.

\myparagraph{Robots} We consider systems of $\mathcal{R}\geq 4$ autonomous mobile entities called
robots moving in a discrete and dynamic environment modeled by an evolving graph
$\mathcal{G}=\{(V,E_0),(V,E_1)\ldots\}$, $V$ being a set of nodes representing 
the set of locations where robots may be, $E_i$ being the set of bidirectional 
edges 
through which robots may move from a location to
another one at time $i$. 
Each robot knows $n$ and $\mathcal{R}$. Each robot $r$ possesses a distinct 
(positive) integer identifier $id_r$ strictly greater than 0.
Initially, a robot only knows the value of its own identifier.
Robots have a persistent memory so they can store local variables. 

Each robot $r$ is endowed with strong local multiplicity detection, meaning that it is able to 
count the exact number of robots that are co-located with it at any time $t$.
When this number equals 1, the robot $r$ is \emph{isolated} at time $t$. By opposition,
we define a \emph{tower} $T$ as a couple 
$(S, \theta)$, where $S$ is a set of robots with $|S| > 1$ and 
$\theta=[t_{s}, t_{e}]$ is an interval of $\mathbb{N}$, such that all the robots
of $S$ are located at a same node at each instant of time $t$ in $\theta$ and 
$S$ or $\theta$ is maximal for this property. We say that the robots of $S$ form
the tower at time $t_{s}$ and that they are involved in the tower between time
$t_{s}$ and $t_{e}$. Robots are able to communicate (by direct reading) the values of 
their variables to each others only when they are involved in the same tower.

Finally, all the robots have the same
chirality, \ie each robot is able to locally label the two ports of its current 
node with \emph{left} and \emph{right} consistently over the ring and time and 
all the robots agree on this labeling. Each robot $r$ has a variable 
$dir_{r}$ that stores the direction it currently \emph{considers}
(\emph{right}, \emph{left} or \emph{$\bot$}). 

\myparagraph{Algorithms and execution}
The \emph{state} of a robot at time $t$ corresponds to the values of its local variables 
at time $t$. The \emph{configuration} $\gamma_t$ of the system at time $t$ gathers the
snapshot at time $t$ of the evolving graph, the positions (\ie the nodes where 
the robots are currently located) and the state of each robot at time $t$.
The \emph{view} of a robot $r$ at time $t$ is composed from the state of $r$ at time $t$,
the state of all robots involved in the same tower as $r$ at time $t$ if any, and 
of the following local functions: $ExistsEdge(dir, round)$, with 
$dir \in \{right, left\}$ and $round \in \{current, previous\}$ which indicates
if there exists an adjacent edge to the location of $r$
at time $t$ and $t-1$ respectively in the 
direction $dir$ in $G_t$ and in $G_{t-1}$ respectively; $NodeMate()$
which gives the set of all the robots co-located with $r$
($r$ is not included in this set); $NodeMateIds()$ which gives the
set of all the identifiers of the robots co-located with $r$
(excluded the one of $r$)
; and
$HasMoved()$ which indicates if $r$ has moved between time 
$t-1$ and $t$ (see below). 

The \emph{algorithm} of a robot is written in the form of an ordered set of guarded rules
$(label)::guard\longrightarrow action$ where $label$ is the name of the rule, 
$guard$ is a predicate on the view of the robot, and $action$ is a sequence
of instructions modifying its state. Robots are uniform in the sense they share the same
algorithm. Whenever a robot has at least one rule whose guard is true at time $t$, we
say that this robot is \emph{enabled} at time $t$. The local algorithm also specifies the initial
value of each variable of the robot but cannot restrict its arbitrary initial position.

Given an evolving graph $\mathcal{G}=\{G_{0}, G_{1}, \ldots\}$ and an
initial configuration $\gamma_0$, the \emph{execution} $\sigma$ 
in $\mathcal{G}$ starting from $\gamma_0$ of an algorithm is the maximal sequence 
$(\gamma_0,\gamma_1)(\gamma_1,\gamma_2)(\gamma_2,\gamma_3)\ldots$ where, 
for any $i \geq 0$, 
the configuration $\gamma_{i+1}$ is the result of the execution of a synchronous
round by all robots from $\gamma_i$ that is composed of three atomic and synchronous phases: 
Look, Compute, Move. During the Look phase, each robot captures its view at time $i$.
During the Compute phase, each robot enabled by the algorithm executes the $action$ associated to 
the first rule of the algorithm whose $guard$ is true in its view.
In the case the direction $dir_{r}$ of a robot $r$ is 
in $\{right, left\}$, the Move phase consists of moving $r$ in the 
direction it considers if there exists an adjacent edge in that direction to its 
current node, otherwise (\ie the adjacent edge is missing) $r$ is \emph{stuck}
and hence remains on its current node. In the case where the 
direction $dir$ of a robot is $\bot$, the robot remains on its current node.

\section{Impossibility Results}\label{sec:impossibility}

In this section, we present the set of impossibility results summarized in 
Table \ref{tab:summary}. These results show that some variants of the 
gathering problem cannot be solved depending on the dynamics of the ring in which 
the robots evolve and hence motivate our gracefully degrading approach.

First, we prove in Theorem~\ref{GEImpossible} that \GE~(the eventual variant of the gathering problem) 
is impossible to solve in \AC. Note that Di Luna \etal~\cite{LFPPSV17} 
provide a similar result but show it in an informal way only. Moreover, our 
result subsumes theirs since the considered models are different: we show 
that the result remains valid even if robots are identified, able to 
communicate, and not necessarily initially all scattered
(other different assumptions exist between the 
two models but have no influence on our proof).

The proof of Theorem~\ref{GEImpossible} relies on a generic framework 
introduced by Braud-Santoni \etal~\cite{BDKP16}. Note that even though this 
generic framework is designed for another model (namely, the classical message 
passing model), it is straightforward to borrow it for our current model.
Indeed, as its proof only relies on the determinism of algorithms and 
indistinguishability of dynamic graphs, its arguments are directly 
translatable in our model. We present briefly this framework here. The 
interested reader is referred to the original work \cite{BDKP16} for more 
details. 

This framework is based on a result showing that, if we take a sequence of 
evolving graphs with ever-growing common prefixes (that hence converges to the
evolving graph that shares all these common prefixes), then the sequence of 
corresponding executions of any deterministic algorithm also converges. 
Moreover, we are able to describe the execution to which it converges as the 
execution of this algorithm in the evolving graph to which the sequence 
converges. This result is useful since it allows us to build counterexamples
in the context of impossibility results. Indeed, it is sufficient to construct
a sequence of evolving graphs with ever-growing common prefixes and prove 
that its corresponding execution violates the specification of the problem 
for ever-growing time to exhibit an execution that never satisfies the
specification of the problem.

\begin{theorem} \label{GEImpossible}
 There exists no deterministic algorithm that satisfies \GE~in rings of \AC~with size $4$ or more for $4$ robots or more.
\end{theorem}

\begin{proof}
By contradiction, assume that there exists a deterministic algorithm 
 $\mathcal{A}$ that satisfies \GE~in any
 ring of \AC~with size $4$ or more for $4$ robots or more. 
Let us choose arbitrarily two of these robots and denote them $r_{1}$ and $r_{2}$.

Note that $\mathcal{A}$ may allow the last robot to terminate only if it is
co-located with all other robots
(otherwise, we obtain a contradiction with the safety of \GE). So, 
proving the existence of an execution of $\mathcal{A}$ in a ring of \AC~where 
$r_{1}$ and $r_{2}$ are never co-located is sufficient to obtain a contradiction
with the liveness property of \GE~and to show the result. This is the goal of the 
remainder of the proof.

To help the construction of this execution, we need introduce some notation as follows.
Given an evolving graph $\mathcal{F}$, an edge $\tilde{e}$ of $\mathcal{F}$, and a time
interval $\mathbb{I}\subseteq\mathbb{N}$, the evolving graph 
$\mathcal{F}\backslash \{\tilde{e},\mathbb{I}\}$ is the evolving graph $\mathcal{F}'$ defined by:
$e\in F'_i$ if and only if $e=\tilde{e} \wedge i\notin\mathbb{I} \wedge e\in F_i$ or 
$e\neq\tilde{e} \wedge e\in F_i$.
Given an evolving graph $\mathcal{F}$ and two integers $t_{1}, t_{2}$ such that $t_{1} \leq t_{2}$, 
we denote $\mathcal{F}^{t_{1}, \ldots, t_{2}}$ the subsequence $\{F_{t_{1}}, \ldots, F_{t_{2}}\}$ 
of $\mathcal{F}$.
Given two evolving graphs, $\mathcal{F}$ and $\mathcal{H}$, and an integer $t$, 
the evolving graph $\mathcal{F}^{\{0,\ldots,t\}}\otimes\mathcal{H}^{\{t+1,\ldots,+\infty\}}$
is the evolving graph $\mathcal{F}'$ defined by: $e\in F'_i$ if and only if
$i\leq t\wedge e\in F_i$ or $i> t\wedge e\in H_i$.

 Let $\mathcal{G}= \{G_{0}, G_{1}, \ldots\}$ be a graph of \AC~
 whose footprint $G$ is a ring of size $4$ or more such that 
 $\forall i \in \mathbb{N}, G_{i} = G$.
 Consider two nodes $u$ and $v$ of $\mathcal{G}$, such that the node $v$ is the 
 adjacent node of $u$ in $G$ to the right. 
 We denote by $e_{uv}$ the edge linking the nodes $u$ and $v$. 
 Let $\mathcal{G}'$ be $\mathcal{G} \backslash \{e_{uv}, \mathbb{N}\}$.
 Let $\varepsilon$ be the execution of $\mathcal{A}$ in $\mathcal{G}'$ starting
 from the configuration where $r_{1}$ is located on node $u$ and $r_{2}$ is 
 located on node $v$. Note that the distance in the footprint of $\mathcal{G}$ 
 between $r_{1}$ and $r_{2}$ (denoted $d(r_{1}, r_{2})$) is equal to one.

 \begin{figure}

 		 \centerline{\includegraphics[scale=1]{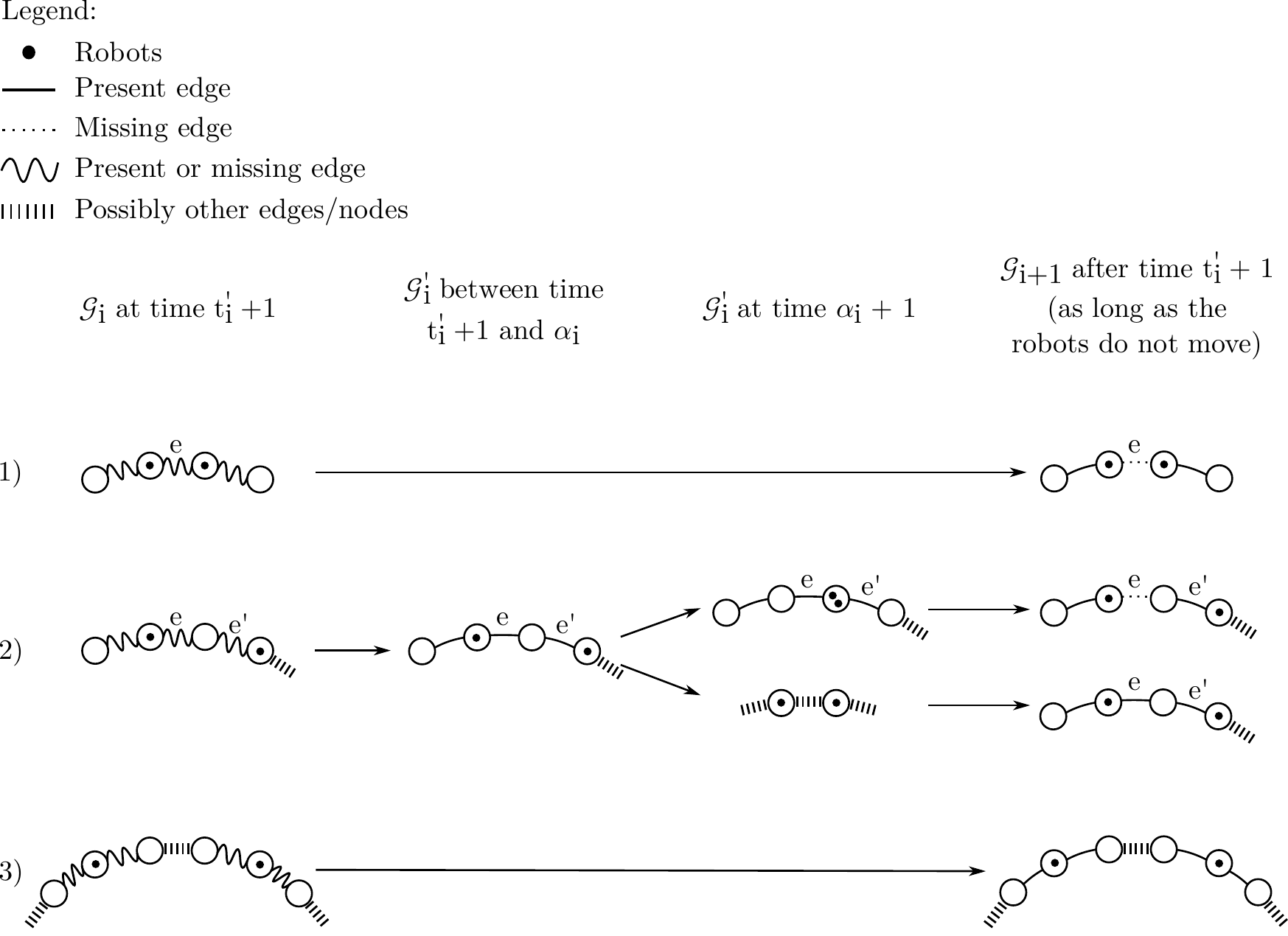}}

 	\caption{Construction of $\mathcal{G}_{i + 1}$ from $\mathcal{G}_{i}$.} \label{impossibility_bis}
 \end{figure}

 Our goal is to construct a sequence of rings of \AC~denoted $(\mathcal{G}_{m})_{m \in \mathbb{N}}$ 
 such that $\mathcal{G}_{0} = \mathcal{G}'$ and, for any $i \geq 0$, $r_{1}$ and
 $r_{2}$ are never co-located before time $t_i$ in $\varepsilon_i$ (the execution 
 of $\mathcal{A}$ in $\mathcal{G}_{i}$ starting from the same configuration as 
 $\varepsilon$), $(t_{m})_{m \in \mathbb{N}}$ being a strictly increasing sequence
 with $t_0=0$.
 First, we show in the next paragraph that, if some such  $\mathcal{G}_{i}$ exists and moreover 
 ensures the existence of a time $t_{i}' + 1>t_i$ where the two robots are 
 still on different nodes in $\varepsilon_i$, then we can construct $\mathcal{G}_{i + 1}$ as shown on Figure~\ref{impossibility_bis}.
 We then prove that our construction guarantees the existence of such a $t_{i}'$,
 implying the well-definition of $(\mathcal{G}_{m})_{m \in \mathbb{N}}$.
 

 As $r_1$ and $r_2$ are not co-located at time $t_{i}$ in $\varepsilon_i$, at least one of them 
 must move in finite time in any execution starting from $\gamma_{t_i}$
 (otherwise, we obtain a contradiction with the liveness of \GE). 
 Let $t_{i}' \geq t_{i}$ be the smallest such time in the execution where
 the topology of the graph does not evolve from time $t_{i}$ to time $t_{i}'$.
 In the following, we show how we construct the evolving graph 
 $\mathcal{G}_{i + 1}$, in function of $t_{i}'$ and $\mathcal{G}_{i}$. 
 As we assume that in $\mathcal{G}_{i}$, at time $t_{i}' + 1$, $r_{1}$ and
 $r_{2}$ are on two different nodes, \ie $d(r_{1}, r_{2}) \geq 1$, the 
 following cases are possible.


 \begin{description}
  \item[Case 1:] $d(r_{1}, r_{2}) = 1$ at time $t_{i}' + 1$.\\ 
  Let $e$ be the edge between the respective locations of $r_{1}$ and $r_{2}$
  at time $t_{i}' + 1$. 
  We define $\mathcal{G}_{i + 1}$ on the same footprint as $\mathcal{G}_{i}$ by
  $\mathcal{G}_{i + 1} = \mathcal{G}_{i}^{0, \ldots, t_{i}'} \otimes
  (\mathcal{G}^{t_{i}' + 1, \ldots, +\infty} \backslash 
  \{e, \{t_{i}'+1, \ldots, +\infty\}\})$.
 
  \item[Case 2:] $d(r_{1}, r_{2}) = 2$ at time $t_{i}' + 1$.\\ 
  Denote $e$ and $e'$ the two consecutive edges between the respective locations 
  of $r_{1}$ and $r_{2}$ at time $t_{i}' + 1$.
  We define first $\mathcal{G}'_{i}$ on the same footprint as $\mathcal{G}_{i}$ by
  $\mathcal{G}_{i}' = \mathcal{G}_{i}^{0, \ldots, t_{i}'} \otimes
  \mathcal{G}^{t_{i}' + 1, \ldots, +\infty}$. Note that $\mathcal{G}'_{i}$ belongs
  to \AC~by assumption on $\mathcal{G}_i$ and since $\mathcal{G}$ is the static ring.
  Then, to avoid a contradiction with the liveness of \GE, we know that
  there exists a time $\alpha_{i} \geq t_{i}' + 1$ 
  in the execution of $\mathcal{A}$ on $\mathcal{G}'_{i}$ where at least one of our two robots
  move (\emph{w.l.o.g.} assume that $\alpha_{i}$ is the smallest one).
  If, at time $\alpha_{i} + 1$, the two robots are on distinct nodes in $\mathcal{G}_{i}'$, then 
  we define $\mathcal{G}_{i + 1}$ on the same footprint as $\mathcal{G}_{i}$ by
  $\mathcal{G}_{i + 1} = \mathcal{G}_{i}^{0, \ldots, t_{i}'} \otimes
  \mathcal{G}^{t_{i}' + 1, \ldots, +\infty}$. 
  If, at time $\alpha_{i} + 1$, the two robots are on the same node in $\mathcal{G}_{i}'$, then 
  we define $\mathcal{G}_{i + 1}$ on the same footprint as $\mathcal{G}_{i}$ by
  $\mathcal{G}_{i + 1} = \mathcal{G}_{i}^{0, \ldots, t_{i}'} \otimes
  (\mathcal{G}^{t_{i}' + 1, \ldots, +\infty} 
  \backslash \{e, \{t_{i}' + 1, \ldots, +\infty\}\})$.

  \item[Case 3:] $d(r_{1}, r_{2}) > 2$ at time $t_{i}' + 1$.\\ 
  We define $\mathcal{G}_{i + 1}$ on the same footprint as $\mathcal{G}_{i}$ by 
  $\mathcal{G}_{i + 1} = \mathcal{G}_{i}^{0, \ldots, t_{i}'} \otimes
  \mathcal{G}^{t_{i}' + 1, \ldots, +\infty}$.
 \end{description}

 Note that $\mathcal{G}_{i}$ and $\mathcal{G}_{i + 1}$ are indistinguishable for
 robots until time $t_{i}'$. This implies that, at time $t_{i}' + 1$, 
 $r_{1}$ and $r_{2}$ are on the same nodes in $\varepsilon_{i}$ and in 
 $\varepsilon_{i + 1}$. By construction of $t_{i}'$, either $r_{1}$ or $r_{2}$
 or both of the two robots move at time $t_{i}'$ in $\varepsilon_{i + 1}$. 
 Moreover, by construction of $\mathcal{G}_{i}$, even if one or both of the 
 robots move during the Move phase of time $t_{i}'$, at time $t_{i}' + 1$ the
 robots are still on two distinct nodes
 (since, in all cases above, either 
 the distance between the robots before the move is strictly greater than 2,
 an edge between the two robots is missing before the move and prevents the meeting, or
 the two robots move in a way that prevents the meeting by indistinguishability
 between $\mathcal{G}_{i}$ and $\mathcal{G}_{i + 1}$).
 Note that, by construction, 
 $\mathcal{G}_{i + 1}$ has at most one edge missing at each instant time
 and hence belongs to \AC.
 
 Defining $t_{i + 1} = t_{i}' + 1$, we succeed to construct 
 $\mathcal{G}_{i + 1}$ with the desired properties. Note that $t_{i}'$ and $\mathcal{G}_{0}$
 trivially satisfy all our assumptions.
 In other words, ($\mathcal{G}_{m}$)$_{m \in \mathbb{N}}$ is well-defined.
 
 We can then define the evolving graph $\mathcal{G}_{\omega}$ such that
 $\mathcal{G}_{\omega}$ and $\mathcal{G}_{0}$ have the same footprint, and such
 that for all $i \in \mathbb{N}$, $\mathcal{G}_{\omega}$ shares a common prefix
 with $\mathcal{G}_{i}$ until time $t_{i}'$.
 As the sequence $(t_{m}$)$_{m \in \mathbb{N}}$ is increasing by construction,
 this implies that the sequence ($\mathcal{G}_{m}$)$_{m \in \mathbb{N}}$
 converges to $\mathcal{G}_{\omega}$.
 Applying the theorem of Braud-Santoni \etal~\cite{BDKP16}, we obtain that, until time $t_{i}'$,
 the execution of $\mathcal{A}$ in $\mathcal{G}_{\omega}$ is identical
 to the one in $\mathcal{G}_{i}$. This implies that, executing $\mathcal{A}$ in 
 $\mathcal{G}_{\omega}$ (whose footprint is a ring of size $4$ or more), $r_{1}$ 
 and $r_{2}$ are always on distinct nodes, contradicting the liveness
 of \GE~and proving the result.
\end{proof}


It is possible to derive some other impossibility results from Theorem 
\ref{GEImpossible}. Indeed, the inclusion $\AC\subset\COT$ allows us to 
state that \GE~is impossible under \COT~as well. 

\begin{corollary} \label{GEImpossibleCOT}
 There exists no deterministic algorithm that satisfies \GE~in rings of \COT~with size $4$ or more for $4$ robots or more.
\end{corollary}

From the very definitions of \G~and \GE, it is straightforward to see that the impossibility 
of \GE~under a given class implies the one of \G~under the same class.

\begin{corollary} \label{GImpossible}
 There exists no deterministic algorithm that satisfies \G~in rings of \COT~or \AC~with size $4$ or more for $4$ robots or more.
\end{corollary}

Finally, impossibility results for bounded variants of the gathering problem
(\ie the impossibility of \G~under \RE~and of \GW~under \COT~and \RE)
are obtained as follows.
The definition of \COT~and \RE~does not exclude the ability for all edges
of the graph to be missing initially and for any arbitrary long time, hence 
preventing the gathering of robots for any arbitrary long time
if they are initially scattered.
This observation is sufficient to prove a contradiction with the existence
of an algorithm solving \G~or \GW~in these classes.

\begin{corollary} \label{GGWImpossible}
 There exists no deterministic algorithm that satisfies \G~or \GW~in rings of \COT~or \RE~with size $4$ or more for $4$ robots or more.
\end{corollary}




\section{Gracefully Degrading Gathering}\label{sec:algo} 

This section presents \Gathering, our gracefully degrading gathering algorithm, that
aims to solve different variants of the gathering problem under various dynamics (refer to Table \ref{tab:summary}). 

In Subsection~\ref{sub:overview}, we informally describe our algorithm clarifying which
variant of gathering is satisfied within which class of evolving graphs.
Next, Subsection~\ref{sub:formalAlgorithm} presents formally the algorithm.

\subsection{Overwiew}\label{sub:overview}

~Our algorithm has to overcome various difficulties. First, robots are evolving in an environment in which no node can
be distinguished.  So, 
the trivial algorithm in which the robots meet on a particular node is impossible. 
Moreover, since the footprint of the graph is a ring, (at most) one of the $n$ 
edges may be an eventual missing edge. 
This is typically the case of classes~$\COT$ and~$\AC$. 
In that case, no robot is able to distinguish an eventual missing edge from a missing edge that will appear later 
in the execution.  In particular, a robot 
stuck by a missing edge does not know whether it can wait for the missing edge to appear again or not.
Finally, despite the fact that no robot is aware of which class of dynamic 
graphs robots are evolving in, the algorithm is required to meet at least the specification of
the gathering according to the class of dynamic graphs in which it is
executed or a better specification than this one. 

The overall scheme of the algorithm consists in first detecting $r_{min}$, the robot having the minimum identifier 
so that the $\mathcal{R}$ robots eventually gather on its node (\ie, satisfying  specification~$\GE$).
Of course, depending on 
 the particular evolving graph 
in which our algorithm is executed,
$\GE$ may not achieved. 
In the weakest class (class~$\COT$) and the ``worst'' possible evolving graph, 
one can expect specification~$\GEW$ only, \ie, at least $\mathcal{R} - 1$ robots gathered.

The algorithm proceeds in four successive phases: \AmITheMin, \MinWaitToBeKnown, \Walk, and \WaitTermination.  
Actually, again depending on the class of graphs and the evolving graph 
in which our algorithm is executed,  
we will see that the four phases are not necessarily all executed since the execution can be
stopped prematurely, especially in case where $\GE$ (or $\G$) is achieved.  By contrast, they can also never be completed 
in some weak settings (namely $\AC$ or $\COT$), solving $\GEW$ (or $\GW$) only.

\begin{algorithm}[t]
 \caption{Predicates used in \Gathering} \label{predicates}
 \footnotesize
      
      \textbf{MinDiscovery()} $\equiv$
      
	    \quad \begin{tabular}{l}
			  $(state_{r} = potentialMin \wedge \exists r' \in NodeMate(), (state_{r'} = righter$ $\wedge$ $id_{r} < id_{r'})) \vee$ \\ 
			  $\exists r' \in NodeMate(), idMin_{r'} = id_{r} \vee$ \\
			  $\exists r' \in NodeMate(), (state_{r'} \in \{dumbSearcher, potentialMin\}$ $\wedge$ $id_{r} < idPotentialMin_{r'}) \vee$ \\
			  $rightSteps_{r} = 4 * id_{r} * n$ \\                         
		   \end{tabular}

      \textbf{\GEBold$()$} $\equiv$
     
	  \quad \quad $|NodeMate()| = \mathcal{R} - 1$

      \textbf{\GEWBold$()$} $\equiv$
      
	  \quad \quad $|NodeMate()| = \mathcal{R} - 2$ $\wedge$ $\exists r' \in \{r\} \cup NodeMate(), state_{r'} \in \{minWaitingWalker, minTailWalker\}$ 
      
      \textbf{HeadWalkerWithoutWalkerMate()} $\equiv$
   
	  \quad \quad $state_{r} = headWalker$ $\wedge$ $ExistsEdge(left, previous)$ $\wedge$ $\lnot HasMoved()$ 
	         $\wedge$ $NodeMateIds() \neq walkerMate_{r}$

      \textbf{LeftWalker()} $\equiv$

	  \quad \quad $state_{r} = leftWalker$ 

      \textbf{HeadOrTailWalkerEndDiscovery()} $\equiv$
      
	  \quad \quad $state_{r} \in \{headWalker, tailWalker, minTailWalker\}$ $\wedge$ $walkSteps_{r} = n$

      \textbf{HeadOrTailWalker()} $\equiv$
      
	  \quad \quad $state_{r} \in \{headWalker, tailWalker, minTailWalker\}$

      \textbf{AllButTwoWaitingWalker()} $\equiv$
       
	  \quad \quad $|NodeMate()| = \mathcal{R} - 3$ $\wedge$
				  $\forall r' \in \{r\} \cup NodeMate(), state_{r'} \in \{waitingWalker, minWaitingWalker\}$

      \textbf{WaitingWalker()} $\equiv$
     
	  \quad \quad $state_{r} \in \{waitingWalker, minWaitingWalker\}$

      \textbf{PotentialMinOrSearcherWithMinWaiting(r')} $\equiv$
    
	  \quad \quad $state_{r} \in \{potentialMin, dumbSearcher, awareSearcher\}$ $\wedge$ $state_{r'} = minWaitingWalker$

      \textbf{RighterWithMinWaiting(r')} $\equiv$
      
	  \quad \quad $state_{r} = righter$ $\wedge$ $state_{r'} = minWaitingWalker$

      \textbf{NotWalkerWithHeadWalker(r')} $\equiv$
    
	  \quad \quad $state_{r} \in \{righter, potentialMin, dumbSearcher, awareSearcher\}$ $\wedge$ $state_{r'} = headWalker$

      \textbf{NotWalkerWithTailWalker(r')} $\equiv$
    
	  \quad \quad $state_{r} \in \{righter, potentialMin, dumbSearcher, awareSearcher\}$ $\wedge$ $state_{r'} = minTailWalker$ 
      
      \textbf{PotentialMinWithAwareSearcher(r')} $\equiv$
      
	  \quad \quad $state_{r} = potentialMin$ $\wedge$ $state_{r'} = awareSearcher$ 
   
      \textbf{AllButOneRighter()} $\equiv$
      
	  \quad \quad $|NodeMate()| = \mathcal{R} - 2$ $\wedge$ $\forall r' \in \{r\} \cup NodeMate(), state_{r'} = righter$ 

      \textbf{RighterWithSearcher(r')} $\equiv$
  
	  \quad \quad $state_{r} = righter$ $\wedge$ $state_{r'} \in \{dumbSearcher, awareSearcher\}$

      \textbf{PotentialMinOrRighter()} $\equiv$
    
	  \quad \quad $state_{r} \in \{potentialMin, righter\}$

      \textbf{DumbSearcherMinRevelation()} $\equiv$
    
	  \quad \quad $state_{r} = dumbSearcher$ $\wedge$ $\exists r' \in NodeMate(), (state_{r'} = righter$ $\wedge$ $id_{r'} > idPotentialMin_{r})$

      \textbf{DumbSearcherWithAwareSearcher(r')} $\equiv$
    
	  \quad \quad $state_{r} = dumbSearcher$ $\wedge$ $state_{r'} = awareSearcher$

      \textbf{Searcher()} $\equiv$
   
	  \quad \quad $state_{r} \in \{dumbSearcher, awareSearcher\}$ 
     
\end{algorithm}

\myparagraph{Phase~\AmITheMin}
This phase leads each robot 
to know whether it possesses the minimum identifier.
Initially every robot $r$ considers the $right$ direction.  
Then $r$ always moves to the $right$ until it moves 
$4 * n * id_{r}$ steps on the right, where $id_{r}$ is the identifier of $r$ and $n$, the size of the ring. 
The first robot that succeeds to do so is necessarily $r_{min}$.  
Depending on the class of graph, one eventual missing edge may exist, preventing 
$r_{min}$ to move on the $right$ direction during $4 * n * id_{r_{min}}$ steps.

However, in that case at least $\mathcal{R} - 1$ robots succeed
to be located on a same node, but not necessarily the node where $r_{min}$ is located. 
Note that the weak form of gathering ($\GEW$) could be solved in that case.  
However, the $\mathcal{R} - 1$ robots gathered cannot stop their execution. Indeed, 
our algorithm aims at gathering the robots on the node occupied by $r_{min}$. However, 
$r_{min}$ may not be part of the $\mathcal{R} - 1$ robots that gathered. Further, it is possible for 
$\mathcal{R} - 1$ robots to gather (without $r_{min}$) even when $r_{min}$ succeeds in moving 
$4 * n * id_{r_{min}}$ steps to the right (\ie even when $r_{min}$ stops to move because it 
completed Phase~\AmITheMin). In that case, if the $\mathcal{R} - 1$ robots that 
gathered stop their execution, \GE~cannot be solved in \RE, \BRE~and \ST~rings,
as \Gathering~should do. Note that, it is also possible for $r_{min}$ to be part 
of the $\mathcal{R} - 1$ robots that gathered. 


Recall that robots can communicate when they are both located in the same node. So, the
$\mathcal{R} - 1$ robots may be aware of the identifier of the robot with the minimum 
identifier among them. Since it can or cannot be the actual $r_{min}$, let us call this robot \potentialMin.  
Then, driven by \potentialMin, a search phase starts during which the
$\mathcal{R} - 1$ robots try to visit all the nodes of the ring infinitely often in both
directions by subtle round trips. Doing so, 
$r_{min}$ eventually knows that it possesses the actual minimum identifier.

\myparagraph{Phase~\MinWaitToBeKnown}
The goal of the second phase consists in spreading the identifier of $r_{min}$ among 
the other robots.  The basic idea is that during this phase, $r_{min}$ stops moving and waits until 
$\mathcal{R} - 3$ other robots join it on its node so that its identifier is known by at least $\mathcal{R} - 3$ other robots. 
The obvious question arises:``{\em Why waiting for $\mathcal{R} - 3$ extra robots only?}''.
A basic idea to gather could be that once $r_{min}$ is aware that it possesses the minimum identifier, 
it can just stop to move and just wait for the other robots to eventually reach its location, just 
by moving toward the right direction. 
Actually, depending on 
the particular evolving graph considered
one missing edge $e$ may eventually appear, preventing robots from reaching $r_{min}$ by moving toward the same direction
only.  That is why the gathering of the $\mathcal{R} - 2$ robots is eventually achieved by the same search phase
as in Phase~\AmITheMin.  
However, by doing this, it is possible to have $2$ robots stuck on each extremity of $e$. 
Further, these two robots cannot change the directions they consider since a robot is not able to distinguish an eventual 
missing edge from a missing edge that will appear again later. This is why
during Phase~\MinWaitToBeKnown, $r_{min}$ stops to move until 
$\mathcal{R} - 3$ other robots join it to form a tower of $\mathcal{R} - 2$ robots. 
In this way these $\mathcal{R} - 2$ robots start the third phase simultaneously.

\myparagraph{Phase~\Walk}
The third phase is a {\em walk} made by the tower of $\mathcal{R} - 2$ robots. 
The $\mathcal{R} - 2$ robots are split into two distinct groups, {\em Head} and {\em Tail}. 
Head is the unique robot
with the maximum identifier of the tower.
Tail, composed of $\mathcal{R} - 3$ robots, is made of the other robots of the tower, led by
$r_{min}$.
Both move alternatively in the $right$ direction during $n$ steps 
such that between two movements of a given group the two groups are again 
located on a same node. This movement permits to prevent the two robots that do 
not belong to any of these two groups to be both stuck on different extremities
of an eventual missing edge (if any) once this walk is finished. Since 
there exists at most one eventual missing edge, we are 
sure that if the robots that have executed the walk stop moving forever, then 
at least one robot can join them during the next and last phase.

As noted, it can exist an eventual missing edge, therefore, Head and Tail may not 
complete Phase \Walk. Indeed, one of the two situations below may occur: $(i)$ Head and Tail together form a tower of $\mathcal{R} - 2$ robots but an eventual missing 
edge on their right prevents them to complete Phase \Walk; $(ii)$ Head and Tail are located on neighboring node and the edge between them is an eventual missing edge
that prevents Head and Tail to continue to move alternatively. 

Call $u$ the node where Tail is stuck on an eventual missing edge.
In the two situations described even if Phase \Walk~is not complete by both Head
and Tail, either $\GE$ or $\GEW$ is solved. Indeed, in the first situation, 
necessarily at least one robot $r$ succeeds to join $u$. In fact, either 
$r$ considers the good direction to reach $u$ or it meets a robot on the other 
extremity of the eventual missing edge that makes it considers the good direction 
to reach $u$. In the second situation, necessarily at least two robots $r$ and
$r'$ succeed to join $u$. This is done either because $r$ and $r'$ consider the
good direction to reach $u$ or because they reach the node where Head is located 
without Tail making them consider the good direction to reach $u$.

Once a tower of $\mathcal{R} - 1$ robots including $r_{min}$ is formed, 
$\GEW$ is solved. Then, the latter robot tries to reach the tower to 
eventually solve $\GE$ in favorable cases.


\myparagraph{Phase~\WaitTermination}
The last phase starts once the robots of Head have completed Phase \Walk. 
If it exists a time at which the robots of Tail complete Phase \Walk, then 
Head and Tail form a tower of $\mathcal{R} - 2$ robots and stop moving. As 
explained in the previous phase, {Phase~\Walk} ensures that at least one extra 
robot eventually joins the node where Head and Tail are located to form a tower of $\mathcal{R} - 1$ robots. 
Once a tower of $\mathcal{R} - 1$ robots including $r_{min}$ is formed, 
$\GEW$ is solved. Then,
the latter robot tries to reach the tower to eventually solve $\GE$ in favorable cases. 
In the case the robots of Tail never complete the phase \Walk, then this implies
that Head and Tail are located on neighboring node and that the edge between them is an 
eventual missing edge. As described in Phase \Walk~either 
\GEW~or \GE~is solved.

\begin{algorithm}[h!]
 \caption{Functions used in \Gathering} \label{functions}
 \footnotesize
 
    \textbf{Function StopMoving()} \label{algo:StopMoving}
    \begin{algorithmic} 
	\State $dir_{r} := \bot$ \smallbreak
    \end{algorithmic}
    
    \textbf{Function MoveLeft()} \label{algo:MoveLeft}
    \begin{algorithmic} 
	\State $dir_{r} := left$ \smallbreak
    \end{algorithmic}
    
    \textbf{Function BecomeLeftWalker()} \label{algo:BecomeLeftWalker}
    \begin{algorithmic}
	 \State $(state_{r}, dir_{r}) := (leftWalker, \bot)$ 
    \end{algorithmic}

    \textbf{Function Walk()} \label{algo:Walk_functions}
    \begin{algorithmic} 
	\State $dir_{r} := 
        \left\{
	       \begin{array}{ll}
		 \bot & $ if $(id_{r} = idHeadWalker_{r} \wedge walkerMate_{r} \neq NodeMateIds()) \vee \\ 
				     & \ \ \ \ (id_{r} \neq idHeadWalker_{r} \wedge idHeadWalker_{r} \in NodeMateIds())\\
		 right & $ otherwise $ 
	       \end{array}
	       \right.$ \smallbreak
	\State $walkSteps_{r} := walkSteps_{r} + 1$ if $dir_{r} = right \wedge ExistsEdge(right, current)$
    \end{algorithmic}
    
   \textbf{Function InitiateWalk()} \label{algo:InitiateWalk}
    \begin{algorithmic} 
	\State $idHeadWalker_{r} :=$ \Call{max}{$\{id_{r}\} \cup NodeMateIds()$}
	\State $walkerMate_{r} := NodeMateIds()$
	\State $state_{r} := 
	       \left\{
	       \begin{array}{ll}
		 headWalker & $ if $id_{r} = idHeadWalker_{r}\\
		 minTailWalker & $ if $state_{r} = minWaitingWalker\\
		 tailWalker & $ otherwise $ 
	       \end{array}
	       \right.$ \smallbreak
    \end{algorithmic} 
     
    \textbf{Function BecomeWaitingWalker(r')} \label{algo:BecomeWaitingWalker} 
    \begin{algorithmic} 
	\State $(state_{r}, idPotentialMin_{r}, idMin_{r}, dir_{r}) := (waitingWalker, id_{r'}, id_{r'}, \bot)$ \smallbreak
    \end{algorithmic}
    
    \textbf{Function BecomeMinWaitingWalker()} \label{algo:BecomeMinWaitingWalker}
    \begin{algorithmic} 
	\State $(state_{r}, idPotentialMin_{r}, idMin_{r}, dir_{r}) := (minWaitingWalker, id_{r}, id_{r}, \bot)$ \smallbreak
    \end{algorithmic}
    
    \textbf{Function BecomeAwareSearcher(r')} \label{algo:BecomeAwareSearcher}
    \begin{algorithmic} 
	 \State $(state_{r}, dir_{r}) := (awareSearcher, right)$
	 \State $(idPotentialMin_{r}, idMin_{r}) := 
	       \left\{
	       \begin{array}{ll}
	           (idPotentialMin_{r'}, idPotentialMin_{r'}) & $ if $state_{r'} = dumbSearcher \\ 
		   (idMin_{r'}, idMin_{r'}) & $ otherwise $ 
	       \end{array}
	       \right.$ \smallbreak
    \end{algorithmic}
        
    \textbf{Function BecomeTailWalker(r')} \label{algo:UpdateWalkDatas}
    \begin{algorithmic} 
	\State $(state_{r}, idPotentialMin_{r}, idMin_{r}) := (tailWalker, idPotentialMin_{r'}, idMin_{r'})$
	\State $(idHeadWalker_{r}, walkerMate_{r}, walkSteps_{r}) := (idHeadWalker_{r'}, walkerMate_{r'}, walkSteps_{r'})$ \smallbreak
    \end{algorithmic}
    
    \textbf{Function MoveRight()} \label{algo:MoveRight}
    \begin{algorithmic} 
	\State $dir_{r} := right$
	\State $rightSteps_{r} := rightSteps_{r} + 1$ if $ExistsEdge(dir, current)$ \smallbreak
    \end{algorithmic}
    
    \textbf{Function InitiateSearch()} \label{algo:InitiateSearch}
    \begin{algorithmic} 
    
	\State $idPotentialMin_{r} :=$ \Call{min}{$\{id_{r}\} \cup NodeMateIds()$}
	\State $state_{r} :=
	       \left\{
	       \begin{array}{ll}
		 potentialMin & $ if $id_{r} = idPotentialMin_{r}\\
		 dumbSearcher & $ otherwise $ 
	       \end{array}
	       \right.$ 
	\State $rightSteps_{r} := rightSteps_{r} + 1$ if $state_{r} = potentialMin$ $\wedge$ $ExistsEdge(dir, current)$ \smallbreak
    \end{algorithmic}

    \textbf{Function Search()} \label{algo:Search}
    \begin{algorithmic}
		\State $dir_{r} :=
		      \left\{
		      \begin{array}{ll}
			left & $ if $|NodeMate()| \geq 1 \wedge id_{r} =$ \Call{max}{$\{id_{r}\} \cup NodeMateIds()$}$\\
			right & $ if $|NodeMate()| \geq 1 \wedge id_{r} \neq$ \Call{max}{$\{id_{r}\} \cup NodeMateIds()$}$\\
			dir_{r} & $ otherwise $
		      \end{array}
		      \right.$ \smallbreak
    \end{algorithmic}
\end{algorithm}

\subsection{Algorithm} \label{sub:formalAlgorithm}

Before presenting formally our algorithm, we first describe the set of
variables of each robot. We recall that each robot $r$ knows $\mathcal{R}$, 
$n$ and $id_{r}$ as constants.

In addition to the variable $dir_{r}$ (initialized to $right$), each robot $r$
possesses seven variables described below. 
Variable $state_{r}$ allows 
the robot $r$ to know which phase of the algorithm
it is performing and (partially) indicates which movement the robot has to
execute. The possible values for this variable are \righter, \dumbSearcher, 
\awareSearcher, \potentialMin, \waitingWalker, \minWaitingWalker, \headWalker,
\tailWalker, \minTailWalker~and \leftWalker. 
Initially, $state_{r}$ is equal to
\righter. Initialized with $0$, $rightSteps_{r}$ counts the number of steps done by
$r$ in the $right$ direction when $state_{r} \in \{$\righter, \potentialMin$\}$.
The next variable is $id\-Po\-ten\-tial\-Min_{r}$.  Initially equals to $-1$, $id\-Po\-ten\-tial\-Min_{r}$
contains the identifier of 
the robot that possibly possesses the minimum 
identifier (a positive integer) of the system. This variable is especially set when $\mathcal{R} - 1$
\righter~are located on a same node. In this case, the variable 
$id\-Po\-ten\-tial\-Min_{r}$ of each robot $r$ that is involved in the tower of 
$\mathcal{R} - 1$ robots is set to the value of the minimum identifier possessed 
by these robots. 
The variable $idMin_{r}$ indicates the identifier of the robot that possesses the actual
minimum identifier among all the robots of the system. This variable is
initially set to $-1$. Let $walkerMate_{r}$ be the set of all the 
identifiers of the $\mathcal{R} - 2$ robots that initiate the Phase \Walk. 
Initially this variable is set to $\emptyset$. The counter $walkSteps_{r}$, initially $0$, 
maintains the number of steps done in the right direction while $r$ 
performs the Phase \Walk.  Finally, the 
variable $idHeadWalker_{r}$ contains the identifier of the robot that
plays the part of Head during the Phase \Walk.

\begin{algorithm}[t!]
\caption{\Gathering \label{algo:main}}
 \footnotesize
     
    \textbf{Rules for Termination} \label{algo:Termination}
    \begin{algorithmic} 
	\State \TerminationOne~:: \GE$()$ $\longrightarrow$ \textbf{terminate}
	\State \TerminationTwo~:: \GEW$()$ $\longrightarrow$ \textbf{terminate} \smallbreak
    \end{algorithmic}

    \textbf{Rules for Phase \WaitTermination} \label{algo:WaitTermination}
    \begin{algorithmic} 
	\State \rFourOne~:: $LeftWalker()$ $\longrightarrow$ \Call{MoveLeft}{}()
	\State \rFourTwo~:: $HeadWalkerWithoutWalkerMate()$ $\longrightarrow$ \Call{BecomeLeftWalker}{}()
	\State \rFourThree~:: $HeadOrTailWalkerEndDiscovery()$ $\longrightarrow$ \Call{StopMoving}{}() \smallbreak 
    \end{algorithmic}
    
    \textbf{Rules for Phase \Walk} \label{algo:Walk_rules}
    \begin{algorithmic} 
	\State \rThreeOne~:: $HeadOrTailWalker()$ $\longrightarrow$ \Call{Walk}{}() \smallbreak
    \end{algorithmic}

    \textbf{Rules for Phase \MinWaitToBeKnown} \label{algo:LeaderWaitToBeKnown}
    \begin{algorithmic} 
	\State \rTwoOne~:: $AllButTwoWaitingWalker()$ $\longrightarrow$ \Call{InitiateWalk}{}()
	\State \rTwoTwo~:: $WaitingWalker()$ $\longrightarrow$ \Call{StopMoving}{}()
		
	\State \rTwoThree~:: $\exists r' \in NodeMate(), PotentialMinOrSearcherWithMinWaiting(r')$ \newline $\text{~~~~~~~~~~~}\longrightarrow$ \Call{BecomeWaitingWalker}{r'} 
	
	\State \rTwoFour~:: $\exists r' \in NodeMate(), RighterWithMinWaiting(r')$ $\wedge$ $ExistsEdge(right, current)$ \newline $\text{~~~~~~~~~~~}\longrightarrow$ \Call{BecomeAwareSearcher}{r'} \smallbreak
	
    \end{algorithmic}

    \textbf{Rules for Phase \AmITheMin} \label{algo:AmITheLeader?}
    \begin{algorithmic} 
	\State \rOneOne~:: $PotentialMinOrRighter()$ $\wedge$ $MinDiscovery()$ $\longrightarrow$ \Call{BecomeMinWaitingWalker}{r}
	
	\State \rOneTwo~:: $\exists r' \in NodeMate(), NotWalkerWithHeadWalker(r')$ $\wedge$ $ExistsEdge(right, current)$ \newline $\text{~~~~~~~~~~~}\longrightarrow$ \Call{BecomeAwareSearcher}{r'}
	\State \rOneThree~:: $\exists r' \in NodeMate(), NotWalkerWithHeadWalker(r')$ \newline$\text{~~~~~~~~~~~}\longrightarrow$ \Call{BecomeAwareSearcher}{r'};  \Call{StopMoving}{}()
	\State \rOneFour~:: $\exists r' \in NodeMate(), NotWalkerWithTailWalker(r')$ $\longrightarrow$ \Call{BecomeTailWalker}{r'}; \Call{Walk}{}()

	\State \rOneFive~:: $\exists r' \in NodeMate(), PotentialMinWithAwareSearcher(r')$ \newline$\text{~~~~~~~~~~~}\longrightarrow$ \Call{BecomeAwareSearcher}{r'}; \Call{Search}{}()
		
	\State \rOneSix~:: $AllButOneRighter()$ $\longrightarrow$ \Call{InitiateSearch}{}()
	\State \rOneSeven~:: $\exists r' \in NodeMate(), RighterWithSearcher(r')$ $\longrightarrow$ \Call{BecomeAwareSearcher}{r'}; \Call{Search}{}() 
	\State \rOneEight~:: $PotentialMinOrRighter()$ $\longrightarrow$ \Call{MoveRight}{}() 
		
	\State \rOneNine~:: $DumbSearcherMinRevelation()$ $\longrightarrow$ \Call{BecomeAwareSearcher}{r}; \Call{Search}{}()
	\State \rOneTen~:: $\exists r' \in NodeMate(), DumbSearcherWithAwareSearcher(r')$\newline $\text{~~~~~~~~~~~}\longrightarrow$ \Call{BecomeAwareSearcher}{r'}; \Call{Search}{}()
	\State \rOneEleven~:: $Searcher()$ $\longrightarrow$ \Call{Search}{}()
    \end{algorithmic}
\end{algorithm}

Moreover, we assume the existence of
a specific instruction: \textbf{terminate}. By executing this instruction, a robot 
stops executing the cycle Look-Compute-Move forever.

To ease the writing of our algorithm, we define a set of predicates 
(presented in Algorithm~\ref{predicates}) and functions (presented in 
Algorithm~\ref{functions}), that are used in our gracefully degrading algorithm 
\Gathering.
Recall that, during the Compute phase, only the first rule whose $guard$ is
true in the view of an enabled robot is executed.

\section{Proofs of correctness of \Gathering} \label{sec:proof}

In this section, we first prove, in Subsection~\ref{GatheSolveCOT}, that \Gathering~solves \GEW~in 
\COT~rings. Then, in Subsection~\ref{gracefull}, we consider \AC, \RE, 
\BRE~and \ST~rings and for each of these classes of 
dynamic rings, we give the problem \Gathering~solves in it.

We want to prove that, while executing \Gathering, at least $\mathcal{R} - 1$
robots terminate their execution on the same node. Therefore, in the proofs of 
correctness, we show that our algorithm forces the robots to execute either
Rule \TerminationOne~or Rule \TerminationTwo~whatever the harsh situation.
Hence, the proofs are given in the case where these rules are not executed 
accidentally.

In the following, for ease of reading, we abuse the various values of the variable $state$ to qualify the robots. 
For instance, if the current value of variable $state$ of a robot is \righter, 
then we say that the robot is a \righter~robot. Let us call $min$ a robot such that
its variable $state$ is equal either to \minWaitingWalker~or to \minTailWalker.

\subsection{\Gathering~solves \GEW~in \COT~rings} \label{GatheSolveCOT}

In this subsection, we prove that \Gathering~solves \GEW~in 
\COT~rings. Since \Gathering~is divided into four phases, we prove
each of these phases hereafter.

\subsubsection{Proofs of Correctness of Phase \AmITheMin}

We recall that the goal of Phase \AmITheMin~of our algorithm is to make the
robot with the minimum identifier aware that it possesses the minimum
identifier among all the robots of the system. In our algorithm a robot is aware 
that it possesses the minimum identifier when it is $min$. Therefore, in this 
section we have to prove that only $r_{min}$ can become $min$, and that $r_{min}$
effectively becomes $min$ in finite time. We prove this respectively in 
Lemmas~\ref{atMostOneLeader} and \ref{atLeastOneLeader}.

%
%
%
%

First we give two observations that help us all along the proves of each phase.

\begin{observation} \label{noMoreMovingRightAndPLeader}
 By the rules of \Gathering, a robot whose $state$ is not either \righter~or
 \potentialMin~cannot become a \righter~or a \potentialMin.
\end{observation}

\begin{observation} \label{noMoreRighter}
 By the rules of \Gathering, a robot whose $state$ is not \righter~cannot 
 become a \righter~robot.
\end{observation}

While executing \Gathering, once a robot knows that it possesses the minimum 
identifier, it remembers this information. In other words, once a robot becomes
$min$ it stays $min$ during the rest of the execution. We prove this statement 
in the following lemma.

\begin{lemma} \label{LeaderTheWholeExecution}
 $min$ is a closed state under \Gathering.
\end{lemma}

\begin{proof}
 A robot is a $min$ when its state is either equal to \minWaitingWalker~or to 
 \minTailWalker. A \minTailWalker~robot can only execute the rules 
 \rFourThree~and \rThreeOne~that do not update the variable $state$. A 
 \minWaitingWalker~robot can only execute the rules \rTwoOne~and \rTwoTwo~that 
 respectively makes it become a \minTailWalker~and does not change its state.
\end{proof}

In the following lemma, we prove that \righter~and \potentialMin~are robots that
always consider the right direction. This lemma helps us to prove the correctness of 
Phase \AmITheMin, as well as the correctness of Phase 
\MinWaitToBeKnown.

\begin{lemma} \label{rightDirection}
 If, at a time $t$, a robot is a \righter~or a \potentialMin, then it considers 
 the $right$ direction from the beginning of the execution until the Look phase 
 of time $t$.
\end{lemma}

\begin{proof}
 Robots that are \righter~robots in a configuration $\gamma_{i}$ at time $i$ and
 that are still \righter~in the configuration $\gamma_{i + 1}$, consider the 
 $right$ direction during the move Phase of time $i$ (Rule \rOneEight). 
 Moreover, by Observation~\ref{noMoreRighter} and since initially all the robots
 are \righter~robots and consider the $right$ direction, if a robot is a
 \righter~during the Look phase of a time $t$, this implies that it considers 
 the $right$ direction from the beginning of the execution until the Look phase 
 of time $t$. 
 
 Similarly, robots that are \potentialMin~robots in a configuration $\gamma_{i}$
 at time $i$ and that are still \potentialMin~in the configuration 
 $\gamma_{i + 1}$, consider the $right$ direction during the move Phase of time 
 $i$ (Rule \rOneEight). The only way for a robot to become a \potentialMin~is to
 be a \righter~and to execute Rule \rOneSix. While executing Rule 
 \rOneSix, a \righter~that becomes \potentialMin~does not change the direction 
 it considers. Therefore, by Observations~\ref{noMoreMovingRightAndPLeader} and 
 \ref{noMoreRighter}, and by the arguments of the first paragraph, this implies
 that if a robot is a \potentialMin~during the Look phase of a time $t$, then it
 considers the $right$ direction from the beginning of the execution until the
 Look phase of time $t$.
\end{proof}

Now we prove one of the two main lemmas of this phase: we prove that only 
$r_{min}$ can be aware that it possesses the minimum identifier among all the 
robots of the system. 

\begin{lemma} \label{atMostOneLeader}
 Only $r_{min}$ can become $min$.
\end{lemma}

\begin{proof}
 Assume that there exists a robot $r \neq r_{min}$ that becomes $min$. Assume 
 also that $r$ is the first robot different from $r_{min}$ that becomes $min$.
 By definition of $r_{min}$, $id_{r} > id_{r_{min}}$.
 
 A robot that is a $min$ is a robot such that its variable $state$ is 
 either equal to \minWaitingWalker~or to \minTailWalker. A robot becomes 
 \minTailWalker~only if it executes Rule \rTwoOne. A robot can execute
 Rule \rTwoOne~only if it is a \minWaitingWalker. A robot becomes 
 \minWaitingWalker~only if it executes Rule \rOneOne. Only \righter~robots
 or \potentialMin~robots can execute Rule \rOneOne~(refer to predicate
 $Po\-ten\-tial\-Min\-Or\-Righter()$). Then by 
 Observation~\ref{noMoreMovingRightAndPLeader}, we conclude that each robot can
 execute Rule \rOneOne~at most once. $(*)$ 
 
 In the following, let us consider the different conditions of the predicate 
 $Min\-Disco\-very()$ of Rule \rOneOne~that permits $r$ to become $min$.

 \begin{description}
  \item [Case 1:] \textbf{$\mathbf{r}$ becomes $\mathbf{min}$ 
  because the condition ``$\mathbf{state_{r} = po\-ten\-tial\-Min}$
  $\mathbf{\wedge}$ $\mathbf{\exists r' \in Node\-Mate(),}$ 
  $\mathbf{(state_{r'} = righter}$ $\mathbf{\wedge}$ 
  $\mathbf{id_{r} < id_{r'})}$'' is true.}
  
  The only way for a robot to have its variable $state$ set to \potentialMin~is 
  to execute Rule \rOneSix. This rule is executed when $\mathcal{R} - 1$ 
  \righter~robots are on a same node. Among these $\mathcal{R} - 1$ 
  \righter~robots, the one with the minimum identifier sets its variable $state$
  to \potentialMin~while the other robots set their variables $state$ to
  \dumbSearcher.
  By Observation~\ref{noMoreMovingRightAndPLeader}, a robot that becomes a 
  \dumbSearcher~robot after the execution of Rule \rOneSix~can never become
  \righter~robot or \potentialMin~robot. Moreover, by 
  Observation~\ref{noMoreRighter}, a
  robot that becomes a \potentialMin~can never become a \righter. Since
  $\mathcal{R} - 1$ \righter~are needed to execute Rule \rOneSix, this rule
  can be executed only once during the execution. Therefore if $r$ is a
  \potentialMin, it is necessarily the robot that possesses the minimum 
  identifier among the $\mathcal{R} - 1$ robots that execute Rule \rOneSix.
  Moreover, if there exists a \righter~robot $r'$ when $r$ is \potentialMin,
  this implies that $r'$ has not executed Rule \rOneSix. Hence if
  $id_{r} < id_{r'}$, this necessarily implies that $r = r_{min}$, therefore
  there is a contradiction with the fact that $r \neq r_{min}$. 
  
  \item [Case 2:] \textbf{$\mathbf{r}$ becomes $\mathbf{min}$ because the
  condition ``$\mathbf{\exists r' \in Node\-Mate(),}$
  $\mathbf{id\-Min_{r'} = id_{r}}$'' is true.}
  
  By $(*)$, $r$ is not yet $min$ at the time of its meeting with $r'$. A 
  robot $r'$ can update its variable $id\-Min$ with the identifier (other than
  its) of a robot that is not $min$ only when it executes Rules \rOneFive, 
  \rOneSeven, \rOneNine~or \rOneTen. Among these rules only the rules 
  \rOneSeven~(in the case a \righter~is located with a \dumbSearcher) and
  \rOneNine~permit a robot to update its variable $id\-Min$ with the identifier
  of a robot without copying the value of the variable $id\-Min$ of another
  robot. Therefore at least one of these rules is necessarily executed at a 
  time, since initially the variables $id\-Min$ of the robots are equal to
  $-1$. To execute Rule \rOneSeven~(in the case a \righter~is located with
  a \dumbSearcher) or Rule \rOneNine, a \dumbSearcher~robot must be present
  in the execution. Only the execution of Rule \rOneSix~permits to have 
  \dumbSearcher~robots in the execution. This rule is executed when
  $\mathcal{R} - 1$ \righter~robots are on a same node. The $\mathcal{R} - 1$ 
  robots that execute this rule, set their variables $id\-Po\-ten\-tial\-Min$ to
  the identifier of the robot that becomes \potentialMin~while executing this
  rule. Moreover if a robot is a \dumbSearcher~in a configuration
  $\gamma_{t}$ at time $t$ and is still a \dumbSearcher~in the configuration
  $\gamma_{t + 1}$ then it does not update its variable $id\-Po\-ten\-tial\-Min$ 
  during time $t$ (since it executes Rule \rOneEleven). 
  
  In the case Rule \rOneSeven~is executed because a \righter~$r_{r}$ is 
  located with a \dumbSearcher~$r_{d}$ necessarily 
  $id_{r_{r}} > idPotentialMin_{r_{d}}$, otherwise it is not possible for 
  $r_{r}$ to execute Rule \rOneSeven, since it would have executed Rule 
  \rOneOne~at the same round (since the predicate $Min\-Disco\-very()$ is true 
  because $(state_{r_{d}} \in \{dumb\-Sear\-cher, 
  po\-ten\-tial\-Min\} \wedge id_{r_{r}} < id\-Po\-ten\-tial\-Min_{r_{d}})$).
  Therefore if Rule \rOneSeven~is executed at round $t$ because a 
  \righter~$r_{r}$ is located with a \dumbSearcher~$r_{d}$, this implies, by the
  predicate $Dumb\-Sear\-cher\-Min\-Re\-ve\-la\-tion()$ of Rule \rOneNine, that 
  Rule \rOneNine~is also executed at round $t$. Indeed, $r_{r}$ executes
  Rule \rOneSeven, while $r_{d}$ executes Rule \rOneNine. The reverse is 
  also true: if a \dumbSearcher~$r_{d}$ executes Rule \rOneNine~at round 
  $t$, then necessarily a \righter~$r_{r}$, such that 
  $id_{r_{r}} > idPotentialMin_{r_{d}}$, executes Rule \rOneSeven~at round
  $t$. While executing respectively these rules the two robots update their 
  variables $id\-Min$ with the value of the variable $id\-Po\-ten\-tial\-Min$ of
  the \dumbSearcher. By using the same arguments as the one used in case 1, we
  know that $id\-Po\-ten\-tial\-Min$ is the identifier of $r_{min}$. Therefore
  the variables $id\-Min$ are either set with the identifier of $r_{min}$ while
  Rules \rOneSeven~and \rOneNine~are executed, or copied from another robots
  while Rules \rOneFive~or \rOneTen~are executed. However whatever the rule 
  executed the value of $id\-Min$ is set with the identifier of $r_{min}$. 
  
  \item [Case 3:] \textbf{$\mathbf{r}$ becomes $\mathbf{min}$ because the 
  condition ``$\mathbf{\exists r' \in Node\-Mate(),}$ 
  $\mathbf{(state_{r'} \in \{dumb\-Sear\-cher,}$ $\mathbf{po\-ten\-tial\-Min\}}$ 
  $\mathbf{\wedge}$ $\mathbf{id_{r} < id\-Po\-ten\-tial\-Min_{r'}})$'' is true.}  
  
  Only the execution of Rule \rOneSix~permits to have \dumbSearcher~or 
  \potentialMin~in the execution. This rule is executed when $\mathcal{R} - 1$ 
  \righter~robots are on a same node. When executing this rule, the 
  $\mathcal{R} - 1$ robots set their variables $id\-Po\-ten\-tial\-Min$ to the 
  identifier of the robot that possesses the minimum identifier among them.
  Moreover among the $\mathcal{R} - 1$ robots that execute Rule \rOneSix,
  one robot becomes \potentialMin~while the other become \dumbSearcher. Besides 
  if a robot is a \dumbSearcher~(resp. a \potentialMin) in a configuration 
  $\gamma_{t}$ at time $t$ and is still a \dumbSearcher~(resp. a \potentialMin)
  in the configuration $\gamma_{t + 1}$ then it does not update its variable 
  $id\-Po\-ten\-tial\-Min$ during time $t$ since it executes Rule 
  \rOneEleven~(resp. \rOneEight). As Rule \rOneSix~can only
  be executed once (see the arguments of case 1), if $r$ meets a
  \dumbSearcher~or a \potentialMin~$r'$, such that 
  $id_{r} < id\-Po\-ten\-tial\-Min_{r'}$, this necessarily implies that $r'$ is 
  issued of the execution of Rule \rOneSix~while $r$ has not executed this
  rule, and therefore $r = r_{min}$, which is a contradiction. 
 
  \item [Case 4:] \textbf{$\mathbf{r}$ becomes $\mathbf{min}$ because 
  $\mathbf{right\-Steps_{r} = 4 * id_{r} * n}$.}
  
  At the time where $r$ becomes $min$, $r_{min}$ is either a \righter~robot, a
  \potentialMin~robot or $min$, otherwise this implies that there already exists
  a $min$ (other than $r_{min}$) in the execution, which is a contradiction with 
  the fact that $r$ is the first robot different from $r_{min}$ that becomes 
  $min$.
  
  By the predicate $Po\-ten\-tial\-Min\-Or\-Righter()$ of Rule \rOneOne, only 
  \righter~robots or \potentialMin~robots can become $min$. By 
  Lemma~\ref{rightDirection}, if, at a time $t$, a robot is a \righter~or a 
  \potentialMin, then it considers the $right$ direction from the beginning of 
  the execution until the Look phase of time $t$.
  Robots that are \righter~robots or \potentialMin~robots in a configuration
  $\gamma_{t}$ at time $t$ and that are either \righter~or \potentialMin~in the
  configuration $\gamma_{t + 1}$ increase from 1 their variables $right\-Steps$
  each time an adjacent edge in the right direction to their positions is 
  present (Rules \rOneSix~and \rOneEight). Therefore, by the predicate
  $Min\-Discovery()$ of Rule \rOneOne~a robot $r"$ moves at most during
  $4 * id_{r"} * n$ steps in the right direction before being $min$.
  
  By Lemma~\ref{LeaderTheWholeExecution}, from the time a robot becomes $min$, 
  it is either a \minWaitingWalker~or a \minTailWalker. Therefore it can only 
  execute Rules \TerminationOne, \TerminationTwo, \rTwoOne, \rTwoTwo, 
  \rThreeOne~and \rFourThree. This 
  implies that once a robot is $min$, it considers only either the $right$ or 
  the $\bot$ direction, and can move during at most $n$ steps in the right 
  direction before stopping to move definitively (by executing the following 
  rules in the order: \rTwoTwo, \rTwoOne, \rThreeOne~and \rFourThree). Therefore
  by the previous paragraph, a $min$ $r"$ considers the right or the $\bot$
  direction from the beginning of the execution until the end of the execution, 
  and can move during at most $4 * id_{r"} * n + n$ steps in the right direction
  during the whole execution.
  
  Because of the dynamism of the ring, by 
  Observation~\ref{noMoreMovingRightAndPLeader} and since when a \righter~or a 
  \potentialMin~robot stops to be a \righter~or a \potentialMin~robot, it 
  stops to update the value of its variable $right\-Steps$, we have:  
  $\forall r_{1}, r_{2} \in \mathcal{R}^{2}, 
  state_{r_{1}}, state_{r_{2}} \in \{righter, po\-ten\-tial\-Min\}^{2},
  |right\-Steps_{r_{1}} - right\-Steps_{r_{2}}| \leq n$. 
  
  Because it takes one round for a robot to update its variable $state$ to
  $min$, a \righter~or a \potentialMin~can be located with a robot $r$ just the 
  round before $r$ becomes $min$. Therefore this \righter~or \potentialMin~can
  move again in the right direction during at most $n$ steps without meeting the 
  $min$.
  
  We know that $id_{r_{min}} < id_{r}$, therefore we have
  $4 * id_{r_{min}} * n + n + n + n < 4 * id_{r} * n$. Hence there exists a time
  at which $r$ meets $r_{min}$ while $r_{min}$ is $min$ and $r$ is not yet 
  $min$. At this time, by the rules of \Gathering, $r$ stops being a \righter~or
  a \potentialMin~robot, and hence by 
  Observation~\ref{noMoreMovingRightAndPLeader}, $r$ cannot be anymore a
  \righter~robot or a \potentialMin~robot and therefore it cannot become $min$,
  which leads to a contradiction.
 \end{description}
\end{proof}

The following lemma helps us to prove the Lemma~\ref{atLeastOneLeader}. This 
lemma is true only if there is no $min$ in the execution. In other words, it 
is true only if all the robots are executing Phase \AmITheMin.

\begin{lemma} \label{NotBotDirection}
 If there is no $min$ in the execution, if, at time $t$, a robot $r$ is such 
 that $state_{r} \in \{dumb\-Searcher, aware\-Searcher\}$, then, during the Move
 phase of time $t - 1$, it does not consider the $\bot$ direction.
\end{lemma}

\begin{proof}
 Consider a robot $r$ such that, at time $t$,
 $state_{r} \in \{dumb\-Searcher, aware\-Searcher\}$.
 
 While executing \Gathering, since initially all
 the robots are \righter, if there is no $min$, only \righter, \potentialMin, 
 \dumbSearcher~and \awareSearcher~robots can be present in the execution. 
 
 Consider then the two following cases.
 
 \begin{description}
  \item [Case 1:] \textbf{At time $\mathbf{t - 1}$, $\mathbf{r}$ is neither a 
  \textit{dumb\-Searcher} nor an \textit{aware\-Searcher}.}
  
  Whatever the state of $r$ at time $t - 1$ (\righter~or \potentialMin), to have 
  its variable state at time $t$ equals either to \dumbSearcher~or to
  \awareSearcher, $r$ executes at time $t - 1$ either Rule \rOneFive, 
  \rOneSix~or \rOneSeven.
  
  Consider first the case where $r$ executes Rule \rOneSix~at time $t - 1$. 
  Only \righter~robots can execute Rule \rOneSix. While executing Rule
  \rOneSix, $r$ becomes a \dumbSearcher~(since while executing this rule a 
  \righter~can become either a \dumbSearcher~or a \potentialMin). Moreover,
  while executing Rule \rOneSix, a \righter~that becomes \dumbSearcher~does 
  not change the direction it considers. By Lemma~\ref{rightDirection}, during
  the Look phase of time $t - 1$, $r$ considers the right direction and since 
  $r$ does not change its direction during the Compute phase of time $t - 1$,
  this implies that the lemma is proved in this case. 
 
  Consider now the case where $r$ executes either Rule \rOneFive~or 
  \rOneSeven. While executing these rules the function \textsc{Search} is 
  called.
  
  While executing the function \textsc{Search}, if there are multiple robots on
  the current node of $r$ at time $t - 1$, it considers either the right or the
  left direction. Therefore, in this case the lemma is proved.
 
  In the case $r$ is alone on its node at time $t - 1$, while executing the
  function \textsc{Search} it does not change its direction. Moreover, while 
  executing Rules \rOneFive~or \rOneSeven, before calling the function 
  \textsc{Search} the robot calls the function \textsc{BecomeAwareSearcher} that
  sets its direction to the right direction. Therefore, in these cases, even if 
  $r$ is alone on its node, it considers a direction different from $\bot$
  during the Move phase of time $t - 1$, hence the lemma is proved.
  
  \item [Case 2:] \textbf{At time $\mathbf{t - 1}$, $\mathbf{r}$ is a
  \textit{dumb\-Searcher} or an \textit{aware\-Searcher}.}

  Whatever the state of $r$ at time $t - 1$ (\dumbSearcher~or \awareSearcher),
  to have its variable state at time $t$ equals either to \dumbSearcher~or to
  \awareSearcher, $r$ executes at time $t - 1$ either Rule \rOneNine, 
  \rOneTen~or \rOneEleven. While executing these rules the function 
  \textsc{Search} is called.
 
  As highlighted in the case 1, if there are multiple robots on the current node
  of $r$ at time $t - 1$, the lemma is proved.
 
  Moreover, while executing Rules \rOneNine~and \rOneTen, before calling the
  function \textsc{Search} the robot calls the function
  \textsc{BecomeAwareSearcher} that sets its direction to the right direction. 
  Therefore, in these cases, even if $r$ is alone on its node, it considers a 
  direction different from $\bot$ during the Move phase of time $t - 1$, hence
  the lemma is proved.
 
  It remains the case where $r$ executes Rule \rOneEleven~at time
  $t - 1$ while it is alone on its node. In this case, while executing Rule 
  \rOneEleven, $r$ does not change its direction (refer to the function 
  \textsc{Search}). Since at time $t - 1$, $r$ is already a \dumbSearcher~or an
  \awareSearcher, and since initially all the robots are \righter, by recurrence
  on all the cases treated previously (Case 1 and 2), the direction $r$
  considers during the Move phase of time $t - 1$ cannot be equal to $\bot$. 
 \end{description}
\end{proof}

Finally, we prove the other main lemma of this phase: we prove that $r_{min}$ is 
aware, in finite time, that it possesses the minimum identifier among all the
robots of the system. 

\begin{lemma} \label{atLeastOneLeader}
 In finite time $r_{min}$ becomes $min$.
\end{lemma}

\begin{proof}
 Assume that $r_{min}$ does not become $min$. By Lemma~\ref{atMostOneLeader},
 only $r_{min}$ can be $min$. While executing \Gathering, since initially all
 the robots are \righter, if there is no $min$, only \righter, \potentialMin, 
 \dumbSearcher~and \awareSearcher~robots can be present in the execution. 
 
 Initially all the robots are \righter. In the case where there is no $min$ in
 the execution, by the rules of \Gathering, from a configuration $\gamma_{t}$ at
 a time $t$ where there are only \righter~robots, it is not possible to have
 \awareSearcher~in the configuration $\gamma_{t + 1}$. A robot can become a 
 \dumbSearcher~or a \potentialMin~only when Rule \rOneSix~is executed. This
 rule is executed when $\mathcal{R} - 1$ \righter~robots are on a same node
 (refer to predicate $All\-But\-One\-Righter()$). 
 
 Let us now consider the three following cases that can occur in the execution.
 
 \begin{description}
  \item [Case 1:] \textbf{Rule \rOneSix~is never executed.}
 
   In this case all the robots are \righter~robots during the whole execution, 
   and execute therefore Rule \rOneEight~at each instant time. While
   executing Rule \rOneEight, a robot always considers the $right$ direction
   and increments its variable $rightSteps$ by one each time there exists an 
   adjacent $right$ edge to its location. Since by assumption $r_{min}$ does
   not become $min$, then by Rule \rOneOne~and predicate $Min\-Disco\-ve\-ry()$,
   $r_{min}$ cannot succeed to have its variable $right\-Steps$ equals to 
   $4 * id_{r_{min}} * n$, otherwise the lemma is true. Therefore it exists a 
   time at which $r_{min}$ is on a node such that its adjacent $right$ edge is 
   missing forever. Since it can exist at most one eventual missing edge in a
   \COT ring, and since all the robots always move in the $right$ 
   direction when there is an adjacent right edge to their location (since they
   execute Rule \rOneEight), it exists a time 
   at which $\mathcal{R} - 1$ \righter~robots are on a same node, cases 2 and 3 
   are then considered.
  
  \item [Case 2:] \textbf{Rule \rOneSix~is executed but $\mathbf{r_{min}}$ 
  is not among the $\mathcal{R} - 1$ $\mathbf{righter}$ robots that execute it.}
  
   While executing Rule \rOneSix, among the $\mathcal{R} - 1$
   \righter~located on a same node that execute this rule, the robot with the 
   minimum identifier $r_{p}$ becomes \potentialMin~while the other robots
   become 
   \dumbSearcher, and all update their variables $id\-Po\-ten\-tial\-Min$ to
   $id_{r_{p}}$. By definition we have $id_{r_{p}} > id_{r_{min}}$. By 
   Observation~\ref{noMoreMovingRightAndPLeader}, a robot that becomes a
   \dumbSearcher~can never become \righter~robot or \potentialMin~robot. 
   Moreover, by Observation~\ref{noMoreRighter}, a robot
   that becomes a \potentialMin~can never become a \righter. Since 
   $\mathcal{R} - 1$ \righter~are needed to execute Rule \rOneSix, this rule
   can be executed only once. Note that if a robot is a \dumbSearcher~(resp. a 
   \potentialMin) in a configuration $\gamma_{t}$ at time $t$ and is still a
   \dumbSearcher~(resp. a \potentialMin) in the configuration $\gamma_{t + 1}$ 
   then it does not update its variable $id\-Po\-ten\-tial\-Min$ during time $t$
   since it executes Rule \rOneEleven~(resp. \rOneEight)
   
   At the time of the execution of Rule \rOneSix, $r_{min}$ is a \righter, 
   since it is not among the robots that execute this rule. After the execution 
   of this rule $r_{min}$, as a \righter, cannot meet a \potentialMin~robot.
   Indeed the only way for a robot to become \potentialMin~is to execute
   Rule \rOneSix. Therefore only $r_{p}$ can be \potentialMin, and we know that
   $idPotentialMin_{r_{p}} = id_{r_{p}} > id_{r_{min}}$. Hence if $r_{min}$ 
   meets a \potentialMin, then by Rule \rOneOne~and predicate 
   $Min\-Disco\-ve\-ry()$ the lemma is true, which is a contradiction.
   
   Similarly, $r_{min}$ as a \righter~cannot meet a \dumbSearcher~$r_{d}$. 
   Indeed, only Rule \rOneSix~permits a robot to become a \dumbSearcher. 
   Therefore, since $idPotenialMin_{r_{d}} = id_{r_{p}} > id_{r_{min}}$, if
   $r_{min}$ meets a \dumbSearcher, then by Rule \rOneOne~and predicate
   $Min\-Disco\-ve\-ry()$ the lemma is true, which is a contradiction.
   
   Moreover it cannot exist \awareSearcher~in this execution. Indeed, as said
   previously, from a configuration $\gamma_{t}$ at a time $t$ where there are 
   only \righter~robots, it is not possible to have \awareSearcher~in the 
   configuration $\gamma_{t + 1}$. Therefore \awareSearcher~can be present in 
   the execution only after the execution of Rule \rOneSix. In the case 
   where
   there is not yet \awareSearcher, a robot can become an \awareSearcher~only if
   a \righter~meets a \dumbSearcher~(Rules \rOneNine~and \rOneSeven). However 
   after the execution of Rule \rOneSix, only $r_{min}$ is a \righter, and 
   as explained in the previous paragraph, if $r_{min}$ as a \righter~meets a
   \dumbSearcher~there is a contradiction. 
  
   Since there is no \awareSearcher~and since $r_{min}$ as a \righter~cannot 
   meet neither \potentialMin~nor \dumbSearcher, this implies that $r_{min}$
   stays a \righter~during the whole execution and therefore executes Rule 
   \rOneEight~at each instant time. By the same arguments as the one used 
   in case 1, necessarily it exists a time at which $r_{min}$ is on node such 
   that its adjacent $right$ edge is missing forever, otherwise the lemma is 
   true. However since there is no $min$ is the
   execution, and there is no \awareSearcher, $r_{p}$ stays a \potentialMin~and
   executes Rule \rOneEight~at each instant time, therefore it always
   considers the $right$
   direction. Since it can only exist one eventual missing edge and since this
   edge is the adjacent $right$ edge to the position of $r_{min}$, all the other
   edges are infinitely often present. Therefore, in finite time, the 
   \potentialMin~is located on the same node as $r_{min}$, which is a 
   contradiction.
 
  \item [Case 3:] \textbf{Rule \rOneSix~is executed and $\mathbf{r_{min}}$ 
  is among the $\mathcal{R} - 1$ $\mathbf{righter}$ robots that execute it.}
  
   We use the same arguments as the one used in case 2. Therefore we know that 
   while executing Rule \rOneSix, $r_{min}$ becomes \potentialMin, since 
   $r_{min}$ possesses the minimum identifier among all the robots of the 
   system. 
   
   Moreover, since $r_{min}$ does not become $min$, as a \potentialMin, it
   cannot meet a \righter~robot otherwise by Rule \rOneOne~and predicate 
   $Min\-Disco\-ve\-ry()$ the lemma is true. 
   
   Similarly, $r_{min}$ as a \potentialMin~cannot meet \awareSearcher. Indeed in
   the case there is not yet \awareSearcher, a robot can become an 
   \awareSearcher~only if a \righter~meets a \dumbSearcher~(Rules \rOneNine~and
   \rOneSeven). While executing these rules a robot that becomes an
   \awareSearcher~sets its variable $id\-Min$ to the identifier of the variable
   \potentialMin~of the \dumbSearcher, which is in this case $id_{r_{min}}$. An
   \awareSearcher~never
   updates the value of its variable $id\-Min$. Once there is at least
   one \awareSearcher~in the execution, it is possible to have other robots that
   become \awareSearcher~thanks to the execution of Rule \rOneTen. However
   while executing this rule, a robot that becomes \awareSearcher~copies the 
   value of the variable $id\-Min$ of the 
   \awareSearcher~it is located with. Therefore if $r_{min}$, as a
   \potentialMin, meets an 
   \awareSearcher, by Rule \rOneOne~and predicate $Min\-Disco\-ve\-ry()$, the 
   lemma is true, which is a contradiction. 
   
   Therefore, as a \potentialMin, $r_{min}$ executes Rule \rOneEight~at each
   instant time. By the same arguments as the one used in case 1, necessarily it
   exists a time at which $r_{min}$ is on node such that its adjacent $right$ 
   edge is missing forever, otherwise the lemma is true.  

   By Observation~\ref{noMoreMovingRightAndPLeader}, \dumbSearcher~and 
   \awareSearcher~robots cannot become \righter~or \potentialMin. As
   explained, if there is no meeting between a \dumbSearcher~robot and a
   \righter~robot, it cannot exist \awareSearcher~robots in the execution.
   As seen previously, no \righter~robot can meet $r_{min}$. At the time
   where Rule 
   \rOneSix~is executed there is a \righter~robot $r$ in the execution. In
   the case $r$ never meets a \dumbSearcher~robot, it executes Rule
   \rOneEight~at each instant time. Hence, using the arguments as the one used 
   in case 2, in 
   finite time, $r$ can be located on the same node as $r_{min}$, which is a 
   contradiction. This implies that there exists a time at which $r$, as a
   \righter~robot, meets at least a \dumbSearcher~robot $r'$. In this case $r$
   executes Rule \rOneSeven~(refer to the predicate
   $Righter\-With\-Sear\-cher()$) and all the \dumbSearcher~robots located with 
   $r$ including $r'$ execute Rule \rOneNine~(by the predicate 
   $Dumb\-Sear\-cher\-Min\-Re\-ve\-la\-tion()$ and since 
   $id_{r} > id_{r_{min}}$). Hence $r$ and all the 
   \dumbSearcher~robots located with $r$ become \awareSearcher~robots and
   execute the function \textsc{Search}. When a robot executes the function 
   \textsc{Search} while there are multiple robots on its node, if it possesses 
   the maximum identifier among the robots of its node, it considers the left
   direction, otherwise it considers the right direction. Therefore, once 
   \rOneSeven~and \rOneNine~are executed, there are at least two
   \awareSearcher~considering two opposite directions. Moreover once 
   \rOneSeven~and \rOneNine~are executed, except
   $r_{min}$ there are only \dumbSearcher~and \awareSearcher~robots in the 
   execution. When a \dumbSearcher~robot meets an \awareSearcher~robot, it 
   executes Rule \rOneTen~and therefore becomes \awareSearcher~robot and 
   executes the 
   function \textsc{Search}. An \awareSearcher~executes Rule \rOneEleven~at 
   each instant time, therefore it calls the function \textsc{Search} at each
   instant time. While executing the function \textsc{Search}, if an 
   \awareSearcher~robot 
   is alone on its node, it considers the last direction it considers (this
   direction cannot be equal to $\bot$ by Lemma~\ref{NotBotDirection}). All this
   implies that in finite time an \awareSearcher~robot is located on the same
   node as $r_{min}$. Therefore by Rule \rOneOne~and
   predicate $Min\-Disco\-ve\-ry()$, $r_{min}$ becomes $min$.
 \end{description} 
\end{proof}

By Lemmas~\ref{atMostOneLeader} and \ref{atLeastOneLeader}, we can deduce 
the following corollary which proves the correctness of Phase \AmITheMin.

\begin{corollary} \label{uniqueLeader}
 Only $r_{min}$ becomes $min$ in finite time.
\end{corollary}

\subsubsection{Proofs of Correctness of Phase \MinWaitToBeKnown}

Once $r_{min}$ completes Phase \AmITheMin, it stops to move and waits for the 
completion of Phase \MinWaitToBeKnown. We recall that, during Phase 
\MinWaitToBeKnown~of 
\Gathering, $\mathcal{R} - 3$ robots must join $r_{min}$ on the node where it is 
waiting. More precisely, while executing \Gathering, Phase \MinWaitToBeKnown~is
achieved when $\mathcal{R} - 3$ \waitingWalker~robots are located on the node 
where $r_{min}$, as $min$, is waiting. In the previous subsection, we prove that, in finite 
time, only $r_{min}$ becomes $min$ (Corollary~\ref{uniqueLeader}) and that once 
$r_{min}$ is $min$ it stays $min$ for the rest of the execution (Lemma~\ref{LeaderTheWholeExecution}).
Note that, by the rules of \Gathering, the
$min$ is necessarily a \minWaitingWalker~robot before being a \minTailWalker~(since
only a \minWaitingWalker~can become a \minTailWalker~while executing Rule \rTwoOne). 
Moreover, by Rule \rTwoTwo, $r_{min}$, as a \minWaitingWalker, does not move 
until $\mathcal{R} - 3$ \waitingWalker~robots are on its node.
Therefore, 
as \minWaitingWalker, $r_{min}$ is, as expected, always on the same node. Let $u$ be the node
on which $r_{min}$, as a \minWaitingWalker, is located. Let $t_{min}$ be the
time at which $r_{min}$ becomes a \minWaitingWalker~robot. In this subsection,
we consider the execution from time $t_{min}$.

To simplify the proofs, we introduce the notion of \towerElection~as follows.

\begin{definition} [\towerElection]
 A \towerElection~corresponds to a configuration of the execution in which 
 $\mathcal{R} - 3$ \waitingWalker~robots are located on the same node as 
 the \minWaitingWalker. 
\end{definition}

To prove the correctness of Phase \MinWaitToBeKnown, we hence have 
to prove that, in finite time, a \towerElection~is formed.

As noted previously, by the rules of \Gathering, as long as there is no 
\towerElection, $r_{min}$ stays a \minWaitingWalker~robot. 

The following observation is useful to prove the correctness of this phase.

\begin{observation} \label{noMoreWaitingWalk}
 There exists no rule in \Gathering~permitting a robot that stops being either
 \minWaitingWalker~or \waitingWalker~robot to be again a \minWaitingWalker~or 
 \waitingWalker~robot.
\end{observation}

To prove the correctness of this phase, we prove, first, that 
if a \potentialMin~is present in the execution then, in finite time, a 
\towerElection~is present in the execution, next, we prove that if there
is no \potentialMin~in the execution then, in finite time, a \towerElection~is
also present in the execution. We prove this respectively in 
Lemmas~\ref{PotentialMinAndTowerElection} and 
\ref{NoPotentialMinAndTowerElection}. To simplify the proofs of these two 
lemmas, we need to prove the nine following lemmas. 

In the following lemma we prove that it can exist at most one \towerElection~in
the whole execution.

\begin{lemma} \label{uniqueTowerElection}
 It can exist at most one \towerElection~in the whole execution. 
\end{lemma}

\begin{proof}
 By definition a \towerElection~is composed of one \minWaitingWalker~and
 $\mathcal{R} - 3$ \waitingWalker~robots. Once a \towerElection~is formed, the 
 $\mathcal{R} - 2$ ($\mathcal{R} - 2 \geq 2$) robots involved in the 
 \towerElection~execute Rule \rTwoOne. While 
 executing this rule the robot with the maximum identifier among the 
 $\mathcal{R} - 2$ robots involved in the \towerElection~becomes 
 \headWalker~while the \minWaitingWalker~becomes \minTailWalker~and 
 the other robots involved in the \towerElection~become \tailWalker.
 
 Then by Observation~\ref{noMoreWaitingWalk} and since by
 Corollary~\ref{uniqueLeader} only $r_{min}$ can be \minWaitingWalker, the lemma
 is proved.
\end{proof}

In the following lemma, we prove that all the \waitingWalker~as well as the
\minWaitingWalker~are located on node $u$ and do not move. This is important to
prove that, in finite time, a \towerElection~is formed.

\begin{lemma} \label{waitingCautiousWalkOnSameNode}
 All \waitingWalker~robots are located on the same node as $r_{min}$ when
 $state_{r_{min}} = min\-Waiting\-Walker$ and neither the \waitingWalker~robots
 nor $r_{min}$, as a \minWaitingWalker, move.
\end{lemma}

\begin{proof}
 By the rules of \Gathering, as long as there is no \towerElection, $r_{min}$ is
 \minWaitingWalker. While $r_{min}$ is the \minWaitingWalker, it executes
 Rule \rTwoTwo~at each instant time. While executing this rule, $r_{min}$ 
 considers the $\bot$ direction and therefore does not move.

 Only Rule \rTwoThree~permits a robot
 $r$ to become a \waitingWalker~robot. For this rule to be executed $r$ must be
 located with a \minWaitingWalker~(refer to predicate 
 $Po\-ten\-tial\-Min\-Or\-Sear\-cher\-With\-Min()$). By 
 Corollary~\ref{uniqueLeader}, only $r_{min}$ can be \minWaitingWalker. While 
 executing Rule \rTwoThree, $r$ considers the $\bot$ direction and therefore 
 at the time of the execution of this rule, $r$ does not move and is on the node
 where $r_{min}$, as a \minWaitingWalker, is located. 
 
 While $r$ is a \waitingWalker~robot, as long as there is no 
 \towerElection~in the execution, it executes Rule \rTwoTwo~at each instant
 time. Therefore $r$ does not move. As noted previously, the location where $r$ 
 stops moving is the location where $r_{min}$, as the \minWaitingWalker, is 
 located.
 
 Once a \towerElection~is present in the execution the \waitingWalker~robots
 and the \minWaitingWalker~composing this \towerElection~execute Rule
 \rTwoOne. While executing this rule the robots do not change the direction they 
 consider and stop being \waitingWalker/\minWaitingWalker~robots. Therefore, by 
 Observation~\ref{noMoreWaitingWalk} and since by Corollary~\ref{uniqueLeader}
 only $r_{min}$ can be \minWaitingWalker, all \waitingWalker~robots are located
 on the same node as $r_{min}$ when $state_{r_{min}} = min\-Waiting\-Walker$ and
 neither the \waitingWalker~robots nor $r_{min}$, as a \minWaitingWalker, move.
\end{proof}

%

Now we prove a property on \potentialMin.

\begin{lemma} \label{OnePotentialMin}
 It can exist at most one \potentialMin~robot in the whole execution.
\end{lemma}

\begin{proof}
 Only the execution of Rule \rOneSix~permits a robot to become a 
 \potentialMin~robot. Rule \rOneSix~is executed when 
 $\mathcal{R} - 1$ \righter~robots are located on a same node. When these
 $\mathcal{R} - 1$ \righter~robots execute Rule \rOneSix, one becomes a
 \potentialMin, and the others become \dumbSearcher. Therefore, by 
 Observations~\ref{noMoreMovingRightAndPLeader} and \ref{noMoreRighter} this
 rule can be executed only once. Moreover, by the rules of \Gathering, once a
 \potentialMin~stops to be a \potentialMin, it cannot be again a \potentialMin.
 Hence the lemma is proved.
\end{proof}

The following lemma demonstrates a property on $min$.

\begin{lemma} \label{OnlyMovingRightAndPLeaderCanBeMin}
 Before being $min$, $r_{min}$ is either a \righter~robot or a
 \potentialMin~robot.
\end{lemma}

\begin{proof}
 A robot that is a $min$ is a robot such that its variable $state$ is either 
 equal to \minWaitingWalker~or to \minTailWalker. The only way to be a
 \minTailWalker~robot is to be a \minWaitingWalker~robot and to execute
 Rule \rTwoOne. The only way to be a \minWaitingWalker~is to execute Rule 
 \rOneOne. Only \righter~robots or \potentialMin~robots can execute Rule
 \rOneOne~(refer to predicate $Po\-ten\-tial\-Min\-Or\-Righter()$). 
\end{proof}

The three following lemmas give properties on the execution, when $r_{min}$ is 
$min$. Indeed, they indicate the presence or absence of 
\righter/\potentialMin~in the execution while $r_{min}$ is $min$.

\begin{lemma} \label{OnlyOnePLeaderOrOneMovingRight}
 In the suffix of the execution starting from the time where $r_{min}$ is $min$,
 it is not possible to have a \potentialMin~robot and a \righter~robot present 
 at the same time.
\end{lemma}

\begin{proof}
 By Lemma~\ref{OnlyMovingRightAndPLeaderCanBeMin}, $r_{min}$ is either a 
 \righter~or a \potentialMin~before being $min$. In the case where $r_{min}$ is
 a \potentialMin~before being $min$, then by Lemma~\ref{OnePotentialMin}, it 
 cannot exist a \potentialMin~in the execution after $r_{min}$ becomes
 $min$. Therefore the lemma is proved in this case. 
 
 Consider now the case where $r_{min}$ is a \righter~before being $min$. For a 
 robot to become a \potentialMin~Rule \rOneSix~must be executed. This rule
 is executed when $\mathcal{R} - 1$ \righter~are located on a same node. While 
 executing Rule \rOneSix, among the $\mathcal{R} - 1$ \righter~located on a
 same node, the one with the minimum identifier becomes \potentialMin~while the
 others become \dumbSearcher. By Observation~\ref{noMoreRighter}, $r_{min}$
 cannot be among the $\mathcal{R} - 1$ \righter~that execute Rule \rOneSix,
 otherwise it cannot be a \righter~before being $min$. Similarly thanks to 
 Observation~\ref{noMoreRighter}, the $\mathcal{R} - 1$ robots that execute
 Rule \rOneSix, cannot be \righter~anymore after the execution of this rule.
 therefore, it is not possible to have a \potentialMin~and a \righter~in the 
 execution once $r_{min}$ is $min$.
\end{proof}

%

\begin{lemma} \label{righterOtherAndMinNoMorePotentialMin}
 If there exists a time $t$ at which a \righter, a robot $r$ ($r \neq r_{min}$)
 such that $state_{r} \neq righter$ and $r_{min}$, as $min$, are present in the
 execution, then there is no more \potentialMin~in the suffix of the execution 
 starting from $t$.
\end{lemma}

\begin{proof}
 By Lemma~\ref{OnlyOnePLeaderOrOneMovingRight}, since there is a \righter~at 
 time $t$, there is no \potentialMin~in the execution at time $t$.

 Since at time $t$, $r_{min}$ and $r$ are not \righter~and can never be 
 \righter~anymore (refer to Observation~\ref{noMoreRighter}), it is not possible
 to have $\mathcal{R} - 1$ \righter~located on a same node after time $t$. 
 However, in order to have a \potentialMin~in the execution, Rule 
 \rOneSix~must be executed. This rule is executed only if $\mathcal{R} - 1$ 
 \righter~are located on a same node. Therefore there is no \potentialMin~in the
 execution after time $t$.
\end{proof}

\begin{lemma} \label{potentialMinAndMinNoMoreRighter}
 If there is a \potentialMin~at a time $t$, and if before being $min$, $r_{min}$
 is a \righter, then there is no more \righter~in the suffix of the execution
 starting from time $t' = max\{t, t_{min}\}$.  
\end{lemma}

\begin{proof}
 Assume that before being $min$, $r_{min}$ is a \righter. Moreover assume that
 there is a \potentialMin~in the execution at time $t$. 
 
 $(*)$ For a robot to become a \potentialMin~Rule \rOneSix~must be executed.
 This rule is executed when $\mathcal{R} - 1$ \righter~are located on a same 
 node. While executing Rule \rOneSix, among the $\mathcal{R} - 1$ 
 \righter~located on a same node, the one with the minimum identifier becomes
 \potentialMin~while the others become \dumbSearcher. By 
 Observation~\ref{noMoreRighter} none of these $\mathcal{R} - 1$ robots can 
 become \righter~anymore after time $t$. 
 
 Consider then the two following cases.
 
 \begin{description}
  \item [Case 1:] \textbf{$\mathbf{t > t_{min}}$.}
 
  By Observation~\ref{noMoreRighter}, $r_{min}$ cannot be a \righter~after time
  $t_{min}$. Therefore $r_{min}$ is not among the $\mathcal{R} - 1$ robots that 
  execute Rule \rOneSix, and hence, by $(*)$, after time $t$, there is no
  more \righter~in the execution.
  
  \item [Case 2:] \textbf{$\mathbf{t \leq t_{min}}$.}
  
  By $(*)$, $r_{min}$ cannot be among the $\mathcal{R} - 1$ \righter~that 
  execute Rule \rOneSix, otherwise it cannot be a \righter~before being 
  $min$. Therefore, by $(*)$ and since after time $t_{min}$, by 
  Observation~\ref{noMoreRighter}, $r_{min}$ cannot be a \righter~anymore, there
  is no more \righter~in the execution after time $t_{min}$.
 \end{description}
\end{proof}

%
%
%
%

The following lemma is an extension of Lemma~\ref{NotBotDirection}. While 
Lemma~\ref{NotBotDirection} is true only when all the robots are executing 
Phase \AmITheMin, the following lemma is true whether the robots are 
executing Phase \AmITheMin~or Phase \MinWaitToBeKnown.

\begin{lemma} \label{lemma4Bis}
 If there is no \towerElection~in the execution, if, at time $t$, a robot $r$ is 
 such that 
 $state_{r} \in \{po\-ten\-tial\-Min, dumb\-Searcher, aware\-Searcher\}$, then, 
 during the Move phase of time $t - 1$, it does not consider the $\bot$ 
 direction.
\end{lemma}

\begin{proof}
 Consider a robot $r$ such that at time $t$, $state_{r} = po\-ten\-tial\-Min$.
 By Lemma~\ref{rightDirection}, $r$ considers the right direction during the 
 Move phase of time $t - 1$. Hence the lemma is proved in this case. 
  
 Consider now a robot $r$ such that, at time $t$, 
 $state_{r} \in \{dumb\-Searcher, aware\-Searcher\}$. Since there is no 
 \towerElection~in the execution, by the rules of \Gathering~and knowing that 
 initially all the robots are \righter, there are only \righter, \potentialMin,
 \dumbSearcher, \awareSearcher, \waitingWalker~and \minWaitingWalker~robots in
 the execution. Note that there is no rule in \Gathering~permitting a 
 \waitingWalker~or a \minWaitingWalker~to become either a \dumbSearcher~or an
 \awareSearcher. 
 
 Consider then the two following cases.
 
 \begin{description}
  \item [Case 1:] \textbf{At time $\mathbf{t - 1}$, $\mathbf{r}$ is neither a 
  \textit{dumb\-Searcher} nor an \textit{aware\-Searcher}.}
 
  Whatever the state of $r$ at time $t - 1$ (\righter~or
  \potentialMin), to have its variable state at time $t$ equals either to 
  \dumbSearcher~or to \awareSearcher, $r$ executes at time $t - 1$ either
  Rule \rTwoFour, \rOneFive, \rOneSix~or \rOneSeven.
 
  When a robot executes Rule \rTwoFour, it calls the function
  \textsc{BecomeAwareSearcher}. When a robot executes the function 
  \textsc{BecomeAwareSearcher}, it sets its direction to the $right$ direction, 
  therefore the lemma is also true in this case.
  
  Then, we can use the arguments of the proof of Lemma~\ref{NotBotDirection} to
  prove that the current lemma is true for the remaining cases. Indeed, even if
  in Lemma~\ref{NotBotDirection} the context is such that there is no $min$ in 
  the execution, the arguments used in its proof are still true in the context
  of the current lemma.
  
  \item [Case 2:] \textbf{At time $\mathbf{t - 1}$, $\mathbf{r}$ is a
  \textit{dumb\-Searcher} or an \textit{aware\-Searcher}.}
  
  Whatever the state of $r$ at time $t - 1$ (\dumbSearcher~or \awareSearcher),
  to have its variable state at time $t$ equals either to \dumbSearcher~or to
  \awareSearcher, $r$ executes at time $t - 1$ either Rule \rOneNine, 
  \rOneTen~or \rOneEleven. Similarly as for the case 1, we can use the arguments
  of the proof of Lemma~\ref{NotBotDirection} to prove that the current 
  lemma is true in these cases.
 \end{description}
\end{proof}

The following lemma proves that in the case where there are at least 3 
robots in the execution such that they are either \potentialMin,
\dumbSearcher~or \awareSearcher, then, in finite time, at least one of this 
kind of robots is located on node $u$. A \potentialMin, a \dumbSearcher~or
an \awareSearcher~located with the \minWaitingWalker~becomes a 
\waitingWalker~(Rule \rTwoThree). Therefore, this lemma permits to prove that in
the case where there are at least 3 robots in the execution (after time $t_{min}$) such that they are
either \potentialMin, \dumbSearcher~or \awareSearcher, then, in finite time, a
supplementary \waitingWalker~is located on node $u$. 

To prove the following lemma, we need to introduce a new notion. We call
$Seg(u, v)$ the set of nodes (of the footprint of the dynamic ring) between node
$u$ not included and $v$ not included 
considering the right direction.

\begin{lemma} \label{search}
 If there is no \towerElection~in the execution but there exists at a time $t$ 
 at least 3 robots such that they are either \potentialMin, \dumbSearcher~or
 \awareSearcher, then it exists a time $t' \geq t$ at which at least a 
 \potentialMin, a \dumbSearcher~or an \awareSearcher, reaches the node $u$.
\end{lemma}

\begin{proof}
 Assume that there is no \towerElection~in the execution. By the rules of 
 \Gathering~and knowing that initially all the robots are \righter, this implies
 that there are only \righter, \potentialMin, \dumbSearcher, \awareSearcher, 
 \waitingWalker~and \minWaitingWalker~robots in the execution. Since there is no
 \towerElection, $r_{min}$ is \minWaitingWalker~and is located on node $u$. By
 Lemma~\ref{waitingCautiousWalkOnSameNode}, we know that all the
 \waitingWalker~robots (if any) are on node $u$ and $r_{min}$ as well as the 
 \waitingWalker~robots do not move. This implies that among the robots that are
 not on node $u$ there are only \righter, \potentialMin, \dumbSearcher~and
 \awareSearcher.
 
 Assume by contradiction that at a time $t$, there are at least 3 robots such
 that they are either \potentialMin, \dumbSearcher~or \awareSearcher~and such
 that for all time $t' \geq t$ none of these kinds of robots succeed to reach 
 the node $u$ at time $t'$. We consider the execution from time $t$.

 Consider a robot $r$ such that at time $t$,
 $state_{r} \in \{po\-ten\-tial\-Min, dumb\-Sear\-cher, aware\-Sear\-cher\}$.
 
 $(i)$ If $r$ is an \awareSearcher, since it cannot reach $u$, it executes
 Rule \rOneEleven, and hence it executes the function \textsc{Search}. The 
 variable state of $r$ is not updated while $r$ executes this function, 
 therefore $r$ is an \awareSearcher~and executes Rule \rOneEleven~and the 
 function \textsc{Search} at each instant time from time $t$. Thus by
 Lemma~\ref{lemma4Bis}, $r$ always considers a direction different from $\bot$
 after time $t$ included.
 
 $(ii)$ If $r$ is a \dumbSearcher, since it cannot reach $u$, it can 
 execute either Rule \rOneTen~(if it 
 is on the same node as an \awareSearcher) and hence becomes an 
 \awareSearcher~robot and executes the function \textsc{Search}, Rule
 \rOneNine~(if it is on the same node as a \righter) and hence becomes an 
 \awareSearcher~and executes the function \textsc{Search}, or Rule 
 \rOneEleven~and hence stays a \dumbSearcher~and executes the function 
 \textsc{Search}. By Lemma~\ref{lemma4Bis} and by $(i)$, $r$ always considers
 a direction different from $\bot$ after time $t$ included.
 
 $(iii)$ If $r$ is a \potentialMin, by 
 Lemma~\ref{OnlyOnePLeaderOrOneMovingRight}, there is no \righter~in the 
 execution at time $t$ and therefore by Observation~\ref{noMoreRighter} there 
 is no \righter~in the execution after time $t$ included. Therefore, since $r$ 
 cannot reach $u$, it can execute either Rule \rOneFive~(if it is on the
 same node as an \awareSearcher) and hence becomes an \awareSearcher~and 
 executes the function \textsc{Search}, or Rule \rOneEight~and hence stays a
 \potentialMin~and considers the right direction. Therefore by
 Lemma~\ref{lemma4Bis} and by $(i)$, $r$ always considers a direction different
 from $\bot$ after time $t$ included.

 $(iv)$ If there is a \righter~robot in the execution after time $t$, then by 
 Lemma~\ref{righterOtherAndMinNoMorePotentialMin}, there
 is no \potentialMin~robot in the execution after time $t$ included. If
 a \righter~robot is on the same node as a \dumbSearcher~or as an
 \awareSearcher, it executes Rule \rOneSeven~and hence becomes an
 \awareSearcher~and executes the function \textsc{Search}.
 
 $(v)$ While executing the function \textsc{Search} if a robot is isolated it 
 considers the last direction it considered. While executing the function 
 \textsc{Search} if a robot is not isolated, if it possesses the maximum 
 identifier among all the robots of its current location it considers the left
 direction otherwise it considers the right direction.
 
 $(vi)$ Note that if there is a \potentialMin~while $r_{min}$ is
 \minWaitingWalker, then it possesses the minimum identifier among all the 
 robots not located on node $u$. Indeed, only Rule \rOneSix~permits a robot
 to become a \potentialMin. For this rule to be executed, $\mathcal{R} - 1$ 
 \righter~robots must be located on a same node. While executing this rule the
 robot with the minimum identifier among the $\mathcal{R} - 1$ robots located on
 a same node becomes \potentialMin. Since, by Lemma~\ref{OnePotentialMin}, there 
 is only one \potentialMin~in the whole execution and since by definition
 $r_{min}$ possesses the minimum identifier among all the robots of the system, 
 $r_{min}$ does not execute Rule \rOneSix. Therefore, while $r_{min}$ is 
 \minWaitingWalker, the \potentialMin~possesses the minimum identifier among all
 the robots not located on node $u$. Thus, when a \potentialMin~executes
 Rule \rOneEight~and hence considers the right direction it possesses the same
 behavior as if it was executing the function \textsc{Search}.
 
 \begin{description}
  \item [Case 1:] \textbf{There is no eventual missing edge.}
  
  Call $d$ the direction of $r$ during the Look phase of time $t$, and let $v$
  be the node where $r$ is located during the Look phase of time $t$. Call $w$ 
  the adjacent node of $v$ in the direction $d$. Call $e$ the edge between $v$ 
  and $w$. As proved in cases $(i)$, $(ii)$ and $(iii)$, $d$ is either equal 
  to $right$ or $left$.
  
  We want to prove that it exists a time $t'$ ($t' \geq t$) such that a robot
  $r'$ (it is possible to have $r' = r$) with
  $state_{r'} \in \{po\-ten\-tial\-Min, dumb\-Sear\-cher, aware\-Sear\-cher\}$
  considers the direction $d$ and is located on $w$ during the Look phase of 
  time $t'$.
  
  Call $t"$ ($t" \geq t$) the first time after time $t$ included where there is 
  an adjacent edge to $v$. If during the Move phase of time $t"$, $r$ does
  not consider the direction $d$, by $(i)-(vi)$ this necessarily implies that
  when $r$ executes the function \textsc{Search} (or a function that behaves 
  like the function \textsc{Search}) there is at least another robot
  on its node. Moreover by $(i)-(vi)$ the other robot(s) with $r$ also executes 
  the function \textsc{Search} (or a function that behaves like the function
  \textsc{Search}) and is or becomes \potentialMin, \dumbSearcher~or
  \awareSearcher. Therefore, since all the robots possess distinct identifiers
  and by $(v)$, during the Move phase of time $t"$, a robot among
  $\{po\-ten\-tial\-Min, dumb\-Sear\-cher, aware\-Sear\-cher\}$ on node $v$
  considers the direction $d$. 
  
  Since all the edges are infinitely often present we can repeat these arguments 
  on each instant time until the time $t_{e}$ where $e$ is present. At time
  $t_{e}$ a robot (either \potentialMin, \dumbSearcher, \awareSearcher) 
  considers the direction $d$ and hence crosses $e$. Since the direction
  considered by a robot can be updated only during Compute phases, we succeed to
  prove that $t'$ exists.  
  
  Applying these arguments recurrently we succeed to prove that in finite time a 
  robot $r"$ such that 
  $state_{r"} \in \{po\-ten\-tial\-Min, dumb\-Sear\-cher, aware\-Sear\-cher\}$
  is on node $u$.
  
  \item [Case 2:] \textbf{There is an eventual missing edge.}
  
  Call $e$ the eventual missing edge. Consider the execution after the time 
  greater or equal to $t$ where $e$ is missing forever. Call $v$ the node such 
  that its adjacent right edge is $e$. Call $w$ the adjacent right node of $v$.
  
  At least two robots that are either \potentialMin, \dumbSearcher~or 
  \awareSearcher~are either on nodes in $Seg(u, v) \cup \{v\}$ or on nodes in 
  $Seg(w, u) \cup \{w\}$. 
  
  Assume that there are at least two robots that are either \potentialMin, 
  \dumbSearcher~or \awareSearcher~which are on nodes in $Seg(u, v) \cup \{v\}$.
  The reasoning when there are at least two robots that are either 
  \potentialMin, \dumbSearcher~or \awareSearcher~which are on nodes in 
  $Seg(w, u) \cup \{w\}$ is similar.
  
  The edge $e$ is an eventual missing edge. It can exist only one eventual 
  missing edge in \COT ring. Therefore all the edges between the 
  nodes in $\{u\} \cup Seg(u, v) \cup\{v\}$ are infinitely often present. Thus,
  if there exists a robot (either \potentialMin, \dumbSearcher~or
  \awareSearcher) that considers the left direction then we can apply the 
  arguments of case 1 to prove that in finite time a robot $r"$, such that 
  $state_{r"} \in \{po\-ten\-tial\-Min, dumb\-Sear\-cher, aware\-Sear\-cher\}$,
  is on node $u$.
  
  Therefore consider that all the robots, that are either \potentialMin, 
  \dumbSearcher~or \awareSearcher~and that are located on nodes in 
  $Seg(u, v) \cup \{v\}$, consider the 
  right direction. In this case a robot either \potentialMin, \dumbSearcher~or
  \awareSearcher~cannot be located on the same node as a robot either \righter,
  \potentialMin, \dumbSearcher~or \awareSearcher, otherwise during the Move
  phase of the time of this meeting, by $(i)-(vi)$, it exists a robot either
  \potentialMin, \dumbSearcher~or \awareSearcher~that considers the left 
  direction. 
  
  Since $e$ is an eventual missing edge, and since there are at least two 
  robots either \potentialMin, \dumbSearcher~or \awareSearcher~that consider the
  right direction, applying the arguments of case 1 on two of 
  these robots, we succeed to prove that in finite time two of these robots are 
  located on $v$. Therefore, by the previous paragraph, in finite time a robot
  $r"$, such that 
  $state_{r"} \in \{po\-ten\-tial\-Min, dumb\-Sear\-cher, aware\-Sear\-cher\}$,
  is on node $u$.
 \end{description}
\end{proof}

Now, we prove one of the two main lemmas of this phase: we prove that if a 
\potentialMin~is present in the execution, then, in finite time, a 
\towerElection~is present in the execution. While proving this lemma, we also
prove that, at the time when the \towerElection~is formed, among the two robots 
not involved in this \towerElection, it can exit at most one \righter. This 
information is useful to prove Phase \WaitTermination.

\begin{lemma} \label{PotentialMinAndTowerElection}
 If there is a \potentialMin~in the execution, then there exists a time
 $t$ at which a \towerElection~is present and among the robots not involved in 
 the \towerElection~there is at most one \righter~robot at time $t$.
\end{lemma}

\begin{proof}
 Assume that there exists a time $t$ at which a \potentialMin~robot is present
 in the execution. Assume by contradiction that there is no \towerElection~in 
 the execution. In the following, we consider the execution from time 
 $t' = max\{t, t_{min}\}$.
 
 Since there is no \towerElection, by the rules of \Gathering~and knowing that
 initially all the robots are \righter, there are in 
 the execution only \righter, \potentialMin, \dumbSearcher, 
 \awareSearcher, \waitingWalker~and \minWaitingWalker~robots. 
 By Lemma~\ref{waitingCautiousWalkOnSameNode}, all the \waitingWalker~are 
 located on the same node as $r_{min}$, when 
 $state_{r_{min}} = min\-Waiting\-Walker$, and both $r_{min}$, as a
 \minWaitingWalker, and the \waitingWalker~robots do not move. By
 Corollary~\ref{uniqueLeader}, only $r_{min}$ can be a \minWaitingWalker. We
 recall that $r_{min}$ as \minWaitingWalker~is located on node $u$. Therefore
 the \minWaitingWalker~and all the \waitingWalker~(if any) are located on node
 $u$.
 
 By Lemma~\ref{OnlyMovingRightAndPLeaderCanBeMin} we know that, before being 
 $min$, $r_{min}$ is either a \righter~robot or a \potentialMin~robot. We can 
 then consider the two following cases.
 
 
 \begin{description}
  \item [Case 1:] \textbf{Before being $\mathbf{min}$, $\mathbf{r_{min}}$ is a 
  $\mathbf{righter}$ robot.} 
  
  By Lemma~\ref{potentialMinAndMinNoMoreRighter}, at time $t'$ there
  are only \potentialMin, \dumbSearcher, \awareSearcher, \waitingWalker~and 
  \minWaitingWalker~robots in the execution. Moreover, in this case, all the 
  robots that are not on node $u$ are necessarily either \potentialMin, 
  \dumbSearcher~or \awareSearcher.
  
  When a \potentialMin, a \dumbSearcher~or an \awareSearcher~robot meets the
  \minWaitingWalker, it executes Rule \rTwoThree, hence it becomes a 
  \waitingWalker~and stops to move.
  
  Then each time there are at least 3 robots in the execution such that they are
  either \potentialMin, \dumbSearcher~and/or \awareSearcher, using
  Lemma~\ref{search}, we succeed to prove that at least one \potentialMin,
  \dumbSearcher~or \awareSearcher~succeeds to join the node $u$ and therefore 
  becomes a \waitingWalker. Therefore, 
  by Lemma~\ref{waitingCautiousWalkOnSameNode}, a \towerElection~is formed in
  finite time.  
  
  \item [Case 2:] \textbf{Before being $\mathbf{min}$, $\mathbf{r_{min}}$ is a
  $\mathbf{potentialMin}$.}
  
  For a robot to become a \potentialMin, Rule \rOneSix~must be executed.
  This rule is executed when $\mathcal{R} - 1$ \righter~are located on a same 
  node. While executing this rule, among the $\mathcal{R} - 1$ \righter~located 
  on a same node, one becomes \potentialMin~while the other become
  \dumbSearcher. By Observation~\ref{noMoreRighter}, none of these
  $\mathcal{R} - 1$ robots can become \righter~anymore. Therefore, by
  Lemma~\ref{OnePotentialMin}, once $r_{min}$
  is $min$, there are only \dumbSearcher, \awareSearcher, \waitingWalker,
  \minWaitingWalker~robots and at most one \righter~robot in the execution. 
  Moreover, in this case, among the robots that are not on node $u$, there are
  only \dumbSearcher~and \awareSearcher~and at most one \righter. By the rules 
  of \Gathering, as long as a \dumbSearcher~or an \awareSearcher~is not on node
  $u$, its variable state stays in $\{dumb\-Searcher, aware\-Searcher\}$.  
  
  Once $r_{min}$ is $min$, if there exists a time at which there is no more
  \righter~robot in the execution, then, using the arguments of case 1, we 
  succeed to prove that a \towerElection~is formed in finite time. Therefore
  assume that there is always a \righter~robot $r$ in the execution.
  
  When a \dumbSearcher~or an \awareSearcher~robot is located on the same node as 
  the \minWaitingWalker, it executes Rule \rTwoThree, hence it becomes a 
  \waitingWalker~and stops to move. Then, using multiple times
  Lemma~\ref{search} and Lemma~\ref{waitingCautiousWalkOnSameNode}, we know that 
  in finite time there are in the execution only one \righter~and only 2 robots
  $r'$ and $r"$ such that
  $state_{r'}, state_{r"} \in \{aware\-Sear\-cher, dumb\-Sear\-cher\}^{2}$ (all 
  the other robots are \minWaitingWalker~and \waitingWalker~robots and are 
  located on node $u$). Note that $r'$ (resp. $r"$) cannot be located on node 
  $u$, otherwise, by Rule \rTwoThree, a \towerElection~is formed. Therefore,
  $r'$ and $r"$ always have their variable state in
  $\{dumb\-Searcher, aware\-Searcher\}$.  
  
  When a \righter~robot is located on the same node as an \awareSearcher~or as a
  \dumbSearcher, it executes Rule \rOneSeven~and becomes an \awareSearcher.
  Similarly, if a \righter~is on the same node as a \minWaitingWalker~while the 
  adjacent right edge to its position is present, then the \righter~robot 
  executes Rule \rTwoFour~and becomes an \awareSearcher. Therefore, as 
  highlighted previously, these situations cannot happen, otherwise a 
  \towerElection~is formed in finite time. This implies that, as long as the
  robot $r$ is not on node $u$, it must be isolated. Since $r'$ and $r"$ cannot 
  be located on node $u$, if $r$ succeeds to join the node $u$ in the case there 
  is no present adjacent right edge to $u$, then $r$ executes Rule 
  \rOneEight~and therefore stays a \righter~and considers the right direction. 
  Therefore, since an isolated \righter~robot always executes Rule
  \rOneEight, hence always considers the right direction, this implies that
  either $r$ is on a node $v$ ($v \neq u$) such that the adjacent right edge of
  $v$ is an eventual missing edge at least from the time where $r$ is on node 
  $v$ (case 2.1) or $r$ succeeds to reach $u$ but the adjacent right edge of $u$
  is an eventual missing edge at least from the time where $r$ is on node $u$ 
  (case 2.2).
  
  $(*)$ When an \awareSearcher~or a \dumbSearcher~is isolated it 
  executes Rule \rOneEleven, hence
  executes the function \textsc{Search}, therefore it considers the last
  direction it considered. By Lemma~\ref{lemma4Bis}, this direction cannot 
  be equal to $\bot$. 
  
  $(**)$ Since only $r'$ and $r"$ have their variable state in 
  $\{dumb\-Searcher, aware\-Searcher\}^{2}$, and since $r'$ and $r"$ cannot be 
  located on node $u$ and cannot be located with $r$, if a \dumbSearcher~is
  located on the same node as an \awareSearcher~or if an \awareSearcher~(resp. a
  \dumbSearcher) is located on the same node as another \awareSearcher~(resp.
  \dumbSearcher), necessarily this means that $r'$ and $r"$ are located 
  on a same node, and there is no other robot on the same node as them. When a 
  \dumbSearcher~is on the same node as an \awareSearcher~it executes Rule
  \rOneTen, hence it becomes an \awareSearcher~and executes the function
  \textsc{Search}. When an \awareSearcher~is on the same node as a 
  \dumbSearcher~it executes Rule \rOneEleven~and hence executes the function
  \textsc{Search}. Since $r'$ and $r"$ have distinct identifiers, 
  when an \awareSearcher~and a \dumbSearcher~are on a same node, they both
  execute the function \textsc{Search}, therefore one considers the right
  direction, while the other one considers the left direction. Similarly, if two 
  \awareSearcher~(resp. \dumbSearcher) robots are on the same node, they both 
  execute Rule \rOneEleven~and hence the function \textsc{Search}, therefore 
  one considers the right direction, while the other one considers the left
  direction. 
  
  \begin{description}
   \item [Case 2.1:]
   
   Let $w$ be the adjacent node of $v$ in the right direction. It can 
   exist only one eventual missing edge, which is the adjacent right edge of 
   node $v$. Therefore, if a robot, in $Seg(u, v)$ or in $Seg(w, u)$, considers
   a direction $d$ and does not change this direction, it eventually succeeds to
   move in this direction. Similarly, if a robot is on node $w$ and always 
   considers the right direction, it eventually succeeds to move in this 
   direction ($***$). 
   
   Firstly, assume that only $r'$ (resp. $r"$) is on a node in $Seg(u, v)$. By 
   $(*)$ and $(***)$, $r'$ (resp. $r"$) cannot consider the right direction, 
   otherwise it reaches $r$ in finite time. Therefore $r'$ (resp. $r"$) 
   considers the left direction. By $(*)$ and $(***)$, in finite time, $r'$ 
   (resp. $r"$) succeeds to reach $u$, implying that a \towerElection~is formed.
   
   Secondly, assume that $r'$ and $r"$ are on nodes in $Seg(u, v)$. By $(*)$, 
   $(**)$ and $(***)$, they cannot meet otherwise one of them reaches $u$ in
   finite time. Moreover, if they do not meet none of them can consider the left
   direction otherwise, by $(*)$ and $(***)$, they reach $u$ in finite time.
   Therefore, they cannot meet and must consider the right direction. By $(*)$ 
   and $(***)$, in finite time one robot among $r'$ and $r"$ succeeds to reach
   $r$, implying that a \towerElection~is formed.
   
   Thirdly, assume that $r$ and $r"$ are on nodes in $Seg(v, u)$. By $(*)$,  
   $(**)$ and $(***)$, they cannot meet otherwise one of them reaches $u$ in 
   finite time. Moreover, if they do not meet none of them can consider the 
   right direction otherwise, by $(*)$ and $(***)$, they reach $u$ in finite 
   time. Therefore, they cannot meet and must consider the left direction. However,
   by $(*)$ and $(***)$, since the adjacent right edge of $v$ is missing 
   forever, in finite time $r'$ and $r"$ reach $w$, which is a contradiction
   with the fact that they do not meet.  
   
   \item [Case 2.2:]
   
   Applying the arguments used in the case 2.1, when $r'$ and $r"$ are on nodes
   in $Seg(v, u)$, to $r'$ and $r"$ when there are on nodes in $Seg(u, u)$, we
   succeed to prove that in finite time at least one of them reaches node $u$,
   making Rule \TerminationTwo~true, which leads to a contradiction.   
  \end{description}
 \end{description}
\end{proof}

Finally, we prove the other main lemma of this phase: we prove that even if 
there is no \potentialMin~in the execution, then, in finite time, a 
\towerElection~is present in the execution. While proving this lemma, we also
prove that, at the time when the \towerElection~is formed, among the two robots 
not involved in this \towerElection, it can exit at most one \righter. This 
information is useful to prove Phase \WaitTermination.

\begin{lemma} \label{NoPotentialMinAndTowerElection}
 If there is no \potentialMin~in the execution, then there exists a time
 $t$ at which a \towerElection~is present and among the robots not involved in
 the \towerElection~there is at most one \righter~robot at time $t$.
\end{lemma}

\begin{proof}
 Assume, by contradiction, that there is no \towerElection~in the execution.
 By the rules of \Gathering~and knowing that initially all the robots are 
 \righter, this implies that there are only \righter, \potentialMin,
 \dumbSearcher, \awareSearcher, \waitingWalker~and \minWaitingWalker~robots in
 the execution.
 
 Assume that there is no \potentialMin~in the execution. If there is no 
 \potentialMin~in the execution, it cannot exist \dumbSearcher~in 
 the execution. Indeed, the only way for a robot to become \dumbSearcher~is 
 to execute Rule \rOneSix. However, when this rule is executed, a robot 
 becomes \potentialMin. Therefore, there are in the execution only \righter, 
 \awareSearcher, \waitingWalker~and \minWaitingWalker~robots. 

 Before time $t_{min}$, by the rules of \Gathering, there are only \righter~in
 the execution. Indeed, by Corollary~\ref{uniqueLeader}, only $r_{min}$ can be 
 \minWaitingWalker~and it becomes \minWaitingWalker~at time $t_{min}$. Moreover,
 the only way for a robot to become \waitingWalker~is to execute Rule 
 \rTwoThree. In the case where there is no \potentialMin~in the execution, only
 an \awareSearcher~located with $r_{min}$, as a \minWaitingWalker, can execute 
 this rule. Besides, the only ways for a robot to become an \awareSearcher~is 
 either to be a \righter~and to be located with an \awareSearcher~(refer to
 Rule \rOneSeven), or to be a \righter~and to be located with $r_{min}$, as a
 \minWaitingWalker, while an adjacent right edge to their location is present 
 (refer to Rule \rTwoFour). Since initially all the robots are \righter, 
 the first \awareSearcher~of the execution can be present only thanks to the
 execution of Rule \rTwoFour. 
 
 All this implies that, even after time $t_{min}$, as long as no \righter~robot
 is on node $u$ with $r_{min}$, as a \minWaitingWalker, while there is a present
 adjacent right edge to $u$, it cannot exist neither \awareSearcher~nor 
 \waitingWalker~in the execution: there is at most one \minWaitingWalker~and 
 there are at least $\mathcal{R} - 1$ \righter.
 Moreover, this implies that as long as the situation described has not
 happened, all the \righter~robots only execute Rule \rOneEight, hence
 always consider the right direction. 
 
 Consider the execution just after time $t_{min}$. In this context, necessarily,
 in finite time, there exists a \righter~robot $r$ that succeeds to reach $u$ 
 (while $r_{min}$ is \minWaitingWalker). Indeed, if this is not the case, this
 implies that there exists an eventual missing edge $e$. Since all the 
 \righter~robots always consider the right direction and since it can exist at
 most one eventual missing edge, this implies that $\mathcal{R} - 1$ 
 \righter~robots reach in finite time the same extremity of $e$. Thus, Rule
 \rOneSix~is executed, which leads to a contradiction with the fact that there
 is no \potentialMin~in the execution. 
 
 Similarly, necessarily, in finite time, there exists an adjacent right edge to
 $u$ while $r$ is on $u$. Indeed, if this is not the case, this implies that the
 adjacent right edge of $u$ is an eventual missing edge. Since all the 
 \righter~robots always consider the right direction and since it can exist at 
 most one eventual missing edge, in finite time all the \righter~succeed to be 
 located on node $u$. This implies that Rule \TerminationOne~is 
 executed, which leads to a contradiction.
 
 Therefore there exists a time $t'$ at which $r$ executes Rule \rTwoFour. At 
 this time $r$ becomes an \awareSearcher~robot and considers the right 
 direction. We then consider the execution from time $t'$.
 
 $(*)$ From this time $t'$, as long as there exists \righter~in the execution, 
 it always exists an \awareSearcher~robot $r'$ considering the right direction,
 such that there is no \righter~robots on $Seg(u, v)$, where $v$ is the node
 where $r'$ is currently located. This can be proved by analyzing the movements 
 of the different kinds of robots that we describe in $(i)-(vii)$.
 
 $(i)$ By Lemma~\ref{waitingCautiousWalkOnSameNode}, all the 
 \minWaitingWalker~and \waitingWalker~(if any) are on a same node (which is the 
 node $u$) and do not move.
 
 $(ii)$ If an \awareSearcher~is located on node $u$, therefore if it is located
 with $r_{min}$, as a \minWaitingWalker, it executes Rule \rTwoThree~and 
 becomes a \waitingWalker~robot.
 
 $(iii)$ If an \awareSearcher~is on a node different from the node $u$, the only
 rule it can execute is Rule \rOneEleven, in which the function
 \textsc{Search} is called. While executing this function, an isolated 
 \awareSearcher~considers the direction it considers during its last Move phase.
 By Lemma~\ref{lemma4Bis}, this direction cannot be $\bot$.
 
 $(iv)$ If a \righter~robot is located only with other \righter~robots or if it
 is located on node $u$, therefore if it is located with $r_{min}$, as a 
 \minWaitingWalker, such that there is no adjacent right edge to $u$, it 
 executes Rule \rOneEight, hence it stays a \righter~and considers the right
 direction. 
 
 $(v)$ If a \righter~robot is with $r_{min}$, as a \minWaitingWalker, such that
 there is an adjacent right edge to $u$, then it executes Rule \rTwoFour~and
 hence becomes an \awareSearcher.
 
 $(vi)$ If a \righter~robot is on a node different from node $u$ with an 
 \awareSearcher, it executes Rule \rOneSeven~and therefore becomes
 \awareSearcher~and executes the function \textsc{Search}. 
 
 $(vii)$ Note that by the movements described in $(i)$ to $(vi)$, if a robot
 executes the function \textsc{Search}, then all the robots that are on the same
 node as it also execute this function. While executing the function 
 \textsc{Search}, if multiple robots are on the same node, one considers the 
 left direction, while the others consider the right direction.
  
 Applying these movements on $r'$ and recursively on the robots that $r'$ meet
 that consider the right direction after their meeting with $r'$ and so on, we
 succeed to prove the property $(*)$. 
 
 $(**)$ Note that if there exists a time at which there is no more 
 \righter~in the execution, then by applying $(ii)$, 
 Lemma~\ref{waitingCautiousWalkOnSameNode} and Lemma~\ref{search} multiple times
 we succeed to prove that a \towerElection~is formed. Therefore at least one 
 robot is always a \righter~during the whole execution. Call $\mathcal{S}_{r}$ 
 the set of \righter~robots that stay \righter~during the whole execution. 
 
%
 
 Let us consider the following cases.
 
 \begin{description}
  \item [Case 1:] \textbf{There does not exist an eventual missing edge.}
    
  None of the robots of $\mathcal{S}_{r}$ can be located on the same node as an 
  \awareSearcher, otherwise, by $(vi)$, they become \awareSearcher. Therefore, 
  all the robots of $\mathcal{S}_{r}$ that are not on node $u$ can only consider
  the right direction (refer to $(iv)$). Since all the edges are infinitely 
  often present, for each robot $r"$ of $\mathcal{S}_{r}$, it exists a time at 
  which $r"$ is on node $u$. Moreover, once on node $u$, as long as there is no
  adjacent right edge to $u$, $r"$ considers the right direction (refer to 
  $(iv)$), and therefore stays on node $u$. Thus, since all the edges are
  infinitely often present, for each robot $r"$ of $\mathcal{S}_{r}$, it exists 
  a time at which $r"$ is on node $u$ such that an adjacent right edge to $u$ is
  present. Therefore, by $(v)$, in finite time, all the robots of 
  $\mathcal{S}_{r}$ are \awareSearcher~robots. Hence, by $(**)$, the lemma is
  proved.
  
  \item [Case 2:] \textbf{There exists an eventual missing edge.}
  
  Call $x$ the node such that its adjacent right edge is the eventual missing
  edge. Consider the execution after time $t'$ such that the eventual missing 
  edge is missing forever. 
  
  \begin{description}
   \item [Case 2.1:] \textbf{$\mathbf{x = u}$.}
   
   None of the robots of $\mathcal{S}_{r}$ can be located on the same node as an 
   \awareSearcher, otherwise, by $(vi)$, they become \awareSearcher. Therefore, 
   all the robots of $\mathcal{S}_{r}$ that are not on node $u$ can only consider
   the right direction (refer to $(iv)$). Since it can exist at most one 
   eventual missing edge, in finite time the robots of $\mathcal{S}_{r}$ succeed
   to reach node $u$, and stay on node $u$ (refer to $(iv)$). Necessarily, 
   $|\mathcal{S}_{r}| < \mathcal{R} - 2$, otherwise Rule \TerminationTwo~is 
   executed. At the time at which all the robots of $\mathcal{S}_{r}$ are on 
   node $u$, by $(*)$, we know that at least one \awareSearcher, on a node $v$, 
   considers the right direction. By $(vi)$, none of the \righter~of 
   $\mathcal{S}_{r}$ can be located on node $v$. Therefore, this
   \awareSearcher~is not on node $u$. By the movements described in $(iii)$ and
   $(vii)$, we know that in finite time an \awareSearcher~succeeds to reach node
   $u$. Then all the \righter~of $\mathcal{S}_{r}$ become \awareSearcher, hence 
   by $(**)$, the lemma is proved.
   
   \item [Case 2.2:] \textbf{$\mathbf{x \neq u}$.}
   
   None of the robots of $\mathcal{S}_{r}$ can be located on the same node as an 
   \awareSearcher, otherwise, by $(vi)$, they become \awareSearcher. Therefore,
   none of the robots of $\mathcal{S}_{r}$ can be located on 
   $Seg(u, x) \cup \{x\}$, otherwise, in finite time, by $(iv)$ they are located 
   on node $x$. However, once all the robots of $\mathcal{S}_{r}$ are on node 
   $x$, by $(*)$, and the movements described in $(iii)$ and $(vii)$ an
   \awareSearcher~succeeds to be located on node $u$ in finite time,
   which leads to a contradiction. Therefore all the robots of $\mathcal{S}_{r}$
   are on nodes in $Seg(x, u)$. Since it can exist only one eventual missing
   edge, and since this edge is the adjacent right edge of $x$, for each robot 
   $r"$ of $\mathcal{S}_{r}$, by $(iv)$, it exists a time at which $r"$ is on
   node $u$ while there is a present adjacent right edge to $u$. Therefore, by 
   $(v)$, in finite time all the robots of $\mathcal{S}_{r}$ are 
   \awareSearcher~robots. Hence, by $(**)$, the lemma is proved.
  \end{description}
 \end{description}

 We just proved that it exists a time $t_{tower}$ at which a \towerElection~is
 present in the execution. We now prove that, at time $t_{tower}$, among the
 robots not involved in the \towerElection, there is at most one \righter. By
 Lemma~\ref{uniqueTowerElection}, there is only one \towerElection~in the whole
 execution. Necessarily, as explained above when there is no \potentialMin~in
 the execution, in order to have a \towerElection, a \righter~must become an
 \awareSearcher~while executing Rule \rTwoFour. The property $(*)$ is then 
 true. By definition of a \towerElection, only two robots are not involved in
 the \towerElection. Assume, by contradiction, that there are two \righter~not
 involved in the \towerElection~at time $t_{tower}$. By $(*)$, this implies that
 there is an \awareSearcher~at time $t_{tower}$. However, by definition, a 
 \towerElection~is composed of one \minWaitingWalker~and $\mathcal{R} - 3$
 \waitingWalker, therefore, since there are $\mathcal{R}$ robots in the system 
 and among them, at time $t_{tower}$, two are \righter~and one is an
 \awareSearcher, there is a contradiction with the fact that there is a 
 \towerElection~at time $t_{tower}$.
\end{proof}

By Lemmas~\ref{PotentialMinAndTowerElection} and 
\ref{NoPotentialMinAndTowerElection}, we can deduce the following corollary 
which proves the correctness of Phase \MinWaitToBeKnown.

\begin{corollary} \label{towerElection}
 There exists a time $t$ in the execution at which a \towerElection~is 
 present and among the robots not involved in the \towerElection~there is at 
 most one \righter~robot at time $t$.
\end{corollary}

%
%
%
%
%
%
%
%
%
%
%
%

\subsubsection{Proofs of Correctness of Phases \Walk~and \WaitTermination}

The combination of Phases \Walk~and \WaitTermination~of \Gathering~permit to 
solve \GEW~in \COT~rings. Since \GEW~is divided into a safety and a liveness
property, to prove the correctness of Phases \Walk~and \WaitTermination, we
have to prove each of these two properties. We recall that, to satisfy the safety
property of the gathering problem, all the robots that terminate their execution have to do so on the same
node, and to satisfy the liveness property of \GEW, at least $\mathcal{R} - 1$ 
robots must terminate their execution in finite time. In this subsection, we, 
first, prove that \Gathering~solves the safety of the gathering problem in \COT~rings, and then, we prove
that \Gathering~solves the liveness of \GEW~in \COT~rings. We prove this 
respectively in Lemmas~\ref{safety} and \ref{theoremC5}. To prove these two 
lemmas, we need to prove some other lemmas.

By Corollary~\ref{towerElection}, we know that, in finite time, a 
\towerElection~is formed. By Lemma~\ref{uniqueTowerElection}, there is at most
one \towerElection~in the execution. Therefore, there is one and only one 
\towerElection~in the execution. Call $T$ such a \towerElection. Let $t_{tower}$
be the time at which $T$ is formed. By definition, a \towerElection~is composed 
of $\mathcal{R} - 2$ robots. Call $r_{1}$ and $r_{2}$ the two robots that are
not involved in $T$. 

In the previous subsection, we prove that, at time $t_{tower}$, at most one of the
robots among $r_{1}$ and $r_{2}$ is a \righter. In the following lemma, we go
farther and give the set of possible values for the variable $state$ at time
$t_{tower}$ of each of these robots. 

\begin{lemma} \label{kindR1AndR2}
 At time $t_{tower}$, 
 $state_{r_{1}} \in \{$\righter, \potentialMin, \dumbSearcher, \awareSearcher$\}$
 and $state_{r_{2}} \in \{$\dumbSearcher, \awareSearcher$\}$.
\end{lemma}

\begin{proof}
 Until the Look phase of time $t_{tower}$, by the rules of 
 \Gathering~and knowing that all the robots are initially \righter, there are 
 only \righter, \potentialMin, \dumbSearcher, \awareSearcher, \waitingWalker~and
 \minWaitingWalker~robots in the execution.
 
 By Corollary~\ref{uniqueLeader}, only $r_{min}$ can be $min$, therefore
 only $r_{min}$ can be \minWaitingWalker. By definition of a \towerElection, a 
 \minWaitingWalker~is involved in $T$. Since $r_{1}$ and $r_{2}$ are not 
 involved in $T$, this implies that neither $r_{1}$ nor $r_{2}$ can be
 \minWaitingWalker~at time $t_{tower}$.

 By definition of a \towerElection, at time $t_{tower}$, the $\mathcal{R} - 2$
 robots involved in $T$ are on a same node. This node is the node $u$. 
 Therefore, at time $t_{tower}$ neither $r_{1}$ nor $r_{2}$ can be located on
 node $u$, otherwise Rule \TerminationTwo~is executed. By 
 Lemma~\ref{waitingCautiousWalkOnSameNode}, this implies that neither $r_{1}$ 
 nor $r_{2}$ can be a \waitingWalker~at time $t_{tower}$. 
 
 By Corollary~\ref{towerElection}, at time $t_{tower}$, only one robot among 
 $r_{1}$ and $r_{2}$ can be a \righter~robot. Assume without lost of generality 
 that $r_{1}$ is a \righter~at time $t_{tower}$. In this case by 
 Corollary~\ref{towerElection}, $r_{2}$ cannot be a \righter~at time
 $t_{tower}$. Moreover, in this case, by 
 Lemma~\ref{OnlyOnePLeaderOrOneMovingRight}, $r_{2}$ cannot be a 
 \potentialMin~at time $t_{tower}$.
 
 Now assume without lost of generality that $r_{1}$ is a \potentialMin~robot at
 time $t_{tower}$. By Lemmas~\ref{OnePotentialMin} and 
 \ref{OnlyOnePLeaderOrOneMovingRight}, $r_{2}$ can neither be a 
 \potentialMin~nor a \righter~at time $t_{tower}$.
 
 This prove the lemma.
\end{proof}

In the following lemma, we prove a property on Rules \TerminationOne~and 
\TerminationTwo~that helps us to prove that \Gathering~solves \GEW~in \COT~rings.

\begin{lemma} \label{term}
 If a robot $r$, on a node $x$, at a time $t$, executes Rule 
 \TerminationOne~(resp. \TerminationTwo), then there are $\mathcal{R}$ (resp.
 $\mathcal{R} - 1$) robots on node $x$ at time $t$ and they all execute Rule
 \TerminationOne~(resp. \TerminationTwo) at time $t$ (if they are not already 
 terminated).
\end{lemma}

\begin{proof}
 If a robot $r$, on a node $x$, executes Rule \TerminationOne~(resp. 
 \TerminationTwo) at a time $t$, by the predicate \GE$()$ (resp.
 \GEW$()$), there are $\mathcal{R}$ (resp. $\mathcal{R} - 1$) robots 
 on $x$ at time $t$. Moreover, if the predicate \GE$()$ (resp. 
 \GEW$()$) is true for $r$ at time $t$, since the robots are 
 fully-synchronous, it is necessarily true for all the robots (not already 
 terminated) located on node $x$ at time $t$. This implies that all the robots
 (not already terminated) located on $x$ at time $t$, execute Rule
 \TerminationOne~(resp. \TerminationTwo) at time $t$. 
\end{proof}

Now we prove one of the two main lemmas of this subsection: we prove that 
\Gathering~solves the safety property of the gathering problem in \COT~rings.

\begin{lemma} \label{safety}
 \Gathering~solves the safety of the gathering problem in \COT~rings.
\end{lemma}

\begin{proof}
 We want to prove that, while executing \Gathering, all robots that terminate
 their execution terminate it on the same node. While executing \Gathering, the 
 only way for a robot to terminate its execution is to execute either Rule
 \TerminationOne~or Rule \TerminationTwo. 
 
 By Lemma~\ref{term}, if a robot $r$, on a node $x$, at a time $t$, executes
 Rule \TerminationOne, then there are $\mathcal{R}$ robots on node $x$ at 
 time $t$ and they all execute Rule \TerminationOne~at time $t$ (if they are
 not already terminated). Therefore, in the case where $r$ executes Rule 
 \TerminationOne~at time $t$, all the robots of the system are terminated on $x$
 at time $t$, hence the lemma is proved in this case.
 
 By Lemma~\ref{term}, if a robot $r$, on a node $x$, at a time $t$, executes
 Rule \TerminationTwo, then there are $\mathcal{R} - 1$ robots on node $x$ 
 at time $t$ and they all execute Rule \TerminationTwo~at time $t$ (if they
 are not already terminated). Therefore, in the case where $r$ executes Rule 
 \TerminationTwo~at time $t$, $\mathcal{R} - 1$ robots of the system are 
 terminated on $x$ at time $t$. Call $r'$ the robot that is not on the node $x$
 at time $t$. Let $y$ ($y \neq x$) be the node where $r$ is located at time $t$.
 To prove the lemma, it stays to prove that $r'$ is not terminated at time $t$,
 and that after time $t$, $r'$ either terminates its execution on node $x$ or
 never terminates its execution.
 
 Assume, by contradiction, that at time $t$, $r'$ is terminated. This implies 
 that there exists a time $t' \leq t$ at which $r'$ executes either Rule 
 \TerminationOne~or Rule \TerminationTwo. By Lemma~\ref{term}, this implies 
 that at least $\mathcal{R} - 2$ other robots are terminated on node $y$ at time 
 $t'$. Therefore, there is a contradiction with the fact that $r$ executes
 Rule \TerminationTwo~at time $t$ on node $x$. Indeed, to execute Rule 
 \TerminationTwo~at time $t$ on node $x$, $\mathcal{R} - 1$ robots must be 
 located on node $x$ at time $t$, since $\mathcal{R} \geq 4$, it is not possible
 to have $\mathcal{R} - 1$ robots on node $x$ at time $t$. 

 Moreover, after time $t$, by Lemma~\ref{term}, $r'$ can terminate its execution
 only on node $x$ (since it is the only node where $\mathcal{R} - 1$ robots are 
 located). Therefore, the lemma is proved. 
\end{proof}

The following lemma is an extension of Lemma~\ref{lemma4Bis}. While 
Lemma~\ref{lemma4Bis} is true when the robots are either executing Phase 
\AmITheMin~or Phase \MinWaitToBeKnown, the following lemma is true whatever 
the phase of the algorithm the robots are executing.

\begin{lemma} \label{NotBotDirection3}
 If, at time $t$, an isolated robot $r$ is such that 
 $state_{r} \in \{dumb\-Searcher, aware\-Searcher\}$, then, during the Move 
 phase of time $t - 1$, it does not consider the $\bot$ direction. 
\end{lemma}

\begin{proof}
 By the rules of \Gathering, \minWaitingWalker, \waitingWalker, \minTailWalker,
 \tailWalker, \headWalker~and \leftWalker~cannot become \dumbSearcher~or 
 \awareSearcher.
 
 Consider an isolated robot $r$ such that, at a time $t$, 
 $state_{r} \in \{dumb\-Searcher, aware\-Searcher\}$. 
 
 Consider then the two following cases.
 
 \begin{description}
  \item [Case 1:] \textbf{At time $\mathbf{t - 1}$, $\mathbf{r}$ is neither a 
  \textit{dumb\-Searcher} nor an \textit{aware\-Searcher}.}
 
  Whatever the state of $r$ at time $t - 1$ (\righter~or
  \potentialMin), to have its variable state at time $t$ equals either to 
  \dumbSearcher~or to \awareSearcher, $r$ executes at time $t - 1$ either
  Rule \rTwoFour, \rOneTwo, \rOneThree, \rOneFive, \rOneSix~or \rOneSeven.
  
  When a robot executes Rule \rOneTwo, it calls the function 
  \textsc{BecomeAwareSearcher}. When a robot executes the function
  \textsc{BecomeAwareSearcher}, it sets its direction to the $right$ direction, 
  therefore the lemma is true in this case. 
 
  A robot executes Rule \rOneThree~only if it is located with a 
  \headWalker~on a node $x$. Necessarily there is no present adjacent right edge
  to $x$ at time $t - 1$, otherwise the robot would have executed Rule
  \rOneTwo. By the rules of \Gathering, a \headWalker~only considers the $\bot$
  direction or the right direction. Indeed, a \headWalker~can only execute
  Rules \rFourTwo, \rFourThree~and \rThreeOne. While executing Rule 
  \rFourTwo, a \headWalker~becomes a \leftWalker~and considers the $\bot$ 
  direction. While executing Rule \rFourThree, a \headWalker~considers the 
  $\bot$ direction. Finally, while executing Rule \rThreeOne, a 
  \headWalker~considers either the right direction or the $\bot$ direction. 
  Therefore, even if, after the execution of Rule \rOneThree, $r$ considers
  the $\bot$ direction, it is not isolated at time $t$, hence the lemma is not
  false in this case. 
 
  Then, we can use the arguments of the proof of Lemma~\ref{lemma4Bis} (in 
  the case where the robot $r$ is a \dumbSearcher~or an \awareSearcher~at time 
  $t$) to prove that the current lemma is true for the remaining cases. Indeed,
  even if in Lemma~\ref{lemma4Bis} the context is such that there is no 
  \towerElection~in the execution, the arguments used in its proof are still 
  true in the context of the current lemma.
  
  \item [Case 2:] \textbf{At time $\mathbf{t - 1}$, $\mathbf{r}$ is a 
  \textit{dumb\-Searcher} or an \textit{aware\-Searcher}.}
  
  Whatever the state of $r$ at time $t - 1$ (\dumbSearcher~or \awareSearcher),
  to have its variable state at time $t$ equals either to \dumbSearcher~or to
  \awareSearcher, $r$ executes at time $t - 1$ either Rule \rOneTwo, 
  \rOneThree, \rOneNine, \rOneTen~or \rOneEleven. 
  
  We can use the arguments of Case 1 to prove that while executing
  Rule \rOneTwo~or \rOneThree, the lemma is proved.
  
  Then, similarly as for the Case 1, we can use the arguments of the proof of
  Lemma~\ref{lemma4Bis} (in the case where the robot $r$ is a \dumbSearcher~or 
  an \awareSearcher~at time $t$) to prove that the current lemma is true in 
  the remaining cases of Case 2.
 \end{description}
\end{proof}

Finally, we prove the other main lemma of this subsection: we prove that 
\Gathering~solves the liveness of \GEW~in \COT~rings. In the following proof, we
consider that there exists an eventual missing edge while \Gathering~is
executed, otherwise, during the execution, the ring is a \RE~ring (we
treat the case of \RE~rings in subsection~\ref{gracefull}).

\begin{lemma} \label{theoremC5}
 \Gathering~solves the liveness of \GEW~in \COT~rings.
\end{lemma}

\begin{proof}
 By contradiction, assume that \Gathering~does not solve the liveness of \GEW~in 
 \COT~rings. Since the execution of Rules \TerminationOne~and
 \TerminationTwo~permits a robot to terminate its execution, by 
 Lemma~\ref{term}, this implies that there exists a \COT~ring such
 that, during the execution of \Gathering, neither Rule \TerminationOne~nor
 Rule \TerminationTwo~is executed. Consider the execution of \Gathering~on 
 that ring.
 
 By Corollary~\ref{towerElection}, there exists a time $t$ at which a 
 \towerElection~is formed. Note that $\mathcal{R} - 2 \geq 2$ robots are 
 involved in a \towerElection. Once a \towerElection~is formed the 
 $\mathcal{R} - 3$ \waitingWalker~and the \minWaitingWalker~involved in this
 \towerElection~execute Rule \rTwoOne. While executing this rule, the robot 
 $r$ with the maximum identifier among the $\mathcal{R} - 2$ robots involved in 
 this \towerElection~becomes \headWalker, the \minWaitingWalker~becomes 
 \minTailWalker~and the other robots involved in this \towerElection~become
 \tailWalker. Note that, by Corollary~\ref{uniqueLeader}, only $r_{min}$ can be 
 $min$, and therefore, since $r_{min}$ is the robot with the minimum identifier 
 among all the robots of the system and since at least $2$ robots are involved 
 in the \towerElection, $r_{min}$ cannot become \headWalker. By 
 Lemma~\ref{uniqueTowerElection} and by the rules of \Gathering, only $r$ can 
 be \headWalker~and only $r_{min}$ can be \minTailWalker~during the execution.
  
 There is no rule in \Gathering~permitting a \tailWalker~or a
 \minTailWalker~robot to become another kind of robot. A \tailWalker~and a 
 \minTailWalker~can only execute Rules \rFourThree~and \rThreeOne. By
 the rules of \Gathering, the \minTailWalker~and the \tailWalker~execute the 
 same movements at the same time starting from the same node, therefore, they
 are on a same node at each instant time. Hence, call $tail$ the set of all of 
 these robots.
 
 A \headWalker~can become a \leftWalker. However, since we assume that
 the liveness of \GEW~cannot be solved, then it is not possible for $r$ to become a 
 \leftWalker. Indeed, a \headWalker~can only execute Rules \rFourTwo, 
 \rFourThree~and \rThreeOne. Note that, by the rules of \Gathering, after the
 execution of Rule \rTwoOne, the \headWalker~and the $tail$ both execute
 Rule \rThreeOne. Therefore, since the \headWalker~and the $tail$ start the
 execution of Rule \rThreeOne~on the same node at the same time, by the 
 rules of \Gathering, while the \headWalker~is executing Rule \rFourThree~or
 Rule \rThreeOne, if the $tail$ is not on the same node as the \headWalker, 
 it is either executing Rule \rThreeOne~or it is terminated. Moreover, by 
 the same arguments, in the remaining of the execution,
 the \headWalker~and the $tail$ are either on a same node or the
 $tail$ is on the left adjacent node (on the footprint of the dynamic ring) of
 the node where the \headWalker~is
 located. Hence, if at a time $t'$, the \headWalker~executes Rule \rFourTwo, 
 and therefore becomes a \leftWalker, then this implies that during time $t' - 1$
 it is executing either Rule \rFourThree~or Rule \rThreeOne~while there 
 is an adjacent left edge to its position and at time $t'$ the $tail$ is not on 
 its node. Therefore, necessarily the $tail$ is terminated, otherwise as
 explained the $tail$ would have join the \headWalker~on its node (Rule 
 \rThreeOne). Since only Rules \TerminationOne~and \TerminationTwo~permit a 
 robot to terminate its execution, by Lemma~\ref{term}, this implies that the 
 $tail$ has executed Rule \TerminationTwo, which leads to a contradiction 
 with the fact that \Gathering~does not solve the liveness of \GEW. 
 
 Therefore, during the 
 whole execution (after the execution of Rule \rTwoOne), the \headWalker,
 \tailWalker~and \minTailWalker~stay
 respectively \headWalker, \tailWalker~and \minTailWalker~and can only 
 execute Rule \rThreeOne~until their variables $walkSteps$ reach $n$, 
 and then they can only execute Rule \rFourThree. 

 Call $r_{1}$ and $r_{2}$ the two robots that are not involved in the 
 \towerElection~at time $t$. Since, by contradiction, neither Rule 
 \TerminationOne~nor Rule \TerminationTwo~are true, neither $r_{1}$ nor
 $r_{2}$ can meet the \headWalker~or the $tail$ while they (the \headWalker~and 
 the $tail$) are on a same node. Therefore, we assume that this event never
 happens.
 
 By Lemma~\ref{kindR1AndR2}, at time $t$,
 $state_{r_{1}} \in \{$\righter, \potentialMin, \dumbSearcher, \awareSearcher$\}$
 and $state_{r_{2}} \in \{$\dumbSearcher, \awareSearcher$\}$.
 
 Let us first consider all the possible interactions between only $r_{1}$ and 
 $r_{2}$ while $state_{r_{1}} \in \{$\righter, \potentialMin, \dumbSearcher, \awareSearcher$\}$
 and $state_{r_{2}} \in \{$\dumbSearcher, \awareSearcher$\}$.

 An isolated \potentialMin~or a \potentialMin~that is located only with a
 \dumbSearcher~stays a \potentialMin~and considers the $right$ direction (Rule
 \rOneEight). 
 
 If a \potentialMin~is located only with an \awareSearcher, it becomes an 
 \awareSearcher~and it executes the function \textsc{Search} (Rule \rOneFive). 
 
 An isolated \righter~stays a \righter~and considers the $right$ direction (Rule
 \rOneEight).
 
 If a \righter~is located only with a \dumbSearcher~(resp. an \awareSearcher), 
 it becomes an \awareSearcher~and executes the function \textsc{Search} (Rule
 \rOneSeven).
 
 If a \dumbSearcher~is located only with a \righter, it becomes an
 \awareSearcher~and executes the function \textsc{Search} (Rule\rOneNine). 
 
 If a \dumbSearcher~is located only with a \potentialMin~it stays a 
 \dumbSearcher~and executes the function \textsc{Search} (Rule \rOneEleven). In 
 this case, while executing the function \textsc{Search}, a 
 \dumbSearcher~considers the $left$ direction, since it
 possesses a greater identifier than the one of the \potentialMin. Indeed, only 
 Rule \rOneSix~permits a robot to become \potentialMin~or \dumbSearcher. 
 This rule is executed when $\mathcal{R} - 1$ \righter~are located on a same 
 node. While executing Rule \rOneSix, among the $\mathcal{R} - 1$ \righter, 
 the one with the minimum identifier becomes \potentialMin~while the others 
 become \dumbSearcher. By Observation~\ref{noMoreRighter}, Rule \rOneSix~can 
 be executed only once. Therefore, a \dumbSearcher~necessarily possesses an 
 identifier greater than the one of the \potentialMin.
 
 An isolated \dumbSearcher~or a \dumbSearcher~located only with another
 \dumbSearcher~stays a \dumbSearcher~and executes the function \textsc{Search} 
 (Rule \rOneEleven).
 
 If a \dumbSearcher~is located only with an \awareSearcher, it becomes an 
 \awareSearcher~and it executes the function \textsc{Search} (Rule \rOneTen). 
 
 An isolated \awareSearcher~or an \awareSearcher~located only with a \righter, a
 \potentialMin, a \dumbSearcher~or an \awareSearcher~stays an \awareSearcher~and 
 executes the function \textsc{Search} (Rule \rOneEleven).
 
 When $r_{1}$ and $r_{2}$ are on a same node without any other robot, executing
 the function \textsc{Search}, since all the robots possess distinct 
 identifiers, one considers the $right$ direction, while the other one considers the 
 $left$ direction. 
 
 While executing the function \textsc{Search} at time $i$, a robot that is an
 isolated \dumbSearcher~or an isolated \awareSearcher~considers during the Move 
 phase of time $i$ the same direction it considers during the Move phase of time
 $i - 1$. By Lemma~\ref{NotBotDirection3}, this direction cannot be equal to 
 $\bot$. 
 
 By the previous movements described, note that, as long as $r_{1}$ and
 $r_{2}$ are not located with the \headWalker~or the $tail$, they are always such that 
 $state_{r_{1}} \in \{$\righter, \potentialMin, \dumbSearcher, \awareSearcher$\}$
 and $state_{r_{2}} \in \{$\dumbSearcher, \awareSearcher$\}$.
  
 Now, consider the possible interactions between the \headWalker~and $r_{1}$
 and/or $r_{2}$ when $state_{r_{1}} \in \{$\righter, \potentialMin, \dumbSearcher, \awareSearcher$\}$
 and $state_{r_{2}} \in \{$\dumbSearcher, \awareSearcher$\}$.
  
 If $r_{1}$ and/or $r_{2}$, as a \righter, \potentialMin, \dumbSearcher~or 
 \awareSearcher~is on the same node as the \headWalker~such that 
 there is no adjacent $right$ edge to their location, then it executes Rule 
 \rOneThree, hence it becomes an \awareSearcher~and stops to move. 
 
 $(*)$ If $r_{1}$ and/or $r_{2}$, as a \righter, \potentialMin, \dumbSearcher~or 
 \awareSearcher~is on the same node as the \headWalker~such that
 there is an adjacent $right$ edge to their location, then it executes Rule 
 \rOneTwo, hence it becomes an \awareSearcher~considering the right direction and
 therefore crosses the adjacent right edge to its node.
 
 This implies that, as long as $r_{1}$ and $r_{2}$ are not located with the $tail$ they 
 are always such that $state_{r_{1}} \in \{$\righter, \potentialMin, \dumbSearcher, \awareSearcher$\}$
 and $state_{r_{2}} \in \{$\dumbSearcher, \awareSearcher$\}$.
 
 Finally, consider the possible interaction between the $tail$ and $r_{1}$
 and/or $r_{2}$. If $r_{1}$ and/or $r_{2}$, as a \righter, \potentialMin, 
 \dumbSearcher~or 
 \awareSearcher~is on the same node as the \minTailWalker, then it 
 executes Rule \rOneFour~and becomes a \tailWalker. From this time, by the
 function \textsc{BecomeTailWalker} and the rules of \Gathering, the robot
 belongs to the $tail$.

 We assume that there exists an eventual missing edge. Call $t"$ the time after 
 the execution of Rule \rTwoOne~and after the time when the
 eventual missing edge is missing forever. Consider the execution from $t"$. 
 Since Rule \rTwoOne~is executed before time $t"$, then there are
 \headWalker, \tailWalker~and \minTailWalker~in the execution after time $t"$
 included.
 
 Recall that, while executing Rules \rThreeOne~and \rFourThree, the 
 \headWalker~and the $tail$ are either on a same node or on two adjacent 
 nodes (the $tail$ is on the adjacent left node on the footprint of the dynamic 
 ring of the node where the \headWalker~is located).

 \begin{description}
  \item [Case 1:] \textbf{There is an eventual missing edge $\mathbf{e}$ between
  the node where the $\mathbf{\\headWalker}$ is located and the node where the 
  $\mathbf{tail}$ is located.}
  
  As explained previously, since the \headWalker~and the $tail$ are not on the same node, this necessarily
  implies that the \headWalker~either executes Rule \rThreeOne~or
  Rule \rFourThree~at time $t"$, and the $tail$ executes Rule \rThreeOne~at 
  time $t"$. Therefore, after time $t"$, the \headWalker~does not move either
  because it waits for the $tail$ to join it on its node (Rule \rThreeOne), or 
  because it executes the function \textsc{StopMoving} (Rule \rFourThree).
  Similarly, after time $t"$, the $tail$ does not move, since  it tries to join 
  the \headWalker~considering the right direction (Rule \rThreeOne), but the
  edge is missing forever.    

  Since there is at most one eventual missing edge in a \COT~ring, 
  all the edges, except $e$, are infinitely often present in the execution after 
  time $t"$. Considering the movements of the robots described previously, 
  whatever the direction considered by $r_{1}$ and $r_{2}$ at time $t"$ both of
  them succeed eventually to reach the node where the $tail$ is located, making 
  the liveness of \GEW~solved. 
 
  
  \item [Case 2:] \textbf{The eventual missing edge is not between the node 
  where the $\mathbf{\\headWalker}$ is located and the node where the 
  $\mathbf{tail}$ is located.}
  
  This implies that there exists a time from which the \headWalker~and the $tail$ 
  are located on a same node and do not move, either because they are executing
  Rule \rFourThree, or because they are executing Rule \rThreeOne~but
  the adjacent $right$ edge the \headWalker~tries to cross is the eventual 
  missing edge. In the second case, by the movements of the robots described
  previously, we succeed to prove that, eventually at most one of the robots 
  among $r_{1}$ and $r_{2}$ can be stuck on the extremity of the eventual
  missing edge where the \headWalker~and the $tail$ are not located, and that at
  least one of them succeeds to reach the node where the \headWalker~and the
  $tail$ are located, making the liveness of \GEW~solved.
  
  Consider now the first case. Call $t_{n}$ the first time at which the
  \headWalker~and the $tail$ are on a same node and both execute Rule 
  \rFourThree. If $r_{1}$ and $r_{2}$ consider the same direction at time 
  $t_{n}$, then by the movements of the robots described previously, whatever 
  the place of the eventual missing edge, we succeed to prove that, eventually
  at most one of them can be stuck on one of the extremity of the eventual
  missing edge, and that at least one of them succeeds to reach the node where
  the \headWalker~and the $tail$ are located, making the liveness of \GEW~solved. Similarly,
  if the \headWalker~and the $tail$ are located, at time $t_{n}$, on one of the extremity of the
  eventual missing edge, then, by the movements of the robots described 
  previously, we succeed to prove that, eventually at most one of the robots 
  among $r_{1}$ and $r_{2}$ can be stuck on the extremity of the eventual
  missing edge where the \headWalker~and the $tail$ are not located, and that at
  least one of them succeeds to reach the node where the \headWalker~and the
  $tail$ are located, making the liveness of \GEW~solved.
     
  Now consider the first case, when $r_{1}$ and $r_{2}$ consider opposed
  directions at time $t_{n}$ and such that, at time $t_{n}$, the \headWalker~and 
  the $tail$ are not located on one of the extremity of the eventual missing 
  edge. It is not possible for both $r_{1}$ and $r_{2}$ to be eventually stuck 
  on two
  different extremities of the eventual missing edge. Indeed, if $r_{1}$ and
  $r_{2}$ consider two opposed directions at time $t_{n}$, this is because,
  between times $t_{i}$ and $t_{n}$ (with $t_{i}$ the time at which the
  \headWalker~and the $tail$ both execute Rule \rThreeOne~for the first
  time), they 
  are located on a same node (without any other robot on their node). We prove this by contradiction. Assume, by contradiction, that $r_{1}$ 
  and $r_{2}$ are never located on a same node (without any other robot on their node) between 
  times $t_{i}$ and $t_{n}$. Consider the execution from time $t_{i}$ until time 
  $t_{n}$. Whatever
  the direction considered by $r_{1}$ (resp. 
  $r_{2}$), it cannot be located with the $tail$, otherwise, since there is
  no eventual missing edge between the \headWalker~and the $tail$ and by the 
  movements of the robots described previously, Rule \TerminationTwo~is 
  eventually executed. Therefore, $r_{1}$ (resp. $r_{2}$) can only be located with the \headWalker. 
  When $r_{1}$ (resp. $r_{2}$) is located with the \headWalker, it necessarily exists an 
  adjacent right edge to their position before the adjacent left edge to their
  position appears, otherwise, the $tail$ join them and Rule
  \TerminationTwo~is executed. By $(*)$, after $r_{1}$ (resp. $r_{2}$) is on the
  same node as the \headWalker~while there is an adjacent $right$ edge to their
  location, it becomes an \awareSearcher~considering the right direction. At time 
  $t_{n}$, the 
  \headWalker~and the $tail$ execute Rule \rFourThree, therefore they 
  succeed to execute Rule \rThreeOne~until their variables $walkSteps$ is
  equal to $n$. This implies that, if $r_{1}$ (resp. $r_{2}$) considers the left 
  direction at time $t_{i}$, 
  necessarily, since it cannot be located with $r_{2}$ (resp. $r_{1}$), by the movements of
  the robots described previously, it exists a time $t_{meet} \geq t_{i}$ at 
  which the \headWalker~and the $tail$ execute Rule \rThreeOne~and either 
  the \headWalker~or the $tail$ is located with it. As explained previously, $r_{1}$ (resp.
  $r_{2}$) cannot be located with the $tail$, this implies that, at time $t_{meet}$, 
  $r_{1}$ (resp. $r_{2}$) is located with the \headWalker. Therefore, whatever the 
  direction considered by $r_{1}$ (resp. $r_{2}$) at time $t_{i}$, if $r_{1}$
  and $r_{2}$ are never located on a same node (without any other robot on their node) between times 
  $t_{i}$ and $t_{n}$, it necessarily 
  considers the right direction at time $t_{n}$. Indeed, $r_{1}$ (resp. $r_{2}$)
  considers the right direction at time $t_{n}$ either because it meets the 
  \headWalker~that makes it consider the right direction or because at time 
  $t_{i}$ it considers the right direction and it is never located with the 
  \headWalker~and, by the movements of the robots described previously, it has 
  not change its direction between times $t_{i}$ and time $t_{n}$. Hence, there 
  is a contradiction with the fact that $r_{1}$ and $r_{2}$ consider opposite 
  directions at time $t_{n}$. Therefore, $r_{1}$ and $r_{2}$ consider two 
  opposite directions at time $t_{n}$ because they are located on a same node (without any other 
  robot on their node) between times $t_{i}$ 
  and $t_{n}$. 
  
  Consider the last time $t_{l}$ between times
  $t_{i}$ and $t_{n}$ at which $r_{1}$ and $r_{2}$ are located on a same node (without any
  other robot on their node). At time $t_{l}$, since the two robots are located on a same node, by the
  movements of the robots described, during the Move phase of time $t_{l}$ one
  considers the right direction while the other one considers the left 
  direction. By assumption, between times $t_{l} + 1$ and $t_{n}$, $r_{1}$ and 
  $r_{2}$ are not located on a same node. Moreover, as explained previously, between times $t_{l} + 1$
  and $t_{n}$, neither $r_{1}$ nor $r_{2}$ can be located with the $tail$, otherwise Rule \TerminationTwo~is eventually executed. Besides, 
  between times $t_{l} + 1$ and $t_{n}$ the robot that considers the left 
  direction during the Move phase of time $t_{l}$ cannot be located with the 
  \headWalker, otherwise, as noted previously, it considers the right direction
  at time $t_{n}$. Similarly, it is not possible for the robot that considers 
  the $right$ direction during the Move phase of time $t_{l}$ to be located with the \headWalker~between times $t_{l} + 1$ and 
  $t_{n}$, otherwise, by the movements of the robots described previously, this 
  necessarily implies that either it is also located on the same node as the
  $tail$ and therefore the liveness of \GEW~is solved or $r_{1}$ and $r_{2}$ are 
  on a same node and therefore the robot that considers the left direction
  during the Move phase of time $t_{l}$ is located with the \headWalker. Therefore,
  from time $t_{l} + 1$ to time $t_{n}$, $r_{1}$ and $r_{2}$ are isolated, hence,
  by the movements of the robots, they consider the same respective directions from
  the Move phase of time $t_{l}$ to time $t_{n}$.
  
  Assume, without lost of generality, that this is $r_{1}$ that considers the right direction from 
  the Move phase of time $t_{l}$ to time $t_{n}$. Call
  $v_{1}$ (resp. $v_{2}$) the node on which $r_{1}$ (resp. $r_{2}$) is located
  at time $t_{n}$. The explanations of the previous paragraph imply that $v_{1} \neq v_{2}$, and that, at time $t_{n}$, 
  the node where the \headWalker~and the $tail$ are located is in $Seg(v_{1}, v_{2})$.
  Therefore, since $r_{1}$ (resp. $r_{2}$) considers the right (resp. the left) direction
  at time $t_{n}$, by the movements of the robots and since it exists only one 
  eventual missing edge, this is not possible for these two robots to be eventually stuck
  on each of the extremities of the eventual missing edge. Hence, at least 
  one succeeds to reach the node where the \headWalker~and the $tail$ are located,
  making the liveness of \GEW~solved. 
 \end{description}   
\end{proof}

By Lemmas~\ref{safety} and \ref{theoremC5}, we can deduce the following
theorem which proves the correctness of Phases \Walk~and \WaitTermination.

\begin{theorem} \label{CorollaryC5}
 \Gathering~solves \GEW~in \COT~rings.
\end{theorem}

\subsection{What about \Gathering~executed in \AC, \RE, \BRE~and \ST~rings?} \label{gracefull}

In the previous subsection we prove that \Gathering~solves \GEW~in
\COT~rings. In this subsection, we consider \AC,
\RE, \BRE~and \ST~rings. For each of these
classes of dynamic rings, we give the version of gathering \Gathering~solves in it.

First, we consider the case of \AC~rings. In the following theorem,
we prove that \Gathering~solves \GW~in \AC~rings.

\begin{theorem} \label{theoremC9}
 \Gathering~solves \GW~in \AC~rings in $O(id_{r_{min}}*n^{2} + \mathcal{R} * n)$ rounds.
\end{theorem}

\begin{proof}
 By Corollary~\ref{CorollaryC5}, \Gathering~solves \GEW~in 
 \COT~rings, since $\AC\subset\COT$, this implies that \Gathering~also solves
 \GEW~in \AC~rings. Therefore, to prove that
 \Gathering~solves \GW~in \AC~rings, it stays to
 prove that each phase of \Gathering~is bounded.
 
 \begin{description}
  \item [Phase \AmITheMin:]
  
  By Corollary~\ref{uniqueLeader}, only $r_{min}$ becomes $min$ in finite time. 
  By the rules of \Gathering, when $r_{min}$ becomes $min$, it is first 
  \minWaitingWalker~before being \minTailWalker~(since
  only a \minWaitingWalker~can become a \minTailWalker~while executing Rule \rTwoOne). Therefore, since only Rule 
  \rOneOne~permits a robot to become \minWaitingWalker, by the predicate 
  $Min\-Discovery()$ of this rule, $r_{min}$ becomes $min$ either because it moves during
  $4 * n * id_{r_{min}}$ steps in the right direction or because it meets a 
  robot that permits it to deduce that it is $min$. In this last case, note that, 
  either $r_{min}$ is \potentialMin, or $r_{min}$ meets a \potentialMin~or a 
  \dumbSearcher~or a robot whose variable $id\-Min$ is different from $-1$. 
  Therefore, in this last case, either $r_{min}$ possesses a variable $id\-Po\-ten\-tial\-Min$ 
  different from $-1$, or $r_{min}$ meets a robot $r$ such that
  $id\-Po\-ten\-tial\-Min_{r}$ is different from $-1$ (since a \potentialMin~and 
  a \dumbSearcher~have their variable $id\-Po\-ten\-tial\-Min$ different from 
  $-1$ (Rule \rOneSix) and since, while executing \Gathering, each time the variable
  $id\-Min$ of a robot is set with a variable different from $-1$, this is also the case 
  for its variable $id\-Po\-ten\-tial\-Min$).
  
  Taking back the arguments used in the proof of Lemma~\ref{atLeastOneLeader}, 
  let us consider the following cases.

  \begin{description}
  \item [Case 1.1:] \textbf{Rule \rOneSix~is never executed.}
  
  By the rules of \Gathering, this implies that, before the time when $r_{min}$
  is $min$, there are only \righter~in the execution. First, this implies 
  that $r_{min}$ becomes $min$ because it moves during $4 * n * id_{r_{min}}$ 
  steps in the $right$ direction (since \righter~robots have their variables 
  $id\-Po\-ten\-tial\-Min$ equal to $-1$). Second, in this context, as long as
  $r_{min}$ is not $min$, all the \righter~always consider the $right$ direction
  (Rule \rOneEight). This implies that, as long as $r_{min}$ is not $min$, each 
  time a robot wants to move in the 
  right direction it can be stuck during at most $n$ rounds, otherwise, since in 
  an \AC~ring at most one edge can be missing at each instant time, 
  Rule \rOneSix~is executed. Therefore in case 1.1 $r_{min}$ becomes $min$
  in at most $4 * id_{r_{min}} * n * n$ rounds. 
  \end{description}
  
  Now let consider the case where Rule \rOneSix~is executed at a time $t$. In the following, we consider the 
  execution from time $t$.
  After time $t$, while it is not yet $min$, if $r_{min}$ is stuck more than $4 * n$ consecutive rounds on
  a same node then it becomes $min$. We prove this considering the two following
  cases. In each of these cases we assume that $r_{min}$ is not yet $min$ and that it is stuck more than 
  $n$ rounds on a same node.
  
  \begin{description}
  \item [Case 1.2:] \textbf{Rule \rOneSix~is executed but $\mathbf{r_{min}}$ 
  is not among the $\mathcal{R} - 1$ $\mathbf{righter}$ robots that execute it.}  
  
  Taking back the arguments of the proof of Lemma~\ref{atLeastOneLeader}, we 
  know that Rule \rOneSix~can be executed only once, and that the robots 
  that execute this rule can never be \righter~anymore. Moreover, since 
  $r_{min}$ does not execute Rule \rOneSix, since, by 
  Corollary~\ref{uniqueLeader}, $r_{min}$ necessarily becomes $min$, since,
  by Lemma~\ref{OnlyMovingRightAndPLeaderCanBeMin}, only \righter~and
  \potentialMin~can be $min$, and since only Rule \rOneSix~permits robots to
  become \potentialMin, before becoming $min$, $r_{min}$ is a \righter. By the 
  proof of Lemma~\ref{atLeastOneLeader}, as long as $r_{min}$ is not $min$ it cannot
  exist \awareSearcher. Hence, by the rules of \Gathering, as long as $r_{min}$ is not 
  $min$, there are only one \righter, one \potentialMin~and $\mathcal{R} - 2$ \dumbSearcher~in the 
  execution. Therefore, by the rules of \Gathering, the \potentialMin~is \potentialMin~at least until 
  $r_{min}$ becomes $min$. Hence, the \potentialMin~executes Rule \rOneEight~and thus
  considers the right direction at least until $r_{min}$ becomes $min$. We have assumed
  that, while it is not yet $min$, $r_{min}$ is stuck more than $4 * n$ consecutive rounds on a same node. Since
  $r_{min}$ is a \righter~before being $min$, it is stuck because the adjacent 
  right edge to its position is missing (Rule \rOneEight). Therefore, since in an \AC~ring
  of size $n$ at least $n - 1$ edges are present at each 
  instant time, either the \potentialMin~(or a \dumbSearcher) meets $r_{min}$ in 
  at most $n$ rounds. When $r_{min}$ meets a \potentialMin~(or
  a \dumbSearcher), it becomes $min$ by definition of the predicate 
  $MinDiscovery()$ in Rule \rOneOne. Therefore, if it is stuck more than $4 * n$ rounds, 
  $r_{min}$ becomes $min$ in at most $n$ rounds.  
  
  
  \item [Case 1.3:] \textbf{Rule \rOneSix~is executed and $\mathbf{r_{min}}$ 
  is among the $\mathcal{R} - 1$ $\mathbf{righter}$ robots that execute it.}
  
  In this case, by Rule \rOneSix, $r_{min}$ becomes \potentialMin. By 
  Observations~\ref{noMoreRighter} and \ref{noMoreMovingRightAndPLeader}, by Corollary~\ref{uniqueLeader} and
  by Lemma~\ref{OnlyMovingRightAndPLeaderCanBeMin}~$r_{min}$ is \potentialMin~until
  it becomes $min$. Therefore, $r_{min}$, while it is not yet $min$, can be stuck 
  only because the adjacent right edge to its position is missing (Rule \rOneEight).
  
  First, consider that at the time when $r_{min}$, as a \potentialMin, is stuck
  more than $4 * n$ rounds, there does not exist \righter~in the execution. By 
  Observation~\ref{noMoreRighter}, there is no more \righter~in the execution.
  However, at the time when Rule \rOneSix~is executed, the robot $r$ that is not among 
  the robots that execute this rule is a \righter. Therefore, necessarily $r$, as a \righter,
  meets at least one \dumbSearcher~at a time $t'$. Indeed, 
  it cannot meet the \potentialMin, otherwise $r_{min}$ is $min$ (Rule \rOneOne), and thus it is not anymore 
  \potentialMin~at the time at which it is stuck. Moreover, $r$ 
  cannot be isolated forever after time $t$, otherwise it stays a \righter~(Rule \rOneEight).
  Hence, at time $t'$, $r$ becomes an \awareSearcher~(Rule \rOneSeven). Consider an \awareSearcher~$r_{a}$ of the execution. By 
  Lemma~\ref{NotBotDirection}, $r_{a}$ cannot consider the $\bot$ direction. 
  Moreover, by the rules of \Gathering, as long as there is no $min$, an 
  \awareSearcher~executes the function \textsc{Search} (rule \rOneEleven).
  Besides, by the proof of Lemma~\ref{atLeastOneLeader} if a robot is not isolated and executes the 
  function \textsc{Search}, then all the robots of its node are or become 
  \awareSearcher~and execute the function \textsc{Search}. While executing the
  function \textsc{Search}, an isolated robot does not change its direction. When a robot
  executes the function \textsc{Search} while there are multiple robots on its 
  node, if it possesses the maximum identifier among the robots of its node, it 
  considers the left direction, otherwise it considers the right direction. In 
  an \AC~ring of size $n$, at least $n - 1$ edges are present at 
  each instant time. Therefore, if $r_{a}$ considers the $right$ direction, 
  either it, as an \awareSearcher~or a robot that is or becomes an 
  \awareSearcher~is located, in at most $n$ rounds, on the node where $r_{min}$, 
  as a \potentialMin, is stuck. In
  the case where $r_{a}$ considers the $left$ direction then, by the same 
  arguments, in at most $4 * n$ rounds an \awareSearcher~is located on the node
  where $r_{min}$, as a \potentialMin, is stuck. Indeed, at most $n$ rounds are needed for an 
  \awareSearcher to reach the extremity of the missing edge where $r_{min}$ is
  not located. Then, at most $2 * n$ other rounds are needed for a 
  \dumbSearcher~(execution of the function \textsc{Search}, rule \rOneEleven) or
  an \awareSearcher~to reach also this node. These $2*n$ rounds are especially 
  needed for a \dumbSearcher~that may take $n$ rounds (considering the left 
  direction) to reach the node where $r_{min}$~is stuck and then again
  $n$ rounds (considering the right direction) to reach the other extremity of the 
  missing edge. 
  From this time there is in the 
  execution an \awareSearcher~that considers the $right$ direction. Finally, at 
  most $n$ supplementary rounds are needed for an \awareSearcher~to reach the
  node where $r_{min}$, as a \potentialMin, is stuck. Note that
  $\mathcal{R} > 4$, and there are $\mathcal{R} - 1$ 
  \dumbSearcher/\awareSearcher~in the execution as long as $r_{min}$ is not 
  $min$. Therefore, the previous scenario can effectively happen. When $r_{min}$ 
  meets an \awareSearcher, it becomes $min$ by definition of the predicate
  $MinDiscovery()$ of rule \rOneOne. Therefore, $r_{min}$ becomes $min$ in at 
  most $4 * n$ rounds if it is stuck more than $4 * n$ rounds. 
  
  Second, consider that at the time when $r_{min}$, as a \potentialMin, is stuck more than $4 * n$ 
  rounds, there exists a \righter. In this case, since an isolated 
  \righter~considers the $right$ direction (Rule \rOneEight), and by the
  arguments of the previous paragraph, either a \righter~or a robot that is an
  \awareSearcher~or that becomes an \awareSearcher (Rules \rOneSeven, 
  \rOneNine~or \rOneTen) meets $r_{min}$ in at most $n$ rounds. When $r_{min}$ 
  meets a \righter~or an \awareSearcher, it becomes $min$ by definition of the 
  predicate $Min\-Discovery()$ of Rule \rOneOne. Therefore, $r_{min}$ becomes
  $min$ in at most $n$ rounds if it is stuck more than $4 * n$
  rounds. 
 \end{description}
 
 Now, we give the worst number of rounds needed for $r_{min}$ to become $min$, 
 in the case where there exists a time $t$ at which Rule \rOneSix~is executed.
 By Case 1.1, before time $t$, $r_{min}$, while it is not yet $min$, 
 can be stuck at most $n$ rounds each time it moves from one step in the $right$
 direction. Similarly, by the two previous cases (Case 1.2 and 1.3), after time
 $t$, $r_{min}$, while it is not yet $min$, can be stuck at most $4 * n$ rounds 
 each time it moves from one step in the $right$ direction. Let $nb$ be the 
 number of steps in the right direction moved by $r_{min}$ before time $t$. 
 As proved previously, $r_{min}$ is either a \righter~or a \potentialMin~before
 being $min$. By Lemma~\ref{rightDirection}, this implies that before being 
 $min$, $r_{min}$ always considers the right direction. Therefore, by the 
 predicate $Min\-Discovery()$ of Rule \rOneOne,
 in at most $nb * n + ((4 * id_{r_{min}} * n) - nb) * 4 * n$ 
 rounds, $r_{min}$ becomes $min$ because it moves during $4 * id_{r_{min}} * n$ 
 steps in the right direction. This function is maximal when $nb = 0$, therefore in at most
 $16 * id_{r_{min}} * n^{2}$ rounds $r_{min}$ becomes $min$ because it moves
 during $4 * id_{r_{min}} * n$ steps in the right direction. Now consider the case 
 where $r_{min}$ becomes $min$ because it meets a robot that permits it to deduce 
 that it is $min$. Once $r_{min}$ is stuck more than $4 * n$ rounds after time $t$, we have seen that it becomes 
 $min$. Since we consider the worst case such that $r_{min}$ does not become 
 $min$ because it moves during $4 * id_{r_{min}} * n$ steps in the right 
 direction, this implies that in at most $(4 * id_{r_{min}} * n - 1) * 4 * n$ 
 rounds $r_{min}$ becomes $min$. Therefore, whatever the situation, Phase 
 \AmITheMin~is bounded.
 
 Moreover, by Case 1.1 and the previous paragraph, we can conclude that 
 Phase \AmITheMin~is in $O(id_{r_{min}} * n^{2})$ rounds.

  Now we consider Phase \MinWaitToBeKnown~of \Gathering. In this
  phase $r_{min}$ is $min$ and waits for a \towerElection~to be formed. We take 
  back the arguments used in the proofs of 
  Lemmas~\ref{PotentialMinAndTowerElection} and 
  \ref{NoPotentialMinAndTowerElection} to prove that this phase is bounded.
  
  \item [Phase \MinWaitToBeKnown:]
  
  \begin{description}
   \item [Case 2.1:] \textbf{There is a $\mathbf{potentialMin}$ in the 
   execution.}
   
   For this case we take back the arguments of the proof of Lemma~\ref{PotentialMinAndTowerElection}.
    
   If before being $min$, $r_{min}$ is a \righter, then all the robots that are
   not located on node $u$ are \potentialMin, \dumbSearcher, and \awareSearcher.   
   As long as it is not on node $u$, a \potentialMin~either executes Rule \rOneEight,
   or it becomes an \awareSearcher~(Rule \rOneFive). While executing Rule
   \rOneEight, a \potentialMin~stays a \potentialMin~and has the same behavior as if it was executing the 
   function \textsc{Search}. Moreover, as long as they are not on node $u$, 
   \dumbSearcher~and \awareSearcher~robots stay either \dumbSearcher~or
   \awareSearcher~and execute the function \textsc{Search}.
   Therefore, by definition of the function \textsc{Search} (refer to Phase 
   \AmITheMin~case 1.3 of this proof) and by Lemma~\ref{lemma4Bis}, at most $3 * n$ rounds are needed (in \AC~rings) for a 
   robot $r$ such that $state_{r} \in \{$\potentialMin, \dumbSearcher, \awareSearcher$\}$
   to be located on node $u$. Indeed, these $3 * n$ rounds are needed especially when 
   a \potentialMin, \dumbSearcher~or \awareSearcher~moves in one direction during 
   $n$ steps and then is stuck on the adjacent node of $u$, then $n$ steps are
   needed for a robot of this kind to be also located on this node and thus to consider
   the opposite direction, then in at most $n$ additional steps a robot of this kind is located on $u$.
   By Rule \rTwoThree, this implies that at most $3 * n$ rounds are
   necessary for a supplementary \waitingWalker~to be located on node $u$. 
   Therefore, at most $(\mathcal{R} - 3) * 3 * n$ rounds are needed for a 
   \towerElection~to be formed.
   
   Now consider the case where before being $min$, $r_{min}$ is a \potentialMin.
   
   In this case among the robots that are not on node $u$, there are \dumbSearcher,
   \awareSearcher~and at most one \righter.
   
   For all the cases of Case 2.1 of the proof of 
   Lemma~\ref{PotentialMinAndTowerElection}, at most 
   $(\mathcal{R} - 4) * 3 * n$ rounds are needed for $\mathcal{R} - 4$ 
   \waitingWalker~to be located on $u$ (for the same reasons as the one explained in the previous paragraph).
   Then among the robots that are not on node $u$, it exists at most one
   \righter, and 2 robots that are either \dumbSearcher~or \awareSearcher. At most $n$ 
   rounds are needed for the \righter~to be stuck on the node called $v$ in the proof
   of Lemma~\ref{PotentialMinAndTowerElection}, and then at most $n$ rounds are 
   needed for a \dumbSearcher~or an \awareSearcher~to be also located on node $v$
   (and thus, by Rule \rOneSeven, for all the robots that are not on node $u$ to be either \dumbSearcher~or \awareSearcher), 
   and then at most $n$ additional rounds are needed for one of the robot to 
   reach node $u$. Therefore, for all the cases of Case 2.1 of the proof of 
   Lemma~\ref{PotentialMinAndTowerElection}, at most $(\mathcal{R} - 4) * 3 * n + 3 * n$
   rounds are needed for a \towerElection~to be formed. 
   
   If we consider Case 2.2 of the proof of
   Lemma~\ref{PotentialMinAndTowerElection}, similarly as in the previous case,
   at most $(\mathcal{R} - 4) * 3 * n + 3 * n$ rounds are needed for Rule 
   \TerminationTwo~to be executed.
  
  \item [Case 2.2:] \textbf{There is no $\mathbf{potentialMin}$ in the 
  execution.}
 
   For this case we take back the arguments of the proof of Lemma~\ref{NoPotentialMinAndTowerElection}.
   
   Just after $r_{min}$ becomes $min$, it takes at most $n * n$ rounds for a 
   robot $r$ to join the node where $r_{min}$ is located. Indeed, as long as no robot
   is on node $u$ with $r_{min}$, as a \minWaitingWalker, all the robots except $r_{min}$ are \righter. By 
   the same arguments as the one used in Phase \AmITheMin~Case 1.1 of 
   this proof, a
   \righter~cannot be stuck more than $n$ rounds on the same node, otherwise
   Rule \rOneSix~is executed, which is a contradiction with the fact that there 
   is no \potentialMin. Moreover, a \righter~can move from at most $n$ steps in
   the right direction to reach $u$.
   
   Once $r$ is on node $u$ an adjacent right edge to $u$ is present in at most 
   $n$ rounds, otherwise Rule \TerminationOne~is executed. Therefore, once 
   $r$ is on node $u$, in at most $n$ rounds it becomes an \awareSearcher. 
   From this time, either it is possible for all the \righter~to become
   \awareSearcher~or it exists at least one \righter~that is stuck on node $u$.
   In the first case at most $2 * n$ rounds are needed for 
   all the \righter~to become \awareSearcher~(either because an
   \awareSearcher~meets them, or because they are located on $u$ such that there
   is an adjacent right edge to $u$). By the arguments above, we
   know that if all the robots that are not located on node $u$ are
   \awareSearcher, and if there are more than 3 such robots, then in a most
   $3 * n$ rounds one robot of this kind reaches node $u$. Therefore, for 
   $\mathcal{R} - 3$ \waitingWalker~to be located on node $u$, with $r_{min}$, at most
   $(\mathcal{R} - 3) * 3 * n$ supplementary rounds are needed. 
   In the second case, at most $2 * n$ rounds are needed for some of the 
   \righter~to reach node $u$ (and to be stuck on this node). Since the robots 
   that are not on node $u$ are \awareSearcher~and since at least one 
   \righter~is stuck on node $u$, by the same arguments as above, at most 
   $(\mathcal{R} - 3) * 3 * n$ additional rounds are needed for
   Rule \TerminationTwo~to be executed. 
   
   Therefore, in this case, at most $n*n+n+2*n+(\mathcal{R} - 3) * 3 * n$ rounds 
   are needed for either Phase \MinWaitToBeKnown~to be achieved or Rule
   \TerminationTwo~to be executed.
 \end{description}

  Therefore Phase \MinWaitToBeKnown~is bounded. Moreover, by the two previous 
  cases, we can conclude that Phase \MinWaitToBeKnown~is in $O(\mathcal{R} * n + n^{2})$ rounds.
  
  Now we consider Phase \Walk~of \Gathering. In this purpose we take 
  back the arguments used in the proof of Lemma~\ref{theoremC5}.
  
  \item [Phase \Walk:]
  
  Here we consider the worst execution in terms of times. Therefore, we consider
  that Rules \TerminationOne~and \TerminationTwo~are executing at the very 
  last moment.
  The robots $r_{1}$ and $r_{2}$ that are not involved in $T$ at time $t_{tower}$
  are such that $state_{r_{1}} \in \{$\righter, \potentialMin, \dumbSearcher, \awareSearcher$\}$
  and $state_{r_{2}} \in \{$\dumbSearcher, \awareSearcher$\}$. Therefore, as 
  explained previously, each time the \headWalker, or the \minTailWalker~/ 
  \tailWalker~robots move from one steps in the right direction, they can be 
  stuck at most during $3 * n$ rounds, otherwise either Rule 
  \TerminationOne~or Rule \TerminationTwo~is executed. Indeed, this is especially
  the case when the \headWalker~and the \minTailWalker~/ \tailWalker~are stuck on
  the same node. In fact, it takes at most $n$ rounds for $r_{1}$ to be stuck on
  the other extremity the missing edge. At most $n$ supplementary rounds are needed
  for $r_{2}$ to reach the node where $r_{1}$ is stuck (and therefore for one robot
  to change its direction), and then $n$ other rounds are needed for one of these
  robots to reach the node where the \headWalker~and the \minTailWalker~/ \tailWalker~are stuck
  (and thus for Rule \TerminationTwo~to be executed).
  Therefore, Phase \Walk~is achieved in
  at most $2 * n * (3 * n)$ rounds since the \headWalker~and the 
  \minTailWalker~/ \tailWalker~robots have to move alternatively during 
  $n$ steps to complete Phase \Walk. In other words, Phase \Walk~is bounded and
  is in $O(n^{2})$ rounds.
  
  Now we consider Phase \WaitTermination~of \Gathering. In this purpose we take back the 
  arguments used in the proof of Lemma~\ref{theoremC5}.
  
  \item [phase \WaitTermination:]
  
   Using similar arguments as the one used in Phase \Walk, once the \headWalker~and the \minTailWalker~/
   \tailWalker~stop to move forever, if they are located on a same node, at most
   $3 * n$ rounds are necessary for Rule \TerminationTwo~to be executed. 
   In the case where the \headWalker~and the \minTailWalker~/
   \tailWalker~stop to move forever, if they are located on different nodes, at most
   $2 * n + 2 * n$ rounds are necessary for Rule \TerminationTwo~to be
   executed. Indeed, at most $2 * n$ rounds are necessary for each of the two robots that
   are not involved in $T$ at time $t_{tower}$ to be located on the node where 
   the \minTailWalker~/\tailWalker~is located. This is true whatever the 
   interactions between $r_{1}$ and $r_{2}$ and whatever the interactions between 
   $r_{1}$ (resp. $r_{2}$) and the \headWalker~since in an \AC~ring there is at 
   most one edge missing at each instant time (and in this precise case the missing
   edge is between the node where the \headWalker~is located and the node where 
   the \minTailWalker~/ \tailWalker~are located). Therefore, Phase \WaitTermination~is
   bounded and is in $O(n)$ rounds.
 \end{description}

 In conclusion each of the four phases of algorithm \Gathering~are bounded when 
 executed in an \AC~ring, therefore \Gathering~solves \GW~in \AC~rings.
 Moreover, \Gathering~solves \GW~in \AC~rings in $O(id_{r_{min}}*n^{2} + \mathcal{R} * n)$ rounds.
\end{proof}

Now, we consider the case of \RE~rings. In the following theorem,
we prove that \Gathering~solves \GE~in \RE~rings.

\begin{theorem} \label{theoremC6}
 \Gathering~solves \GE~in \RE~rings.
\end{theorem}

\begin{proof}
 By Corollary~\ref{CorollaryC5}, \Gathering~solves \GEW~in 
 \COT~rings, therefore it solves the safety and the liveness of \GEW~in \COT~rings. Since $\RE\subset\COT$, 
 \Gathering~also solves the safety and the liveness of \GEW~in \RE~rings. This implies 
 that all robots that terminate their execution terminate it on the same node 
 and it exists a time at which at least $\mathcal{R} - 1$ robots terminate
 their execution. Call $t$ the first time at which at least $\mathcal{R} - 1$ 
 robots terminate their execution. 
 
 By contradiction, assume that \Gathering~does not solve \GE~in
 \RE~rings, this implies that it exists a robot $r$ that never
 terminates its execution.

 Call \towerTermination~the $\mathcal{R} - 1$ robots that, at time $t$, are 
 located on a same node and are terminated. While executing \Gathering, the only
 way for a robot to terminate its execution is to execute either Rule 
 \TerminationOne~or Rule \TerminationTwo. By Lemma~\ref{term}, for a 
 \towerTermination~to be formed at time $t$, Rule \TerminationTwo~has to be 
 executed at this time.
 
 $(*)$ By the predicate of Rule \TerminationTwo, $r_{min}$ belongs to the 
 \towerTermination. By Lemma~\ref{term}, all the robots that are located on the 
 same node as $r_{min}$ at time $t$ belong to the \towerTermination. 
 
 Call $w$ the node where the \towerTermination~is located at time $t$.
 
 Note that $r$ cannot be located on node $w$ after time $t$ included, otherwise
 it executes Rule \TerminationOne~and the lemma is proved.
 
 Since $r_{min}$ belongs to the \towerTermination, and since by 
 Corollary~\ref{uniqueLeader}, only $r_{min}$ can be \minWaitingWalker~or a 
 \minTailWalker, $r$ is neither \minWaitingWalker~nor \minTailWalker.
 
 At time $t$, $r$ cannot be a \tailWalker. Indeed, to become a \tailWalker, a 
 robot must either execute Rule \rTwoOne~or Rule \rOneFour. To execute
 Rule \rTwoOne~a robot must be a \waitingWalker. By 
 Lemma~\ref{waitingCautiousWalkOnSameNode}, all \waitingWalker~are located on
 the same node as the \minWaitingWalker. Moreover, when a 
 \waitingWalker~executes Rule \rTwoOne, by the predicate 
 $All\-But\-Two\-Waiting\-Walker()$, the \minWaitingWalker~also executes this 
 rule becoming a \minTailWalker. Then by the rules of \Gathering, the robot 
 that becomes \tailWalker~while executing Rule \rTwoOne~and the
 \minTailWalker~execute the same movements (refer to Rules \rThreeOne~and 
 \rFourThree), and therefore are always on a same node. Besides, to execute
 Rule \rOneFour~a robot must be located on the same node as the
 \minWaitingWalker~(refer to the predicate 
 $Not\-Walker\-With\-Tail\-Walker(r')$). Then, thanks to the function 
 \textsc{BecomeTailWalker} and by the rules of \Gathering, the robot that 
 becomes \tailWalker~while executing Rule \rOneFour~cannot be
 on a node different from the one where the \minTailWalker~is located (refer to
 Rules \rThreeOne~and \rFourThree). Therefore, by $(*)$, $r$ cannot be a 
 \tailWalker~at time $t$, otherwise, at time $t$, it is on the same node as the 
 \minTailWalker~(thus, by Corollary~\ref{uniqueLeader}, it is on the same node
 as $r_{min}$) and hence it terminates its execution.
 
 At time $t$, $r$ cannot be a \waitingWalker~robot. Indeed by the rules of 
 \Gathering~and the previous remarks, it cannot exists \waitingWalker~if there 
 is no \minWaitingWalker~in the execution, and by
 Lemma~\ref{waitingCautiousWalkOnSameNode} all the \waitingWalker~and 
 \minWaitingWalker~are located on a same node. Therefore, by $(*)$, $r$ cannot
 be a \waitingWalker~at time $t$, otherwise, at time $t$, it is on the same node
 as the \minWaitingWalker~(thus, by Corollary~\ref{uniqueLeader}, it is on the
 same node as $r_{min}$) and hence it terminates its execution.
 
 Therefore, at time $t$, $r$ can be either a \righter, a \potentialMin, a
 \dumbSearcher, an \awareSearcher, a \headWalker~or a \leftWalker~robot.
 
 As long as $r$ is not on node $w$, it is isolated.
 
 An isolated \righter~or an isolated \potentialMin~only executes Rule 
 \rOneEight. While executing this rule, a robot considers the $right$ direction
 and stays a \righter~or a \potentialMin. Since all the edges are infinitely 
 often present, such a robot is infinitely often able to move in the $right$
 direction until reaching the node $w$. 
 
 An isolated \dumbSearcher~or an isolated \awareSearcher~only executes Rule 
 \rOneEleven. While executing this rule, an isolated robot stays a 
 \dumbSearcher~or an \awareSearcher, and considers the direction it considers 
 during the previous Move phase. By Lemma~\ref{NotBotDirection3}, this direction
 cannot be equal to $\bot$. Therefore, an isolated \dumbSearcher~or an isolated
 \awareSearcher~always considers the same direction $d$ (either $right$ or
 $left$). Since all the edges are infinitely often present, such a robot is 
 infinitely often able to move in the direction $d$ until reaching the node $w$.  
 
 Now assume that, at time $t$, $r$ is a \leftWalker. A \leftWalker~only executes
 Rule \rFourOne. While executing this rule, a robot considers the $left$
 direction and stays a \leftWalker. Since all the edges are infinitely often 
 present, such a robot is infinitely often able to move in the $left$ direction
 until reaching the node $w$.
 
 Now assume that, at time $t$, $r$ is a \headWalker. A \headWalker~can execute
 either Rule \rFourTwo~or Rule \rFourThree~or Rule \rThreeOne. While 
 executing Rule \rFourTwo, a \headWalker~becomes a \leftWalker, then, by the 
 previous paragraph, $r$ reaches the node $w$ in finite time. Consider now the 
 cases where, at time $t$, $r$ executes either Rule \rFourThree~or Rule
 \rThreeOne. In these cases, after time $t$, it necessarily exists a 
 time at which $r$ executes Rule \rFourTwo. Assume, by contradiction, that
 this is not true. The only way for a robot to become a \headWalker~is to
 execute Rule \rTwoOne. Rule \rTwoOne~is executed when $\mathcal{R} - 2$
 robots are located on a same node. While executing this rule, a robot sets its
 variable $walker\-Mate$ with the identifiers of the robots that are located on 
 its node. Only Rule \rTwoOne~permits a robot to update its variable 
 $walker\-Mate$. Note that, since $\mathcal{R} - 2 \geq 2$, the variable
 $walker\-Mate$ of $r$, after time $t$, contains at least one identifier $i$
 different from the
 identifier of $r$. The robot of identifier $i$ necessarily belongs to the
 \towerTermination, since only $r$ does not terminate. $(1)$ Hence, at time $t$,
 the robot of identifier $i$ is terminated on node $w$, thus it does not move, 
 and therefore, after time $t$, $r$ is never on the same node as $i$. $(2)$ All
 the edges are infinitely often present. While executing Rule \rFourThree~at
 time $t$, $r$ considers the $\bot$ direction and does not update its other 
 variables. $(3)$ Hence, by the rules of \Gathering, since $r$ cannot execute 
 Rule \rFourTwo, after time $t$, $r$ can only execute Rule \rFourThree, 
 and therefore only considers the $\bot$ direction. Hence, necessarily by $(1)$,
 $(2)$ and $(3)$, this implies that, after time $t$, it exists a time at which
 the predicate $Head\-Walker\-Without\-Walker\-Mate()$ is true, thus at this 
 time Rule \rFourTwo~is executed. Similarly, if at time $t$, $r$ executes 
 Rule \rThreeOne, since $r$ can never be located on the same node as $i$,
 while executing Rule \rThreeOne, it considers the $\bot$ direction and does
 not update its other variables. $(4)$ Hence, by the rules of \Gathering, since 
 $r$ cannot execute Rule \rFourTwo, after time $t$, $r$ can only execute
 Rule \rThreeOne, and therefore always considers the $\bot$ direction. Thus, by
 $(1)$, $(2)$ and $(4)$, necessarily, after time $t$, it exists a time at which 
 the predicate $Head\-Walker\-Without\-Walker\-Mate()$ is true, hence at this 
 time Rule \rFourTwo~is executed. Therefore, even in the cases where, at
 time $t$, $r$ executes either Rule \rFourThree~or Rule \rThreeOne, it
 exists a time greater than $t$ at which $r$ becomes a \leftWalker~and hence, by
 the previous paragraph, $r$ succeeds to reach the node $w$ in finite time.
 
 Therefore whatever the kind of robot $r$ is, it is always able to reach the 
 node $w$. Once $r$ reaches the node $w$ it executes Rule 
 \TerminationOne~making \GE~solved.
\end{proof}

Now, we consider the case of \BRE~rings. We prove, in 
Theorem~\ref{theoremC7}, that \Gathering~solves \G~in
\BRE~rings. To prove this, we first need to prove the following 
lemma that it useful to bound Phase \MinWaitToBeKnown~of \Gathering~in
\BRE~rings.

\begin{lemma} \label{searchBounded}
 If the ring is a \BRE~ring and if there is no \towerElection~in the
 execution but there exists at a time $t$ at least 3 robots such that they are
 either \potentialMin, \dumbSearcher~or \awareSearcher, then at least a 
 \potentialMin, a \dumbSearcher~or an \awareSearcher~reaches the node $u$ 
 between time $t$ and time $t + n * \delta$ included, with $\delta \geq 1$.
\end{lemma}

\begin{proof}
 We prove this lemma using the arguments of the proof of Lemma~\ref{search} and 
 the fact that in a \BRE~ring each edge appears at least once every
 $\delta$ units of time. 
\end{proof}

\begin{theorem} \label{theoremC7}
 \Gathering~solves \G~in \BRE~rings in $O(n * \delta * (id_{r_{min}} + \mathcal{R}))$ rounds.
\end{theorem}

\begin{proof} 
 By Lemma~\ref{theoremC6}, \Gathering~solves \GE~in \RE~rings. Therefore, since
 $\BRE\subset\RE$, then \Gathering~also solves \GE~in 
 \BRE~rings.
 We want to prove that \Gathering~solves \G~in~\BRE~rings. Therefore, we have to
 prove that each phase of the algorithm is bounded.

 \begin{description}
  \item [Phase \AmITheMin:]

 By Corollary~\ref{uniqueLeader}, we know that only $r_{min}$ becomes $min$ in 
 finite time. By Lemma~\ref{OnlyMovingRightAndPLeaderCanBeMin}, before being 
 $min$, $r_{min}$ is either a \righter~or a \potentialMin~robot. By
 Lemma~\ref{rightDirection}, if, at a time $t$, a robot is a \righter~or a
 \potentialMin~robot, then it considers the $right$ direction from the beginning
 of the execution until the Look phase of time $t$. Since initially all the
 robots are \righter, and since, by the rules of \Gathering, only \righter~can
 become \potentialMin~(refer to Rule \rOneSix), then by 
 Observations~\ref{noMoreRighter} and \ref{noMoreMovingRightAndPLeader}, a robot
 that is a \righter~(resp. \potentialMin) is a \righter~(resp. is either a 
 \righter~or a \potentialMin) since the beginning of the execution. Besides,
 each time $r_{min}$, as a \righter~or as a \potentialMin, crosses an edge in 
 the right direction, it increases its variable $rightSteps$ of one (refer to 
 Rules \rOneEight~and \rOneSix). 
 Therefore, since each edge of the footprint of a \BRE ring is 
 present at least once every $\delta$ units of time, by definition of $min$ and 
 of the predicate $Min\-Discovery()$ of Rule \rOneOne, $r_{min}$ becomes
 $min$ in at most $4 * n * id_{r_{min}} * \delta$ rounds. Hence, Phase 
 \AmITheMin~is bounded and is in $O(id_{r_{min}} * n * \delta)$.
 
 \item [Phase \MinWaitToBeKnown:]
 
 Now, consider the execution when $r_{min}$ just becomes $min$. Therefore, we 
 consider the execution once $r_{min}$ is \minWaitingWalker. By 
 Corollary~\ref{towerElection}, we know that in finite time a \towerElection~is 
 formed. By Lemma~\ref{uniqueTowerElection}, there is only one \towerElection~in
 the whole execution. Therefore, before a \towerElection~is formed, by the rules 
 of \Gathering~and since initially all the robots are \righter, there are only
 \righter, \potentialMin, \dumbSearcher, \awareSearcher, \minWaitingWalker~and
 \waitingWalker~robots. By Lemma~\ref{waitingCautiousWalkOnSameNode}, we know 
 that all the \minWaitingWalker~and \waitingWalker~robots are located on a same
 node and do not move. By Rule \rTwoThree, if a \potentialMin, a 
 \dumbSearcher~or an \awareSearcher~is located on the same node as a
 \minWaitingWalker, it becomes \waitingWalker~$(*)$. If there is no more 
 \righter~robot in the execution, we use Lemma~\ref{searchBounded} and $(*)$
 multiple times to prove that it takes at most $n * \delta * (\mathcal{R} - 3)$ 
 rounds for a \towerElection~to be formed. To prove that Phase
 \MinWaitToBeKnown~is bounded, we hence have to prove that the number of 
 rounds that are necessary to stop to have \righter~in the execution is bounded.
 
 If a \righter~is located on the same node as the \minWaitingWalker~while there
 is an adjacent right edge to its location, then by Rule \rTwoFour, the 
 \righter~becomes an \awareSearcher~and moves on the right. If a \righter~is 
 located only with $\mathcal{R} - 2$ other \righter, they all execute Rule 
 \rOneSix, hence one becomes \potentialMin~while the others become
 \dumbSearcher. If a \righter~is located either with a \dumbSearcher~or with an
 \awareSearcher, then it becomes an \awareSearcher~(Rule \rOneSeven). Note that,
 by Lemma~\ref{OnlyOnePLeaderOrOneMovingRight}, since we consider the execution 
 once $r_{min}$ is $min$, it cannot exist a \righter~and a \potentialMin~in the
 execution. Therefore, a \righter~cannot meet a \potentialMin. In all the other
 cases, (a \righter~that is isolated, a \righter~that is only with other 
 \righter~on its node such that $|NodeMate()| < \mathcal{R} - 2$, and a 
 \righter~that is located on the same node as the \minWaitingWalker~while there
 is no adjacent right edge to its location) a \righter~stays a \righter~and 
 considers the $right$ direction (Rule \rOneEight). Therefore, by 
 Observation~\ref{noMoreRighter} and since each edge of the footprint of a
 \BRE~ring is present at least once every $\delta$ units of time, it
 takes at most $n * \delta$ rounds in order to stop having \righter~robots in
 the execution. Indeed, even if a \righter~does not execute Rule 
 \rOneSeven~or Rule \rOneSix, at most $n * \delta$ rounds are needed 
 for it to be located on the node where the \minWaitingWalker~is located while 
 there is an adjacent right edge to its position. Hence, Phase 
 \MinWaitToBeKnown~is bounded and is in $O(n * \mathcal{R} * \delta)$.
 
 \end {description}
 
 Once a \towerElection~is present in the execution, the robots forming this
 \towerElection~execute Rule \rTwoOne. While executing this rule, the robot
 $r$ with the maximum identifier among the robots involved in this
 \towerElection~becomes \headWalker~while the \minWaitingWalker~becomes 
 \minTailWalker~and the other robots involved in this \towerElection~become
 \tailWalker. Note that, by Corollary~\ref{uniqueLeader}, only $r_{min}$ can be 
 $min$, and therefore, since $r_{min}$ is the robot with the minimum identifier 
 among all the robots of the system and since at least $2$ robots are involved 
 in a \towerElection, $r_{min}$ cannot become \headWalker. By 
 Lemma~\ref{uniqueTowerElection} and by the rules of \Gathering, only $r$ can 
 be \headWalker~during the execution.
 
 There is no rule in \Gathering~permitting a \tailWalker~or a
 \minTailWalker~robot to become another kind of robot. A \headWalker~can become 
 a \leftWalker. Let then consider the two following cases.
 
 \begin{description}
  \item [Case 1:] \textbf{$\mathbf{r}$ is a \textit{head\-Walker} during the
  whole execution.}
 
 \begin{description}
  \item[Phase \Walk:]
  
  A \headWalker~can execute Rules \rFourTwo, \rFourThree~and \rThreeOne. 
  Since $r$ does not become a \leftWalker, it cannot execute Rule \rFourTwo.
  Moreover, since we consider the worst case 
  execution in terms of time, this implies that $r$ is able to execute Rule
  \rThreeOne~entirely. This means that $r$ is able to execute Rule 
  \rThreeOne~until its variable $walk\-Steps$ reaches the value $n$. In other 
  words, $r$ is able to execute Rule \rThreeOne~until it executes Rule 
  \rFourThree. 
  
  In this case, the \tailWalker~and \minTailWalker~are also able to execute
  Rule \rThreeOne~entirely. Indeed, if, at a time $t'$, 
  while executing Rule \rThreeOne~or Rule \rFourThree, the 
  \headWalker~is waiting on its node for the \tailWalker~and the 
  \minTailWalker~to join it while there is an adjacent left edge to its 
  position, and if at time $t' + 1$ the \tailWalker~and the \minTailWalker~have 
  not join it on its node, this necessarily implies that they stop their 
  execution, otherwise by Rule \rThreeOne~they would have join it. Moreover, 
  if such an event happens, $r$ executes Rule \rFourOne~and therefore
  becomes a \leftWalker, which leads to a contradiction.

  If the \headWalker~and the \minTailWalker/\tailWalker~execute Rule 
  \rThreeOne~entirely, this implies that they move alternatively in the right 
  direction during $n$ steps. Since each edge of the footprint of a 
  \BRE~ring is present at least once every $\delta$ units of time,
  this takes at most $2 * n * \delta$ rounds. Phase \Walk~being
  composed only of the execution of Rule \rThreeOne, this phase is bounded.
 
  \item [Phase \WaitTermination:]
 
  Call $t_{v}$ the time at which the \headWalker~and
  \minTailWalker/\tailWalker~robots finish to execute Rule 
  \rThreeOne~entirely. Since the \headWalker~and 
  \minTailWalker/\tailWalker~start the execution of Rule \rThreeOne~on the
  same node, at time $t_{v}$, they are on the same node $v$.
  
  Call $r_{1}$ and $r_{2}$ the two robots that are not involved in the 
  \towerElection~at time $t_{tower}$. 
  
  If at time $t_{v}$, $r_{1}$ and $r_{2}$ are on node $v$, then Rule 
  \TerminationOne~is executed at time $t_{v}$. In this case, by 
  Lemma~\ref{term}, Phase \WaitTermination~last 0 round, hence it is
  bounded. 
  
  If at time $t_{v}$, only one robot among $r_{1}$ and $r_{2}$ is located on 
  node $v$, then Rule \TerminationTwo~is executed at time $t_{v}$. Hence, by 
  Lemma~\ref{term}, $\mathcal{R} - 2$ robots are terminated on node $v$ at time 
  $t_{v}$. By Lemma~\ref{kindR1AndR2}, at time $t_{tower}$, $r_{1}$ and $r_{2}$
  are such that
  $state_{r_{1}} \in \{$\righter, \potentialMin, \awareSearcher, \dumbSearcher$\}$
  and $state_{r_{2}} \in \{$\awareSearcher, \dumbSearcher$\}$. By the movements 
  of the robots given in the proof of Lemma~\ref{theoremC5}, and since each edge
  of the footprint of a \BRE~ring is present at least once every
  $\delta$ units of time, it takes at most $n * \delta$ rounds for the last 
  robot to reach node $v$. Therefore, it takes at most $n * \delta$ rounds for 
  Rule \TerminationOne~to be executed, and thus, by Lemma~\ref{term}, for 
  all the robots to be terminated on node $v$. Hence, in this case Phase  
  \WaitTermination~last at most $n * \delta$ rounds, therefore it is 
  bounded. 
  
  Now, consider that at time $t_{v}$ neither $r_{1}$ nor $r_{2}$ is located on 
  node $v$. In this
  case, at time $t_{v}$, the \headWalker~and \minTailWalker/\tailWalker~execute
  Rule \rFourThree. While executing Rule \rFourThree, the 
  \headWalker~(resp. \minTailWalker/\tailWalker) stays a \headWalker~(resp.
  \minTailWalker/\tailWalker) and considers the $\bot$ direction. Then, by the 
  rules of \Gathering, they can only execute Rule \rFourThree~until they 
  terminate. Therefore, they remain on node $v$ from time $t_{v}$ until the end
  of their execution. Moreover, as noted previously, by Lemma~\ref{kindR1AndR2}, 
  at time $t_{tower}$, $r_{1}$ and $r_{2}$ are such that
  $state_{r_{1}} \in \{$\righter, \potentialMin, \awareSearcher, \dumbSearcher$\}$
  and $state_{r_{2}} \in \{$\awareSearcher, \dumbSearcher$\}$. By the movements
  of the robots given in the proof of Lemma~\ref{theoremC5}, since each edge of
  the footprint of a \BRE~ring is present at least once every 
  $\delta$ units of time, it takes at most $2 * n * \delta$ rounds for $r_{1}$ 
  and $r_{2}$ to both reach the node $v$ (in case $r_{1}$ and $r_{2}$ meet on an 
  adjacent node of $v$ after at most $n * \delta$ rounds of movements in the 
  same direction). In the case the two robots reach node $v$ at the same time,
  then Rule \TerminationOne~is executed, hence, by Lemma~\ref{term}, all the
  robots terminate at that time. In the case the two robots do not reach node
  $v$ at the same time, then the first one that reaches $v$ permits the 
  execution of Rule \TerminationTwo~(hence, by Lemma~\ref{term}, permits the
  termination of $\mathcal{R} - 2$ robots on node $v$) and the second that 
  reaches $v$ permits the execution of Rule \TerminationOne. Hence, Phase
  \WaitTermination~last at most $2 * n * \delta$ rounds, therefore
  it is bounded. 
 \end {description}
 
 \item[Case 2:] \textbf{It exists a time at which $\mathbf{r}$ is a 
  \textit{left\-Walker}.}
  
  By the explanations given in the Case 1, Phase \Walk, at most 
  $2 * n * \delta$ rounds are needed for $r$ to become \leftWalker~and for the 
  $\mathcal{R} - 2$ other robots to terminate their execution on a node $v$. 
  
  By the rules of \Gathering, a \leftWalker~only executes Rule \rFourOne. 
  While executing this rule, a robot considers the $left$ direction and stays a
  \leftWalker. Since each edge of the footprint of a \BRE~ring is
  present at least once every $\delta$ units of time, such a robot reaches the 
  node $v$ in at most $n * \delta$ rounds. Hence, in this case, Phases 
  \Walk~and \WaitTermination~take at most $3 * n * \delta$ 
  rounds, hence they are bounded.
 \end{description}
 
 By the two previous cases, phase \Walk~and Phase \WaitTermination~take
 $O(n * \delta)$ rounds.
 
 Whatever the \BRE~ring considered, each phase of \Gathering~is 
 bounded, therefore, \Gathering~solves \G~in \BRE~rings. Moreover, 
 \Gathering~solves \G~in \BRE~rings in $O(n * \delta * (id_{r_{min}} + \mathcal{R}))$ rounds.
\end{proof}

Now, we consider the case of \ST~rings. We know that \ST~rings are \BRE~rings 
such that $\delta = 1$, hence, by Lemma~\ref{theoremC7}, we can deduce the following corollary.

\begin{corollary} \label{theoremStatic}
 \Gathering~solves \G~in \ST~rings in $O(n * (id_{r_{min}} + \mathcal{R}))$ rounds.
\end{corollary}

\section{Conclusion}\label{sec:conclu}

In this paper, we apply for the first time the gracefully degrading
approach to robot networks. This approach consists in circumventing
impossibility results in highly dynamic systems by providing algorithms
that adapt themselves to the dynamics of the graph: they solve the
problem under weak dynamics and only guarantee that some weaker
but related problems are satisfied whenever the dynamics increases
and makes the original problem impossible to solve.

Focusing on the classical problem of gathering a squad of autonomous robots,
we introduce a set of weaker variants of this problem that preserves its safety
(in the spirit of the indulgent approach that shares the same underlying idea).
Motivated by a set of impossibility results, we propose a gracefully degrading 
gathering algorithm (refer to Table \ref{tab:summary} for a summary of our results).
We highlight that it is the first gracefully degrading algorithm dedicated to robot 
networks and the first algorithm focusing on the gathering problem 
in \COT, the class of dynamic graphs that exhibits the weakest recurrent connectivity.

A natural open question arises on the \emph{optimality} of the graceful
degradation we propose. Indeed, we prove that our algorithm provides for each
class of dynamic graphs the best specification \emph{among the ones we proposed}.
We do not claim that another algorithm could not be able to satisfy stronger 
variants
of the original gathering specification. 
Aside gathering in robot networks, defining
a general form 
of \emph{degradation optimality} 
seems 
to be a challenging future work.


\newpage
\bibliographystyle{plain}
\bibliography{biblio.bib}

\end{document}